\documentclass{llncs}

\usepackage[table]{xcolor}
\usepackage{amsmath,amssymb}
\usepackage{mathtools}
\usepackage{tikz}
\usetikzlibrary{arrows,automata,calc,backgrounds,decorations.pathreplacing}
\usepackage{booktabs}
\usepackage{fancyhdr}
\usepackage[plainpages=false,pdfpagelabels,colorlinks,linkcolor=blue,citecolor=blue,urlcolor=blue,linktoc=all]{hyperref}
\usepackage[T1]{fontenc}
\usepackage[protrusion=false,expansion=false]{microtype} 
\usepackage{multicol}
\usepackage{enumitem}
\usepackage{subcaption}
\usepackage{pifont}
\usepackage{lscape}
\usepackage{multirow}
\usepackage{multicol}
\usepackage{vwcol}
\usetikzlibrary{positioning}

\mathchardef\shorthyphen="2D

\definecolor{MyCornflowerBlue}{cmyk}{0.58,0.37,0,0.07}

\newcommand{\ExcludeOutInProp}[1]{\textsc{exclude\protect\nobreakdash-out\protect\nobreakdash-in} property#1}
\newcommand{\ExcludeOutIn}[1]{\textsc{exclude\protect\nobreakdash-out\protect\nobreakdash-in}#1}
\newcommand{\SingleLetterProp}[1]{\textsc{single letter} property#1}
\newcommand{\SingleLetter}[1]{\textsc{single letter}#1}
\newcommand{\ConvexityProp}[1]{\textsc{convexity} property#1}

\newcommand{\InflexionFreeProp}[1]{\textsc{no\protect\nobreakdash-inflexion} property#1}
\newcommand{\InflexionFree}[1]{\textsc{no\protect\nobreakdash-inflexion}#1}
\newcommand{\OneInflexionProp}[1]{\textsc{one\protect\nobreakdash-inflexion} property#1}
\newcommand{\OneInflexion}[1]{\textsc{one-inflexion}#1}
\newcommand{\LetterProp}[1]{\textsc{letter} property#1}

\newcommand{\IncompressibleProp}[1]{\textsc{incompressible} property#1}
\newcommand{\Incompressible}[1]{\textsc{incompressible}#1}
\newcommand{\SuffixUnavoidableProp}[1]{\textsc{suffix\protect\nobreakdash-unavoidable}~property#1}
\newcommand{\SuffixUnavoidable}[1]{\textsc{suffix\protect\nobreakdash-unavoidable}#1}
\newcommand{\FactorProp}[1]{\textsc{factor} property#1}
\newcommand{\Factor}[1]{\textsc{factor}#1}

\newcommand{\SumDecompositionProp}[1]{\textsc{sum decomposition} property#1}
\newcommand{\SumDecomposition}[1]{\textsc{sum decomposition}#1}
\newcommand{\SameValueProp}[1]{\textsc{same value} property#1}
\newcommand{\SameValue}[1]{\textsc{same value}#1}
\newcommand{\SinglePositionProp}[1]{\textsc{single position} property#1}
\newcommand{\SinglePosition}[1]{\textsc{single position}#1}

\newcommand{\PositiveProp}[1]{\textsc{positive} property#1}
\newcommand{\Positive}[1]{\textsc{positive}#1}

\newcommand{\SinglePositionInflexionProp}[1]{\textsc{single position inflexion} property#1}
\newcommand{\SinglePositionInflexion}[1]{\textsc{single position inflexion}#1}
\newcommand{\SinglePositionInflexionFreeProp}[1]{\textsc{single position no\protect\nobreakdash-in\-fle\-xion} property#1}
\newcommand{\SinglePositionInflexionFree}[1]{\textsc{single position no\protect\nobreakdash-inflexion}#1}

\DisableLigatures[<,>]{encoding=T1}

\newcommand{\pattern}{\sigma}
\newcommand{\after}{a}
\newcommand{\before}{b}

\newcommand{\MathReg}[1]{\textnormal{`} #1\textnormal{'}}
\newcommand{\reg}[1]{\textnormal{`#1'}}
\newcommand{\word}{w}
\newcommand{\len}[1]{|#1|}
\newcommand{\Alphabet}{\Sigma}
\newcommand{\Language}[1]{\ensuremath{\mathcal{L}_{#1}}}
\newcommand{\mir}{\textnormal{mir}}

\newcommand{\shuffle}{\textnormal{shuffle}}

\newcommand{\seqlength}{n}
\newcommand{\windowsize}{m}
\newcommand{\result}{r}
\newcommand{\seq}{x_1 x_2\dots x_\seqlength}
\newcommand{\seqrev}{x_\seqlength x_{\seqlength-1}\dots x_1}

\newcommand{\sequence}{\mathcal{X}}
\newcommand{\signature}{\mathcal{S}}
\newcommand{\low}{\mathit{low}}
\newcommand{\up}{\mathit{up}}
\newcommand{\feature}{f}
\newcommand{\agg}{g}
\newcommand{\One}{\mathtt{one}}
\newcommand{\Width}{\mathtt{width}}
\newcommand{\Surf}{\mathtt{surf}}
\newcommand{\MaxFeature}{\mathtt{max}}
\newcommand{\MinFeature}{\mathtt{min}}

\newcommand{\SumAggr}{\mathtt{Sum}}

\newcommand{\tw}{[i,j]}

\newcommand{\constraint}[1]{\textsc{#1}}

\newcommand{\Tuple}[1]{\left \langle#1 \right \rangle}  

\newcommand{\BumpOnDecreasingSequencePatternName}{\mathtt{BumpOnDecSeq}}
\newcommand{\DipOnIncreasingSequencePatternName}{\mathtt{DipOnIncSeq}}
\newcommand{\DecreasingPatternName}{\mathtt{Dec}}
\newcommand{\IncreasingPatternName}{\mathtt{Inc}}
\newcommand{\DecreasingSequencePatternName}{\mathtt{DecSeq}}
\newcommand{\IncreasingSequencePatternName}{\mathtt{IncSeq}}
\newcommand{\DecreasingTerracePatternName}{\mathtt{DecTerrace}}
\newcommand{\IncreasingTerracePatternName}{\mathtt{IncTerrace}}
\newcommand{\GorgePatternName}{\mathtt{Gorge}}
\newcommand{\InflexionPatternName}{\mathtt{Inflexion}}
\newcommand{\SummitPatternName}{\mathtt{Summit}}
\newcommand{\PeakPatternName}{\mathtt{Peak}}
\newcommand{\ValleyPatternName}{\mathtt{Valley}}
\newcommand{\PlainPatternName}{\mathtt{Plain}}
\newcommand{\PlateauPatternName}{\mathtt{Plateau}}
\newcommand{\ProperPlainPatternName}{\mathtt{ProperPlain}}
\newcommand{\ProperPlateauPatternName}{\mathtt{ProperPlateau}}
\newcommand{\SteadyPatternName}{\mathtt{Steady}}
\newcommand{\SteadySequencePatternName}{\mathtt{SteadySeq}}
\newcommand{\StrictlyDecreasingSequencePatternName}{\mathtt{StrictlyDecSeq}}
\newcommand{\StrictlyIncreasingSequencePatternName}{\mathtt{StrictlyIncSeq}}
\newcommand{\ZigzagPatternName}{\mathtt{Zigzag}}

\newcommand{\BumpOnDecreasingSequencePatternU}{>><>>}
\newcommand{\DipOnIncreasingSequencePatternU}{<<><<}
\newcommand{\DecreasingPatternU}{>}
\newcommand{\DecreasingSequencePatternU}{>(>|=)^*>|>}
\newcommand{\DecreasingTerracePatternU}{>=^+>}
\newcommand{\GorgePatternU}{(>(>|=)^*)^*><((<|=)^*<)^*}
\newcommand{\IncreasingPatternU}{<}
\newcommand{\IncreasingSequencePatternU}{<(<|=)^*<|<}
\newcommand{\IncreasingTerracePatternU}{<=^+<}
\newcommand{\InflexionPatternU}{<(<|=)^*>~|~>(>|=)^*<}
\newcommand{\PeakPatternU}{<(<|=)^*(>|=)^*>}
\newcommand{\PlainPatternU}{>=^*<}
\newcommand{\PlateauPatternU}{<=^*>}
\newcommand{\ProperPlainPatternU}{>=^+<}
\newcommand{\ProperPlateauPatternU}{<=^+>}
\newcommand{\SteadyPatternU}{=}
\newcommand{\SteadySequencePatternU}{=^+}
\newcommand{\StrictlyDecreasingSequencePatternU}{>^+}
\newcommand{\StrictlyIncreasingSequencePatternU}{<^+}
\newcommand{\SummitPatternU}{(<(<|=)^*)^*<>((>|=)^*>)^*}
\newcommand{\ValleyPatternU}{>(>|=)^*(<|=)^*<}
\newcommand{\ZigzagPatternU}{(<>)^+<(>|\epsilon)~|~(><)^+>(<|\epsilon)}

\newcommand{\DecreasingSequencePattern}{\textnormal{`}\DecreasingSequencePatternU\textnormal{'}}

\newcommand{\IncreasingSequencePattern}{\textnormal{`}\IncreasingSequencePatternU\textnormal{'}}

\newcommand{\InflexionPattern}{\textnormal{`}\InflexionPatternU\textnormal{'}}

\newcommand{\PlateauPattern}{\textnormal{`}\PlateauPatternU\textnormal{'}}

\newcommand{\featurePre}{\feature_\pattern(x_{1,j})}
\newcommand{\featureSuf}{\feature_{\pattern^r}(x_{n,i})}
\newcommand{\featureTw}{\feature_\pattern(x_{i,j})}
\newcommand{\featureTotal}{\feature_\pattern(x_{1,n})}

\newcommand{\lPre}{t=\ell+b_\pattern}
\newcommand{\uPre}{\finf-a_\pattern}
\newcommand{\lSuf}{t=\startf+a_{\pattern^r}}
\newcommand{\uSuf}{u-b_{\pattern^r}}
\newcommand{\lTotal}{t=\ell+b_\pattern}
\newcommand{\uTotal}{u-a_\pattern}
\newcommand{\lTw}{t=\startf+b_\pattern}
\newcommand{\uTw}{\finf-a_\pattern}

\newcommand{\ffac}{\mathrm{f}}
\newcommand{\fin}{\mathrm{i}}
\newcommand{\fout}{\mathrm{o}}
\newcommand{\fsuf}{\mathrm{s}}
\newcommand{\fpre}{\mathrm{p}}

\newcommand{\cfac}{\mathrm{fac}}
\newcommand{\cin}{\mathrm{in}}
\newcommand{\cout}{\mathrm{out}}
\newcommand{\csuf}{\mathrm{suf}}
\newcommand{\cpre}{\mathrm{pre}}

\newcommand{\gfac}{\textsc{fac}}
\newcommand{\gin}{\textsc{in}}
\newcommand{\gout}{\textsc{out}}
\newcommand{\gsuf}{\textsc{suf}}
\newcommand{\gpre}{\textsc{pre}}
\newcommand{\gps}{\textsc{ps}}

\newcommand{\afac}{\textsc{f}}
\newcommand{\ain}{\textsc{i}}
\newcommand{\aout}{\textsc{o}}
\newcommand{\asuf}{\textsc{s}}
\newcommand{\apre}{\textsc{p}}

\newcommand{\starto}{\beta}
\newcommand{\fino}{\alpha}
\newcommand{\startf}{\lambda}
\newcommand{\finf}{\psi}

\begin{document}

\author{N. Beldiceanu$^{1}$, M. Carlsson$^{2}$, C.-G. Quimper$^{3}$ \and M. I.~Restrepo$^{1}$}

\title{Classifying Pattern and Feature Properties\\ 
to Get a $\Theta(n)$ Checker and Reformulation\\ 
for Sliding Time-Series Constraints}

\institute{
  $^{1}$TASC (LS2N-CNRS), IMT Atlantique, FR -- 44307 Nantes, France \\
  $^{2}$RISE SICS, Sweden \\
  $^{3}$Laval University, Qu\'ebec, Canada 
}
\maketitle

\begin{abstract}
Given,
a sequence $\sequence$ of $\seqlength$ variables,
a time\nobreakdash-series constraint \constraint{ctr} using the $\SumAggr$ aggregator, and
a sliding time\nobreakdash-series constraint enforcing the constraint \constraint{ctr}
on each sliding window of $\sequence$ of $\windowsize$ consecutive variables,
we describe
a $\Theta(\seqlength)$ time complexity checker, as well as
a $\Theta(\seqlength)$ space complexity reformulation
for such sliding constraint.
\end{abstract}

\section{Introduction}

While sequence constraints on sliding windows were introduced a long time ago for counting and for sum constraints,
e.g. see~\constraint{among\_seq} in~\cite{BeldiceanuContejean94,ReginPuget97,SLIDE} and \constraint{sliding\_sum} in~\cite{Smooth,MaherNarodytskaQuimperWalsh08},
no sliding automaton constraint was yet introduced, even if automaton constraints were known since 2004~\cite{Beldiceanu:automata,Pesant:seqs}.
More recently in the context of planning problems,
constraints on streams were introduced in~\cite{LallouetLawLeeSiu11,LeeLeeZhong18}
for comparing pointwise two stream variables or for stating constraints adapted from
Linear Temporal Logic.
However, in the context of a long sequence or of a data stream~\cite{AlurFismanRaghothaman16},
imposing a constraint on a full sequence does not make much sense,
as we rather want to focus on sliding windows.
Compositional time\nobreakdash-series constraints combining
a regular expression $\pattern$,
a feature $\feature$,
and an aggregator $\agg$ were introduced in~\cite{Beldiceanu:synthesis,Catalog16}.
We first provide an example of sliding time series constraint.

\begin{example}\label{ex:sliding_time_series_ctr_example}
Given a sequence $\sequence=3\hspace*{1pt}1\hspace*{1pt}3\hspace*{1pt}3\hspace*{1pt}2\hspace*{1pt}1\hspace*{1pt}1\hspace*{1pt}2\hspace*{1pt}2\hspace*{1pt}2\hspace*{1pt}4\hspace*{1pt}4\hspace*{1pt}3\hspace*{1pt}1\hspace*{1pt}2\hspace*{1pt}2$, we want to compute the sum of subsequences of $\sequence$ corresponding to increasing sequences, i.e.~to maximal occurrences of the pattern
$\IncreasingSequencePattern$, in every window of size $10$ of $\sequence$.
Such \hspace*{1pt}windows \hspace*{1pt}are \hspace*{1pt}shown \hspace*{1pt}in \hspace*{1pt}the \hspace*{1pt}figure \hspace*{1pt}on \hspace*{1pt}the \hspace*{1pt}right \hspace*{1pt}by \hspace*{1pt}a \hspace*{1pt}dotted \hspace*{1pt}line,
\end{example}
\vspace{-20pt}
\begin{multicols}{2}
\noindent where each solid line-segment indicates an increasing sequence.
The number to the left of each window is the sum of the elements
of the window belonging to an increasing sequence located inside the window.
Beyond this example we want a generic approach to deal with a variety of patterns and features.

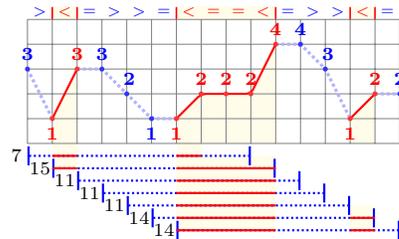
\begin{tikzpicture}[scale=0.33]
\node (s1) at ($(1,5)+(0.5,0.3)$) {\tiny{\color{blue}$>$}};
\fill[yellow!10] (2,5) rectangle (3,5.5);
\node (s2) at ($(2,5)+(0.5,0.3)$) {\tiny{\color{red}$<$}};
\node (s3) at ($(3,5)+(0.5,0.3)$) {\tiny{\color{blue}$=$}};
\node (s4) at ($(4,5)+(0.5,0.3)$) {\tiny{\color{blue}$>$}};
\node (s5) at ($(5,5)+(0.5,0.3)$) {\tiny{\color{blue}$>$}};
\node (s6) at ($(6,5)+(0.5,0.3)$) {\tiny{\color{blue}$=$}};
\fill[yellow!10] (7,5) rectangle (8,5.5);
\node (s7) at ($(7,5)+(0.5,0.3)$) {\tiny{\color{red}$<$}};
\fill[yellow!10] (8,5) rectangle (9,5.5);
\node (s8) at ($(8,5)+(0.5,0.3)$) {\tiny{\color{red}$=$}};
\fill[yellow!10] (9,5) rectangle (10,5.5);
\node (s9) at ($(9,5)+(0.5,0.3)$) {\tiny{\color{red}$=$}};
\fill[yellow!10] (10,5) rectangle (11,5.5);
\node (s10) at ($(10,5)+(0.5,0.3)$) {\tiny{\color{red}$<$}};
\node (s11) at ($(11,5)+(0.5,0.3)$) {\tiny{\color{blue}$=$}};
\node (s12) at ($(12,5)+(0.5,0.3)$) {\tiny{\color{blue}$>$}};
\node (s13) at ($(13,5)+(0.5,0.3)$) {\tiny{\color{blue}$>$}};
\fill[yellow!10] (14,5) rectangle (15,5.5);
\node (s14) at ($(14,5)+(0.5,0.3)$) {\tiny{\color{red}$<$}};
\node (s15) at ($(15,5)+(0.5,0.3)$) {\tiny{\color{blue}$=$}};
\draw[color=red,semithick] (2,5) -- (2,5.5);
\draw[color=red,semithick] (3,5) -- (3,5.5);
\draw[color=red,semithick] (7,5) -- (7,5.5);
\draw[color=red,semithick] (11,5) -- (11,5.5);
\draw[color=red,semithick] (14,5) -- (14,5.5);
\draw[color=red,semithick] (15,5) -- (15,5.5);
\fill[fill=yellow!10] (2,1) -- (3,3) -- (3,0) -- (2,0);
\fill[fill=yellow!10] (7,1) -- (8,2) -- (9,2) -- (10,2) -- (11,4) -- (11,0)--(7,0);
\fill[fill=yellow!10] (14,1) -- (15,2) -- (15,0)--(14,0);
\fill[yellow!10] (2,-1) rectangle (3,0);
\fill[yellow!10] (7,-0.5) rectangle (8,0);
\fill[yellow!10] (7,-3.5) rectangle (11,-0.5);
\fill[yellow!10] (14,-3.5) rectangle (15,-2.5);
\draw[step=1cm,gray,very thin] (1,0) grid (16,5);
\draw[draw=blue!30,line width=1.2pt,rounded corners=2pt,densely dotted] (1,3) -- (2,1);
\draw[draw=red,thick,rounded corners=2pt] (2,1) -- (3,3);
\draw[draw=blue!30,line width=1.2pt,rounded corners=2pt,densely dotted] (3,3) -- (4,3) -- (5,2) -- (6,1) -- (7,1);
\draw[draw=red,thick,rounded corners=2pt] (7,1) -- (8,2) -- (9,2) -- (10,2) -- (11,4);
\draw[draw=blue!30,line width=1.2pt,rounded corners=2pt,densely dotted] (11,4) -- (12,4) -- (13,3) -- (14,1);
\draw[draw=red,thick,rounded corners=2pt] (14,1) -- (15,2);
\draw[draw=blue!30,line width=1.2pt,rounded corners=2pt,densely dotted] (15,2) -- (16,2);
\coordinate (c1) at (1,3); \fill[color=blue!80] (c1) circle (0.1);
\node at ($(c1)+(0pt,0)$) [above] {\scriptsize{\color{blue}$\mathbf{3}$}};
\coordinate (c2) at (2,1); \fill[color=red!80] (c2) circle (0.1);
\node at ($(c2)+(0pt,0)$) [below] {\scriptsize{\color{red}$\mathbf{1}$}};
\coordinate (c3) at (3,3); \fill[color=red!80] (c3) circle (0.1);
\node at ($(c3)+(0pt,0)$) [above] {\scriptsize{\color{red}$\mathbf{3}$}};
\coordinate (c4) at (4,3); \fill[color=blue!80] (c4) circle (0.1);
\node at ($(c4)+(0pt,0)$) [above] {\scriptsize{\color{blue}$\mathbf{3}$}};
\coordinate (c5) at (5,2); \fill[color=blue!80] (c5) circle (0.1);
\node at ($(c5)+(3pt,0)$) [above] {\scriptsize{\color{blue}$\mathbf{2}$}};
\coordinate (c6) at (6,1); \fill[color=blue!80] (c6) circle (0.1);
\node at ($(c6)+(0pt,0)$) [below] {\scriptsize{\color{blue}$\mathbf{1}$}};
\coordinate (c7) at (7,1); \fill[color=red!80] (c7) circle (0.1);
\node at ($(c7)+(0pt,0)$) [below] {\scriptsize{\color{red}$\mathbf{1}$}};
\coordinate (c8) at (8,2); \fill[color=red!80] (c8) circle (0.1);
\node at ($(c8)+(0pt,0)$) [above] {\scriptsize{\color{red}$\mathbf{2}$}};
\coordinate (c9) at (9,2); \fill[color=red!80] (c9) circle (0.1);
\node at ($(c9)+(0pt,0)$) [above] {\scriptsize{\color{red}$\mathbf{2}$}};
\coordinate (c10) at (10,2); \fill[color=red!80] (c10) circle (0.1);
\node at ($(c10)+(0pt,0)$) [above] {\scriptsize{\color{red}$\mathbf{2}$}};
\coordinate (c11) at (11,4); \fill[color=red!80] (c11) circle (0.1);
\node at ($(c11)+(0pt,0)$) [above] {\scriptsize{\color{red}$\mathbf{4}$}};
\coordinate (c12) at (12,4); \fill[color=blue!80] (c12) circle (0.1);
\node at ($(c12)+(0pt,0)$) [above] {\scriptsize{\color{blue}$\mathbf{4}$}};
\coordinate (c13) at (13,3); \fill[color=blue!80] (c13) circle (0.1);
\node at ($(c13)+(3pt,0)$) [above] {\scriptsize{\color{blue}$\mathbf{3}$}};
\coordinate (c14) at (14,1); \fill[color=red!80] (c14) circle (0.1);
\node at ($(c14)+(0pt,0)$) [below] {\scriptsize{\color{red}$\mathbf{1}$}};
\coordinate (c15) at (15,2); \fill[color=red!80] (c15) circle (0.1);
\node at ($(c15)+(0pt,0)$) [above] {\scriptsize{\color{red}$\mathbf{2}$}};
\coordinate (c16) at (16,2); \fill[color=blue!80] (c16) circle (0.1);
\node at ($(c16)+(0pt,0)$) [above] {\scriptsize{\color{blue}$\mathbf{2}$}};
\draw[color=blue,line width=0.8pt,|-|,densely dotted] (1,-0.5) -- (10,-0.5) node[pos=-0.05] {\scriptsize{\color{black}$7$}};
\draw[color=red,thick] (2,-0.5) -- (3,-0.5);
\draw[color=red,thick] (7,-0.5) -- (8,-0.5);
\draw[color=blue,line width=0.8pt,|-|,densely dotted] (2,-1.0) -- (11,-1.0) node[pos=-0.055] {\scriptsize{\color{black}$15$}};
\draw[color=red,thick] (2,-1.0) -- (3,-1.0);
\draw[color=red,thick] (7,-1.0) -- (11,-1.0);
\draw[color=blue,line width=0.8pt,|-|,densely dotted] (3,-1.5) -- (12,-1.5) node[pos=-0.055] {\scriptsize{\color{black}$11$}};
\draw[color=red,thick] (7,-1.5) -- (11,-1.5);
\draw[color=blue,line width=0.8pt,|-|,densely dotted] (4,-2.0) -- (13,-2.0) node[pos=-0.055] {\scriptsize{\color{black}$11$}};
\draw[color=red,thick] (7,-2.0) -- (11,-2.0);
\draw[color=blue,line width=0.8pt,|-|,densely dotted] (5,-2.5) -- (14,-2.5) node[pos=-0.055] {\scriptsize{\color{black}$11$}};
\draw[color=red,thick] (7,-2.5) -- (11,-2.5);
\draw[color=blue,line width=0.8pt,|-|,densely dotted] (6,-3.0) -- (15,-3.0) node[pos=-0.055] {\scriptsize{\color{black}$14$}};
\draw[color=red,thick] (7,-3.0) -- (11,-3.0);
\draw[color=red,thick] (14,-3.0) -- (15,-3.0);
\draw[color=blue,line width=0.8pt,|-|,densely dotted] (7,-3.5) -- (16,-3.5) node[pos=-0.055] {\scriptsize{\color{black}$14$}};
\draw[color=red,thick] (7,-3.5) -- (11,-3.5);
\draw[color=red,thick] (14,-3.5) -- (15,-3.5);
\end{tikzpicture}
\end{multicols}

\paragraph{Contributions and methodology}
Our contributions are threefold.
\begin{itemize}
\item
By pursuing the compositional style for defining time\nobreakdash-series constraints~\cite{Beldiceanu:synthesis},
we introduce sliding time\nobreakdash-series constraints, assuming $\agg$ is the $\SumAggr$ aggregator.
This allows one to define a fair variety of sliding constraints in a generic way,
in fact~$99$ constraints in the time\nobreakdash-series catalogue~\cite{Catalog16}.
\item
It provides a $\Theta(\seqlength)$ linear time complexity checker for such constraints,
which is crucial when extracting patterns from long sequences
in the context of model acquisition~\cite{PicardCantinBouchardQuimperSweeney16,Vandrager17}.
\item
It describes a $\Theta(\seqlength)$ linear space complexity reformulation,
which allows a memory efficient reformulation.
\end{itemize}
To obtain our contributions we use the following methodology.
\begin{itemize}
\item
We come up with three simple equations allowing one to compute
the contribution of a window $[i,j]$ (with $i\leq j$) wrt the results
(a)~on the full sequence $\sequence=\seq$,
(b)~on the prefix $x_1 x_2\dots x_j$ which ends at position $j$, and
(c)~on the suffix $x_i x_{i+1}\dots x_\seqlength$ which starts at position $i$.
\item
We study both the properties of regular expressions and features:
\begin{itemize}
\item
We systematically categorise regular expressions by
partitioning their words into a restricted set of classes,
so that each regular expression can be compactly represented
by a finite set of classes.
\item
We identify key pattern and feature properties.
\end{itemize}
For each pair of word classes and feature properties, we prove that a given equation holds
or provide some counterexample.
\item
Finally, we show how equations can be directly turned
into checkers and reformulations. 
\end{itemize}
The categorisation of a regular expression and the identification of the properties
of a pattern are done mechanically by checking that some derived regular languages
are empty or not.

Section~\ref{sec:background} provides the necessary background on words and time-series constraints.
Section~\ref{sec:pattern_prop} introduces a small number of pattern properties, while
Section~\ref{sec:sliding_time_series_constraints}
\emph{(i)}~defines the sliding time-series constraints we consider,
\emph{(ii)}~classifies regular expressions in relation to sliding windows,
\emph{(iii)}~shows how to compute the contribution of a sliding window based on pattern and feature properties, and finally
\emph{(iv)}~presents a $\Theta(\seqlength)$ time complexity checker and a $\Theta(\seqlength)$
space complexity reformulation for such sliding time-series constraints.

\section{Background}\label{sec:background}

\paragraph{Word}
Consider a finite alphabet $\Alphabet$.
A \emph{word} $\word$ over $\Alphabet$
is a sequence of letters $\word_1\word_2\dots\word_\ell$ of the alphabet $\Alphabet$,
and its \emph{length} $\ell$ is denoted by $\len{\word}$.
The \emph{empty word} is denoted by $\epsilon$.
The \emph{reverse} of $\word$ is
the word $\word_\ell\word_{\ell-1}\dots\word_1$ denoted $\word^r$.
The concatenation of two words is denoted by putting them side by side.
A word $v$ is a \emph{factor} of a word $x$ if there exists two words $u$ and $w$ such that $x=uvw$;
when $u=\epsilon$, $v$ is a \emph{prefix} of $x$,
when $w=\epsilon$, $v$ is a \emph{suffix} of $x$.
If $v$ is not empty and different from $x$, then $v$ is a \emph{proper factor} of $x$.

\paragraph{Time-series constraints}
We assume the reader is familiar with regular expressions and automata~\cite{Hopcroft:automata}.
A time-series constraint $\agg\_\feature\_\pattern(\result,\sequence)$ is a constraint which restricts an integer result 
variable $\result$ to be the result of some computations over a sequence of integer variables $\sequence$.
The components of a time-series constraint we reuse from \cite{Beldiceanu:synthesis} are
a \emph{pattern} $\pattern$, a \emph{feature} $\feature$, and an \emph{aggregator} $\agg$.
A pattern $\pattern$ is described by a \emph{regular expression} over the alphabet
$\Alphabet=\{\MathReg{<},\MathReg{=},\MathReg{>}\}$ whose language $\Language{\pattern}$
does not contain the empty word,
and by two non-negative integers $b_\pattern$ and $a_\pattern$, where
$b_\pattern+a_\pattern$ is smaller than or equal to the length of the smallest word of $\Language{\pattern}$.
A feature and an aggregator are functions over integer sequences
as illustrated in~Table~\ref{tab:fgp}. 
Note that all functions $\feature$ and $\agg$ introduced in~Table~\ref{tab:fgp} are commutative. 
Let $\signature = s_1 s_2\dots s_{\seqlength-1}$ be the \emph{signature}
of a time series $\sequence$, which is defined by constraints:
$(x_i<x_{i+1} \Leftrightarrow s_i = \textnormal{`}<\textnormal{'}) \land (x_i=x_{i+1} \Leftrightarrow s_i = \textnormal{`}=\textnormal{'})\land (x_i>x_{i+1} \Leftrightarrow s_i = \textnormal{`}>\textnormal{'})$ for all $i\in [1,\seqlength-1]$.
If a sub\nobreakdash-signature $s_i s_{i+1}\dots s_{j-1}$
is a maximal word matching $\pattern$ in the signature of $\sequence$,
then the subsequence $x_{i+b_\pattern}x_{i+b_\pattern+1}\dots x_{j-a_\pattern}$
is called a \emph{$\pattern$-pattern} wrt~$\sequence$, and
the subsequence $x_{i}x_{i+1}\dots x_{j}$ is called
an \emph{extended $\pattern$-pattern} wrt~$\sequence$.
The non-negative integers $b_\pattern$ and $a_\pattern$ 
trim the left and right borders of an extended
$\pattern$-pattern to obtain a $\pattern$-pattern from which a feature value is computed.

{
\begin{table}[!b]
  \scriptsize
  \setlength{\tabcolsep}{1pt}
  \centering
  \begin{subtable}[b]{0.30\textwidth}
    \centering
    \begin{tabular}{lc}
      $f$                  & $\textnormal{value}$                                              \\ \midrule
      $\One$               & $1$                                                                     \\
      $\Width$             & $j-i-b_\pattern-a_\pattern+1$                                                       \\
      $\Surf$           & $\sum \limits_{k=i+b_\pattern}^{j-a_\pattern} x_k$                                             \\
      $\MaxFeature$        & $\max \limits_{k \in [i+b_\pattern,j-a_\pattern]} x_k$                                  \\
      $\MinFeature$        & $\min \limits_{k \in [i+b_\pattern,j-a_\pattern]} x_k$                                  \\ \\ \\ \\ \\ \\ \\ \\ \\ \\ \\ \\ \vspace{3.5pt}
      $\agg$               & $\textnormal{value}$       \\ \midrule
      $\SumAggr$           & $\sum \limits_{k=1}^c f_k$ \\
    \end{tabular}
  \end{subtable}
  \begin{subtable}[b]{0.69\textwidth}
\setlength{\tabcolsep}{1pt}
    \centering
    \begin{tabular}{lcccccccc}
      $\pattern$ & $\Language{\sigma}$ & $\before_\sigma$ & $\after_\sigma$ & {\sc r} & {\sc n} & {\sc o} & {\sc e} & {\sc s} \\ \midrule
      $\InflexionPatternName$                  & $\InflexionPatternU$                  & $1$ & $1$
                                                                                       & n & n & y & n & n \\
                                                                                \specialrule{0.2pt}{0pt}{0pt}
      $\BumpOnDecreasingSequencePatternName$   & $\BumpOnDecreasingSequencePatternU$   & $2$ & $1$
                                                                                       & n & n & n & n & n \\
      $\DipOnIncreasingSequencePatternName$    & $\DipOnIncreasingSequencePatternU$    & $2$ & $1$
                                                                                       & n & n & n & n & n \\
                                                                                \specialrule{0.2pt}{0pt}{0pt}
      $\DecreasingPatternName$                 & $\DecreasingPatternU$                 & $0$ & $0$
                                                                                       & y & y & n & y & y \\
      $\IncreasingPatternName$                 & $\IncreasingPatternU$                 & $0$ & $0$
                                                                                       & y & y & n & y & y \\
      $\SteadyPatternName$                     & $\SteadyPatternU$                     & $0$ & $0$
                                                                                       & y & y & n & y & y \\
                                                                                \specialrule{0.2pt}{0pt}{0pt}
      $\DecreasingTerracePatternName$          & $\DecreasingTerracePatternU$          & $1$ & $1$
                                                                                       & y & y & n & n & n \\
      $\IncreasingTerracePatternName$          & $\IncreasingTerracePatternU$          & $1$ & $1$
                                                                                       & y & y & n & n & n \\
                                                                                \specialrule{0.2pt}{0pt}{0pt}
      $\PlainPatternName$                      & $\PlainPatternU$                      & $1$ & $1$
                                                                                       & y & n & y & n & n \\
      $\PlateauPatternName$                    & $\PlateauPatternU$                    & $1$ & $1$
                                                                                       & y & n & y & n & n \\
      $\ProperPlainPatternName$                & $\ProperPlainPatternU$                & $1$ & $1$
                                                                                       & y & n & y & n & n \\
      $\ProperPlateauPatternName$              & $\ProperPlateauPatternU$              & $1$ & $1$
                                                                                       & y & n & y & n & n \\
                                                                                \specialrule{0.2pt}{0pt}{0pt}
      $\GorgePatternName$                      & $\GorgePatternU$                      & $1$ & $1$
                                                                                       & y & n & y & n & n \\
      $\SummitPatternName$                     & $\SummitPatternU$                     & $1$ & $1$
                                                                                       & y & n & y & n & n \\
      $\PeakPatternName$                       & $\PeakPatternU$                       & $1$ & $1$
                                                                                       & y & n & y & n & n \\
      $\ValleyPatternName$                     & $\ValleyPatternU$                     & $1$ & $1$
                                                                                       & y & n & y & n & n \\
                                                                                \specialrule{0.2pt}{0pt}{0pt}
      $\DecreasingSequencePatternName$         & $\DecreasingSequencePatternU$         & $0$ & $0$
                                                                                       & y & y & n & y & n \\
      $\IncreasingSequencePatternName$         & $\IncreasingSequencePatternU$         & $0$ & $0$
                                                                                       & y & y & n & y & n \\
      $\SteadySequencePatternName$             & $\SteadySequencePatternU$             & $0$ & $0$
                                                                                       & y & y & n & y & n \\
      $\StrictlyDecreasingSequencePatternName$ & $\StrictlyDecreasingSequencePatternU$ & $0$ & $0$
                                                                                       & y & y & n & y & n \\
      $\StrictlyIncreasingSequencePatternName$ & $\StrictlyIncreasingSequencePatternU$ & $0$ & $0$
                                                                                       & y & y & n & y & n \\
                                                                                \specialrule{0.2pt}{0pt}{0pt}
      $\ZigzagPatternName$                     & $\ZigzagPatternU$                     & $1$ & $1$
                                                                                       & y & n & n & n & n \\
                                                                                \specialrule{0.2pt}{0pt}{0pt}
    \end{tabular}
  \end{subtable}
  \caption{\label{tab:fgp} Consider a sequence $\seq$.
    (Top left)~features $\feature$ with their values computed from an
    extended $\pattern$-pattern $x_i x_{i+1}\dots x_j$;
    (Bottom left)~aggregator $\agg=\SumAggr$, its value
    computed from a sequence of feature values $f_1,f_2,\dots,f_c$;
    (Right)~patterns $\pattern=\langle\Language{\sigma},\before_\sigma,\after_\sigma\rangle$
    grouped by the properties they share, where columns
    {\sc r}, {\sc n}, {\sc o}, {\sc e}, {\sc s} respectively indicate whether a pattern
    has a reverse in the catalogue~\cite{arafailova2016global},
    the \InflexionFree{},
    the \OneInflexion{},
    the \ExcludeOutIn{}, or
    the \SingleLetter{} properties.}
\end{table}

In the following $x_{i,j}$ denotes the integer subsequence $x_i x_{i+1}\dots x_j$ 
when $i\leq j$ and $x_i x_{i-1}\dots x_j$ otherwise. 
The term $\feature_\pattern(x_{i,j})$ denotes the sum of the values of the
feature $\feature$ from every extended $\pattern$\nobreakdash-pattern
in subsequence $x_{i,j}$, i.e.~the contribution of the sliding window $\tw$.

\section{Pattern Properties}\label{sec:pattern_prop}

We introduce a limited number of pattern properties\footnote{Through an abuse of language
and for reasons of brevity we say ``pattern property of $\pattern$'' rather than
``property of the language $\Language{\pattern}$ of the pattern $\pattern$''.}
that will be used to parameterise our proofs:
we will assume that some of these properties hold to prove that a given equation is valid
for calculating the contribution of a sliding window.

\begin{definition}
\label{def:miror}
The \emph{mirror of a regular language} $\mathcal{L}$ over
$\Alphabet =\{\MathReg{<},\MathReg{=},\MathReg{>}\}$, 
denoted by $\mathcal{L}^{\mir}$, consists of the mirrors of all the words in $\mathcal{L}$, 
where the \emph{mirror of a word} $w$, denoted by $w^{\mir}$, has the reverse order of its letters
and has all occurrences of the letter $\reg{<}$ flipped into $\reg{>}$ and vice versa.
\end{definition}

\begin{definition}
\label{def:reversible}
Two patterns $\pattern=\langle \Language{\pattern} , b_\pattern, a_\pattern\rangle$ and $\pattern^r=\langle\Language{\pattern^r} , b_{\pattern^r}, a_{\pattern^r}\rangle$ are the \emph{reverse} of each other iff
$\word\in\Language{\pattern}\Leftrightarrow~ \word^{\mir}\in\Language{\pattern^r}$,
$a_\pattern=b_{\pattern^r}~\text{and}~~b_\pattern=a_{\pattern^r}$.
\end{definition}
As shown by column {\sc r} of the pattern part of Table~\ref{tab:fgp},
$19$ out of the $22$ patterns of the time-series catalogue~\cite{arafailova2016global}
have a reverse pattern defined inside~\cite{arafailova2016global}.

\begin{example}[reverse]
On the one hand,
the $\PlateauPatternName=\langle\PlateauPattern,1,1\rangle$ pattern is the reverse of itself since,
(1)~all letters except the first and last letters of a plateau correspond to the letter $\reg{=}$, 
(2)~the first letter $\reg{<}$ is the mirror of the last letter $\reg{>}$, and 
(3)~$a_\PlateauPatternName=b_{\PlateauPatternName}=1$.
On the other hand,
the $\InflexionPatternName=\langle\InflexionPattern,1,1\rangle$ pattern is not the reverse of itself:
the mirror of the word $\reg{<<>}\in\Language{\InflexionPatternName}$,
i.e.~the word $\reg{<>>}$, is not an inflexion since it ends with two occurrences
of $\reg{>}$ rather than one.
\end{example}

\begin{sloppypar}
\begin{definition}
\label{def:convexity_regexp}
A pattern $\pattern$ has the \ConvexityProp{} if
for any word $\word=s_1 s_2 \dots s_{\seqlength-1}$ in $\Language{\pattern}$ and for any pair of factors
$u=s_c s_{c+1} \dots s_d$ and $v=s_e s_{e+1} \dots s_f$ of $\word$ (with $c,d,e,f\in[1,\seqlength-1]$)
such that,
both $u$ and $v$ are words in $\Language{\pattern}$,
the word $s_{\min(c,e)} s_{\min(c,e)+1} \dots s_{\max(d,f)}$ is also in $\Language{\pattern}$.
\end{definition}
\end{sloppypar}

\begin{sloppypar}
\begin{example}[\ConvexityProp{}]
All patterns of the time series catalogue~\cite{arafailova2016global}
have the \ConvexityProp{},
but the pattern whose language is denoted by $\Language{<=>=|<=|=>}$ has not,
since the word $\reg{<=>=}$ in $\Language{<=>=|<=|=>}$ contains
a factor $\reg{<=>}$ that is not in $\Language{<=>=|<=|=>}$, for which 
both the prefix $\reg{<=}$ and the suffix $\reg{=>}$ belong to $\Language{<=>=|<=|=>}$.
\end{example}
\end{sloppypar}

\begin{definition}
\label{def:inflexion_free}
A pattern $\pattern$ has the \InflexionFreeProp{} if any word in its language $\Language{\pattern}$
does not simultaneously contain the letters \reg{<} and \reg{>}.
\end{definition}

\begin{definition}
\label{def:one_inflexion}
A pattern $\pattern$ has the \OneInflexionProp{} if any word in its language $\Language{\pattern}$ 
contains either one, but not both occurrences of $\reg{<=*>}$ and $\reg{>=*<}$.
\end{definition}

\begin{definition}
\label{def:single_letter}
A pattern $\pattern$ has the \SingleLetterProp{} 
if all words of $\Language{\pattern}$ have a length of one.
\end{definition}

\begin{definition}
\label{def:exclude_out_in}
A pattern $\pattern$ has the \ExcludeOutInProp{} if
for any word $s_1 s_2 \dots s_{\seqlength-1}$ in $\Language{\pattern}$ and for any window $[i,j]$
(with $1\leq i\leq j\leq\seqlength$ and $i>1\lor j<\seqlength$) that does not contain any
word in $\Language{\pattern}$, there is no $s_c s_{c+1} \dots s_d$ in $\Language{\pattern}$
with $c<i\leq d<j \lor i\leq c<j<d$.
\end{definition}

\begin{example}[pattern properties]
\label{ex:pattern_properties}
A `yes' in column {\sc n}, {\sc o}, {\sc e} or {\sc s} of the pattern part of Table~\ref{tab:fgp}
respectively indicates the \InflexionFree{}, the \OneInflexion{},
the \ExcludeOutIn{}, or the \SingleLetter{} property.
\begin{itemize}
\item[$\bullet$]
Ten out of the $19$ reversible patterns of~\cite{arafailova2016global} have the \InflexionFreeProp{}. 
For instance, the pattern $\DecreasingTerracePatternName$ has the \InflexionFreeProp{}
because it does not simultaneously contain the letters \reg{<} and \reg{>}.
\item[$\bullet$]
Nine out of the $19$ reversible patterns of~\cite{arafailova2016global} have the \OneInflexionProp{}.
The pattern $\PlainPatternName$ has the \OneInflexionProp{} because it 
contains an occurrence of $\reg{<=*>}$, but not an occurrence of $\reg{>=*<}$.
\item[$\bullet$]
Eight out of the $19$ reversible patterns of~\cite{arafailova2016global} have the \ExcludeOutInProp{}.
For instance, the pattern $\IncreasingSequencePatternName$ has the \ExcludeOutInProp{}
because any subword $\word$ of an increasing sequence such that $\word\notin\Language{\IncreasingSequencePatternName}$
cannot be the start or the end of an increasing sequence, since $\word$ is of the form $\reg{==*}$,
i.e.~does not start or end with a $\reg{<}$.
\item[$\bullet$]
$\DecreasingPatternName$, $\IncreasingPatternName$ and $\SteadyPatternName$ have the \SingleLetterProp{}.
\end{itemize}
\end{example}

\section{Sliding Time-Series Constraints}\label{sec:sliding_time_series_constraints}
We introduce the sliding time\nobreakdash-series constraint we consider.
\begin{sloppypar}
\begin{definition}\label{def:slide_time_series_ctr}
Given
a feature $\feature$,
a regular expression $\pattern$,
an integer $\windowsize>1$,
two variables $\low$ and $\up$, and
a sequence of variables $\sequence=\seq$ with $\seqlength\geq\windowsize$,
the \constraint{slide\_sum\_}$\feature$\constraint{\_}$\pattern(\windowsize,\low,\up,\sequence)$
constraint holds iff
\vspace{-5pt}
\begin{multicols}{2}
\noindent\begin{align}\label{def:slide_time_series_ctr_min}
\low=\min_{i\in[1,\seqlength-\windowsize+1]}\result_i,
\end{align}
\noindent\begin{align}\label{def:slide_time_series_ctr_max}
\up=\max_{i\in[1,\seqlength-\windowsize+1]}\result_i,
\end{align}
\end{multicols}
\vspace{-5pt}
\noindent with\hspace*{6pt}$\constraint{sum\_}f\constraint{\_}\pattern\left(\result_i,\hspace*{3pt}x_{i,i+\windowsize-1}\right)$, where $r_i$ is called the \emph{contribution} of the time\nobreakdash-series constraint $\constraint{sum\_}f\constraint{\_}\pattern$ in the window $[i,i+\windowsize-1]$.
\end{definition}
\end{sloppypar}
Cond.~(\ref{def:slide_time_series_ctr_min}), (resp.~(\ref{def:slide_time_series_ctr_max})),
of Def.~\ref{def:slide_time_series_ctr} enforces $\low$ (resp.~$\up$) to be the minimum (resp.~maximum) of the sum of the feature values of feature $\feature$ wrt all maximal occurrences of $\pattern$ in each subsequence of $\windowsize$ consecutive variables of sequence $\sequence$.

\begin{example}[Continuation of Example (\ref{ex:sliding_time_series_ctr_example})]
\label{ex:sliding_time_series_ctr_example_continuated}
Given the pattern $\IncreasingSequencePatternName$ and the feature $\Surf$,
\constraint{slide\_sum\_surf\_incseq} $(10,7,15,3\hspace*{1pt}1\hspace*{1pt}3\hspace*{1pt}3\hspace*{1pt}2\hspace*{1pt}1\hspace*{1pt}1\hspace*{1pt}2\hspace*{1pt}2\hspace*{1pt}2\hspace*{1pt}4\hspace*{1pt}4\hspace*{1pt}3\hspace*{1pt}1\hspace*{1pt}2\hspace*{1pt}2)$
is satisfied because the sum of the surfaces of the increasing sequences
in the different sliding windows of size $10$ is between $7$ and $15$ as
shown in Example~\ref{ex:sliding_time_series_ctr_example}.
\end{example}

\subsection{Computing the Contribution in a Window}

\begin{sloppypar}
In this section,
we consider the patterns $\pattern$ and $\pattern^r$ 
which are the reverse of each other,
a feature $\feature$, 
an integer sequence $x_1 x_2 \dots x_\seqlength$,
and all windows $x_i x_{i+1} \dots x_{i+\windowsize-1}$ of size $\windowsize$
(with $i\in[1,\seqlength-\windowsize+1]$).
We investigate how to evaluate directly from an equation the sum of the
feature values of feature $\feature$
of all pattern occurrences located in a window $[i,j=i+\windowsize-1]$,
assuming all the elements of the right\nobreakdash-hand side of an equation have been previously calculated
in time proportional to $\seqlength$.
As we have several features and several patterns,
we use three equations all derived from the same simple idea,
for which we first present the intuition. Then we define 
sufficient properties of features and patterns that ensure the validity
of each of the three equations~(\ref{formula1}), (\ref{formula2}) and (\ref{formula3}).
At the end of this section, Table~\ref{table:summary}
provides an overview of the validity of each of the three equations
according to properties of the patterns and features.
\end{sloppypar}

\paragraph{Intuition}
Assume we want to deal with the following simplified problem:
given an integer sequence $x_1 x_2 \dots x_\seqlength$,
compute for all subsequences of $\windowsize$ consecutive positions the sum $t_{i,j}=\Sigma_{k\in[i,j]}x_k$ 
(with $j=i+\windowsize-1$) of the corresponding elements in time $O(\seqlength)$.
This can be done by first computing the partial sums $\Sigma_{c\in[1,k]}x_c$ (with $k\in[1,\seqlength]$),
and $\Sigma_{c\in[k,\seqlength]}x_c$ (with $k\in[1,\seqlength]$) and by using the identity
\begin{equation}\label{formula0}
t_{i,j}=\Sigma_{k\in[1,j]}x_k+\Sigma_{k\in[i,\seqlength]}x_k-\Sigma_{k\in[1,\seqlength]}x_k~\text{.}
\end{equation}

Equations~(\ref{formula1}), (\ref{formula2}) and (\ref{formula3})
present three alternative ways to compute $f_\sigma(x_{i,j})$
inspired by Equation~(\ref{formula0}).
\begin{equation}\label{formula1}
\feature_\pattern(x_{i,j})=\feature_\pattern(x_{1,j})+\feature_{\pattern^r}(x_{\seqlength,i})-\feature_\pattern(x_{1,\seqlength})
\end{equation}
\begin{equation}\label{formula2}
\feature_\pattern(x_{i,j})=\max\left(0,\feature_\pattern(x_{1,j})+\feature_{\pattern^r}(x_{\seqlength,i})-\feature_\pattern(x_{1,\seqlength})\right)
\end{equation}
\begin{equation}\label{formula3}
\begin{aligned}
\text{if no }\sigma\text{-pattern in}~x_{i,j}~&\text{then}~\feature_\pattern(x_{i,j})=0\\
&\text{else}\hspace*{3pt}~\feature_\pattern(x_{i,j})=\feature_\pattern(x_{1,j})+\feature_{\pattern^r}(x_{\seqlength,i})-\feature_\pattern(x_{1,\seqlength})
\end{aligned}
\end{equation}

Depending on the properties of the pattern $\pattern$ and of the feature $\feature$,
we investigate the cases when Equations~(\ref{formula1}), (\ref{formula2}) and (\ref{formula3}) are valid.
\begin{example}
Consider the $\DecreasingSequencePatternName$ pattern of Table~\ref{tab:fgp},
the sequence $\word=2\hspace*{1pt}1\hspace*{1pt}1\hspace*{1pt}1\hspace*{1pt}0$,
the window size $\windowsize=2$, i.e.~the four sliding windows
$2\hspace*{1pt}1$, $1\hspace*{1pt}1$, $1\hspace*{1pt}1$ and $1\hspace*{1pt}0$.
\begin{itemize}
\item[$\bullet$]
Equation~(\ref{formula1}) provides the incorrect $\Surf$ feature value
for two of the four sliding windows,
namely values $3$, $-1$, $-1$ and $1$ rather than the expected values $3$, $0$, $0$ and $1$.
For the second window, shown in grey on the figure on the right,
this is because there is a non\nobreakdash-empty gap (shown in red) between
the leftmost and rightmost decreasing sequences in $\word$.
Equation~(\ref{formula2})\\ gives the correct value since it cancels
out the contribution of the gap.
\begin{tikzpicture}[remember picture,overlay]
\node[anchor=south west] at ($(current page.north east)-(5.6cm,21.9cm)$) {
\begin{tikzpicture}[scale=1,remember picture,overlay]
\begin{scope}
\fill[gray!50] (0.28,-0.1) rectangle (0.55,0.1);
\node[anchor=west] (seq) at (0,0) {\scriptsize$2 1 {\color{red}1} 1 0$};
\draw[line width=0.5pt] (0.1,0.05) -- (0.1,0.15) -- (0.85,0.15) -- (0.85,0.05);
\draw[line width=0.5pt] (0.1,-0.05) -- (0.1,-0.15) -- (0.4,-0.15) -- (0.4,-0.05);
\draw[line width=0.5pt] (0.55,-0.05) -- (0.55,-0.15) -- (0.85,-0.15) -- (0.85,-0.05);
\node[anchor=west] (pre) at (-0.05,-0.3) {\scriptsize$+3$};
\node[anchor=west] (suf) at (0.4,-0.3) {\scriptsize$+1$};
\node[anchor=west] (word) at (0.1,0.3) {\scriptsize$-5$};
\end{scope}
\end{tikzpicture}
};
\end{tikzpicture}
\item[$\bullet$]
While Equations~(\ref{formula1}) and~(\ref{formula2}) give the incorrect $\MinFeature$ feature value
for two of the four sliding windows,
namely values $1$, $1$, $1$ and $0$ rather than values $1$, $0$, $0$ and $0$,
Equation~(\ref{formula3}) provides the correct values.
\end{itemize}
\end{example}

\paragraph{Case Analysis}

{
\setlength{\tabcolsep}{12pt}
\begin{table}[!h]
\centering
\begin{tabular}{ccc}\toprule
\footnotesize case & \footnotesize condition                                                    & \footnotesize illustration \\ \midrule
\scriptsize(1)  & $u<i$                                                        &
\begin{tikzpicture}[baseline=(current bounding box.center),scale=0.9]
\draw[->] (0,0) -- (3,0);
\draw (1,0) -- (1,0.1);
\draw (2,0) -- (2,0.1);
\node[anchor=south] (1) at (1.05,-0.4) {$i$};
\node[anchor=south] (2) at (1.95,-0.45) {$j$};
\filldraw[fill=gray!50,draw=black] (0.2,0) rectangle (0.8,0.2);
\node[anchor=south] (3) at (0.25,-0.4) {$\ell$};
\node[anchor=south] (4) at (0.75,-0.4) {$u$};
\end{tikzpicture}
\\
\scriptsize(2)  & $\ell>j$                                                     &
\begin{tikzpicture}[baseline=(current bounding box.center),scale=0.9]
\draw[->] (0,0) -- (3,0);
\draw (1,0) -- (1,0.1);
\draw (2,0) -- (2,0.1);
\node[anchor=south] (1) at (1.05,-0.4) {$i$};
\node[anchor=south] (2) at (1.95,-0.45) {$j$};
\filldraw[fill=gray!50,draw=black] (2.2,0) rectangle (2.8,0.2);
\node[anchor=south] (3) at (2.25,-0.4) {$\ell$};
\node[anchor=south] (4) at (2.75,-0.4) {$u$};
\end{tikzpicture}
\\
\scriptsize(3)  & $i\leq\ell\leq u\leq j$                                      &
\begin{tikzpicture}[baseline=(current bounding box.center),scale=0.9]
\draw[->] (0,0) -- (3,0);
\draw (1,0) -- (1,0.1);
\draw (2,0) -- (2,0.1);
\node[anchor=south] (1) at (1.05,-0.4) {$i$};
\node[anchor=south] (2) at (1.95,-0.45) {$j$};
\filldraw[fill=gray!50,draw=black] (1.2,0) rectangle (1.8,0.2);
\node[anchor=south] (3) at (1.25,-0.4) {$\ell$};
\node[anchor=south] (4) at (1.75,-0.4) {$u$};
\end{tikzpicture}
\\
\scriptsize(4)  & $\ell<i\leq u\leq j\land p(s_{i,u-1})\in\Language{\pattern}$    &
\begin{tikzpicture}[baseline=(current bounding box.center),scale=0.9]
\filldraw[fill=gray!50,draw=black] (0.7,0) rectangle (1.8,0.2);
\fill[gray!20] (1,0) rectangle (1.5,0.2);
\draw[draw=black] (0.7,0) rectangle (1.8,0.2);
\draw[->] (0,0) -- (3,0);
\draw (1,0) -- (1,0.1);
\draw (2,0) -- (2,0.1);
\draw (1.5,0) -- (1.5,0.1);
\node[anchor=south] (1) at (1.05,-0.4) {$i$};
\node[anchor=south] (2) at (1.95,-0.45) {$j$};
\node[anchor=south] (3) at (0.65,-0.4) {$\ell$};
\node[anchor=south] (4) at (1.75,-0.4) {$u$};
\node[anchor=south] (5) at (1.5,-0.4) {$\alpha$};
\end{tikzpicture}
\\
\scriptsize(5)  & $\ell<i\leq u< j\land p(s_{i,u-1})\notin\Language{\pattern}$ &
\begin{tikzpicture}[baseline=(current bounding box.center),scale=0.9]
\filldraw[fill=gray!50,draw=black] (0.7,0) rectangle (1.8,0.2);
\draw[->] (0,0) -- (3,0);
\draw (1,0) -- (1,0.1);
\draw (2,0) -- (2,0.1);
\node[anchor=south] (1) at (1.05,-0.4) {$i$};
\node[anchor=south] (2) at (1.95,-0.45) {$j$};
\node[anchor=south] (3) at (0.65,-0.4) {$\ell$};
\node[anchor=south] (4) at (1.75,-0.4) {$u$};
\end{tikzpicture}
\\
\scriptsize(6)  & $i\leq\ell\leq j<u\land s(s_{\ell,j-1})\in\Language{\pattern}$     &
\begin{tikzpicture}[baseline=(current bounding box.center),scale=0.9]
\filldraw[fill=gray!50,draw=black] (1.2,0) rectangle (2.3,0.2);
\fill[gray!20] (1.5,0) rectangle (2,0.2);
\draw[draw=black] (1.2,0) rectangle (2.3,0.2);
\draw[->] (0,0) -- (3,0);
\draw (1,0) -- (1,0.1);
\draw (2,0) -- (2,0.1);
\draw (1.5,0) -- (1.5,0.1);
\node[anchor=south] (1) at (1.05,-0.4) {$i$};
\node[anchor=south] (2) at (1.95,-0.45) {$j$};
\node[anchor=south] (3) at (1.25,-0.4) {$\ell$};
\node[anchor=south] (4) at (2.25,-0.4) {$u$};
\node[anchor=south] (5) at (1.5,-0.45) {$\beta$};
\end{tikzpicture}
\\
\scriptsize(7)  & $i<\ell\leq j<u\land s(s_{\ell,j-1})\notin\Language{\pattern}$  &
\begin{tikzpicture}[baseline=(current bounding box.center),scale=0.9]
\filldraw[fill=gray!50,draw=black] (1.2,0) rectangle (2.3,0.2);
\draw[->] (0,0) -- (3,0);
\draw (1,0) -- (1,0.1);
\draw (2,0) -- (2,0.1);
\node[anchor=south] (1) at (1.05,-0.4) {$i$};
\node[anchor=south] (2) at (1.95,-0.45) {$j$};
\node[anchor=south] (3) at (1.25,-0.4) {$\ell$};
\node[anchor=south] (4) at (2.25,-0.4) {$u$};
\end{tikzpicture}
\\
\scriptsize(8)  & $(\ell\leq i\leq j\leq u) \land (\ell\neq i \lor  j \neq u)$ &
\begin{tikzpicture}[baseline=(current bounding box.center),scale=0.9]
\filldraw[fill=gray!50,draw=black] (0.7,0) rectangle (2.3,0.2);
\draw[->] (0,0) -- (3,0);
\draw (1,0) -- (1,0.1);
\draw (2,0) -- (2,0.1);
\node[anchor=south] (1) at (1.05,-0.4) {$i$};
\node[anchor=south] (2) at (1.95,-0.45) {$j$};
\node[anchor=south] (3) at (0.75,-0.4) {$\ell$};
\node[anchor=south] (4) at (2.25,-0.4) {$u$};
\end{tikzpicture}
\\
\bottomrule
\end{tabular}
\caption{\label{table:cases}{Positioning an occurrence of a pattern wrt a window;
within cases~(4) and~(6) the non-empty words $p(s_{i,u-1})$ and $s(s_{\ell,j-1})$
are shown in light grey.}}
\end{table}
}

\begin{sloppypar}
Consider
a sequence $\seq$,
a window $[i,j]$, and
a maximal occurrence of pattern $o$
whose signature is $s_\ell s_{\ell+1}\dots s_{u-1}$ (with $1\leq\ell\leq u\leq n$).
Table~\ref{table:cases} provides eight cases summarising all the possible positioning
of $x_{\ell,u}$ wrt $[i,j]$, where $p(s_{i,u-1})$ (resp.~$s(s_{\ell,j-1})$) denotes
the longest prefix $s_i s_{i+1}\dots s_{\alpha-1}$ of $s_i s_{i+1}\dots s_{u-1}$
(resp.~the longest suffix $s_\beta s_{\beta+1}\dots s_{j-1}$ of $s_\ell s_{\ell+1}\dots s_{j-1}$)
in $\Language{\pattern}$ if such word
exists, the empty word otherwise.
For cases~(1\nobreakdash--7) of Table~\ref{table:cases},
columns $\feature_\pattern(x_{1,j})$ (resp.~$\feature_{\pattern^r}(x_{\seqlength,i})$)
and $\feature_\pattern(x_{1,n})$ of Table~\ref{table:contribution}
provide the feature value of the $\pattern$-pattern occurrence $o$ 
(resp.~$\pattern^r$-pattern occurrence $o^r$)
wrt~$x_1 x_2\dots x_j$ (resp.~$x_\seqlength x_{\seqlength-1}\dots x_i$) and~$\seq$;
the last three columns give the contribution of $o$ in the right-hand side of Equations (\ref{formula1}), (\ref{formula2}), (\ref{formula3}). These contributions agree with the positioning of $x_{\ell,u}$
wrt $[i,j]$,
except for the three grey cells, which only work for non-negative feature values.
Case~(8) of Table~\ref{table:cases} corresponds to a maximal occurrence of pattern $o$ whose signature starts before $i$ and ends after $j$.
To study Case~(8), the next section classifies a pattern wrt a window.
\end{sloppypar}

{
\setlength{\tabcolsep}{4pt}
\begin{table}[!b]
\centering
\begin{tabular}{ccccccc}\toprule
\footnotesize case  & $\feature_\pattern(x_{1,j})$   & $\feature_{\pattern^r}(x_{n,i})$   & $\feature_\pattern(x_{1,n})$ & \footnotesize Eq.~(\ref{formula1}) & \footnotesize Eq.~(\ref{formula2}) & \footnotesize Eq.~(\ref{formula3}) \\\midrule
\scriptsize (1)  & $\feature_\pattern(x_{\ell,u})$         & $0$            & $\feature_\pattern(x_{\ell,u})$    & $0$              & $0$ & $0$ \\
\scriptsize (2)  & $0$            & $\feature_{\pattern^r}(x_{u,\ell})$         & $\feature_\pattern(x_{\ell,u})$    & $0$              & $0$ & $0$ \\
\scriptsize (3)  & $\feature_\pattern(x_{\ell,u})$         & $\feature_{\pattern^r}(x_{u,\ell})$         & $\feature_\pattern(x_{\ell,u})$    & $\feature_{\pattern^r}(x_{u,\ell})$           & \cellcolor{gray!25}$\max(0,\feature_{\pattern^r}(x_{u,\ell}))$ & $\feature_{\pattern^r}(x_{u,\ell})$ \\
\scriptsize (4)  & $\feature_\pattern(x_{\ell,u})$         & $\feature_{\pattern^r}(x_{\fino,i})$ & $\feature_\pattern(x_{\ell,u})$    & $\feature_{\pattern^r}(x_{\fino,i})$   & \cellcolor{gray!25}$\max(0,\feature_{\pattern^r}(x_{\fino,i}))$ & $\feature_{\pattern^r}(x_{\fino,i})$ \\
\scriptsize (5)  & $\feature_\pattern(x_{\ell,u})$         & $0$            & $\feature_\pattern(x_{\ell,u})$    & $0$              & $0$ & $0$ \\
\scriptsize (6)  & $\feature_\pattern(x_{\starto,j})$ & $\feature_{\pattern^r}(x_{u,\ell})$         & $\feature_\pattern(x_{\ell,u})$    & $\feature_\pattern(x_{\starto,j})$   & \cellcolor{gray!25}$\max(0,\feature_\pattern(x_{\starto,j}))$ & $\feature_\pattern(x_{\starto,j})$ \\
\scriptsize (7)  & $0$            & $\feature_{\pattern^r}(x_{u,\ell})$         & $\feature_\pattern(x_{\ell,u})$    & $0$              & $0$ & $0$ \\
\bottomrule
\end{tabular}
\caption{\label{table:contribution}{[\textsc{columns 2 to 4}]~values of $\feature_\pattern(x_{1,j})$, $\feature_{\pattern^r}(x_{n,i})$ and $\feature_\pattern(x_{1,n})$ wrt cases~(1$-$7) of Table~\ref{table:cases};
[\textsc{columns 5 to 7}]~contribution of an occurrence of $\pattern$ in a window wrt the right-hand side of Equations~(\ref{formula1}), (\ref{formula2}) and (\ref{formula3}).}
}
\end{table}
}

\subsubsection{A Systematic Classification of Patterns wrt Windows}\label{sec:classify_pattern}

\begin{definition}
\label{def:word_status}\emph{[type of a word wrt a pattern]}
Given a pattern $\pattern$,
the \emph{type} of a proper factor $\word=\word_1\word_2\dots\word_k$ of a word in $\Language{\pattern}$
wrt $\pattern$ is defined by five mutually incompatible conditions:

{
\begin{center}
\begin{minipage}{10cm}
\begin{itemize}
\item[$\bullet$]
$\cout~\text{if}\hspace*{7pt}~\nexists c,d:\hspace*{0.1pt}~1\leq c\leq d\leq k\hspace*{0.4pt}~\land~\word_c\word_{c+1}\dots\word_d\in\Language{\pattern}$\hspace*{13pt}
\item[$\bullet$]
$\cfac$ if
$\begin{cases}
\exists c,d:~1\leq c\leq d\leq k~\land~\word_c\word_{c+1}\dots\word_d$\hspace*{3.4pt}$\in\Language{\pattern} \\
\nexists d:\hspace*{9pt}~1\leq d\leq k\hspace*{17.3pt}~\land~\word_1\word_2\dots\word_d$\hspace*{12.5pt}$\in\Language{\pattern}\\
\nexists c:\hspace*{10pt}~1\leq c\leq k\hspace*{18pt}~\land~\word_c\word_{c+1}\dots\word_k$ $\in\Language{\pattern}
\end{cases}$
\item[$\bullet$]
$\cpre$ if
$\begin{cases}
\exists d:\hspace*{8.5pt}~1\leq d\leq k\hspace*{17.3pt}~\land~\word_1\word_2\dots\word_d$\hspace*{13pt}$\in\Language{\pattern} \\
\nexists c:\hspace*{9.5pt}~1\leq c\leq k\hspace*{18pt}~\land~\word_c\word_{c+1}\dots\word_k$ $\in\Language{\pattern}
\end{cases}$
\item[$\bullet$]
$\csuf$ if
$\begin{cases}
\exists c:\hspace*{10pt}~1\leq c\leq k\hspace*{18.3pt}~\land~\word_c\word_{c+1}\dots\word_k$\hspace*{4.2pt}$\in\Language{\pattern} \\
\nexists d:\hspace*{9pt}~1\leq d\leq k\hspace*{17.6pt}~\land~\word_1\word_2\dots\word_d$\hspace*{13.5pt}$\in\Language{\pattern}
\end{cases}$
\item[$\bullet$]
$\cin$\hspace*{4pt} if
$\begin{cases}
\exists d:\hspace*{10pt}~1\leq d\leq k\hspace*{17.6pt}~\land~\word_1\word_2\dots\word_d$\hspace*{12.5pt}$\in\Language{\pattern} \hspace*{11pt} \\
\exists c:\hspace*{10.7pt}~1\leq c\leq k\hspace*{18.3pt}~\land~\word_c\word_{c+1}\dots\word_k$ $\in\Language{\pattern} \hspace*{11pt}
\end{cases}$
\end{itemize}
\end{minipage}
\end{center}
}
\end{definition}

\begin{sloppypar}
In Definition~\ref{def:word_status},
``$\cfac$'', ``$\cpre$'' and ``$\csuf$'' convey the idea of ``factor'', ``prefix'' and ``suffix''.
Note that a word with the ``$\cin$'' type wrt a convex pattern $\pattern$
is in $\Language{\pattern}$.
The languages associated with the five mutually incompatible
conditions of Definition~\ref{def:word_status} are defined as
$\Language{\cout}=\Sigma^+\setminus(\Sigma^*\Language{\pattern}\Sigma^*)$,
$\Language{\cfac}=\Sigma^+\Language{\pattern}\Sigma^+\cap
                  \Sigma^*\setminus(\Language{\pattern}\Sigma^+)\cap
                  \Sigma^*\setminus(\Sigma^+\Language{\pattern})\cap
                  \Sigma^*\setminus\Language{\pattern}$,
$\Language{\cpre}=\Language{\pattern}\Sigma^+\cap
                  \Sigma^*\setminus(\Sigma^+\Language{\pattern})\cap
                  \Sigma^*\setminus\Language{\pattern}$,
$\Language{\csuf}=\Sigma^+\Language{\pattern}\cap
                  \Sigma^*\setminus(\Language{\pattern}\Sigma^+)\cap
                  \Sigma^*\setminus\Language{\pattern}$, and
$\Language{\cin}=\Language{\pattern}\Sigma^*\cap
                 \Sigma^*\Language{\pattern}$.
Note that because of our hypothesis that $\Language{\pattern}$
does not contain the empty word,
the languages $\Language{\cout}$, $\Language{\cfac}$, $\Language{\cpre}$,
$\Language{\csuf}$ and $\Language{\cin}$ do not contain the empty word.

\begin{definition}
\label{def:word_triple}\emph{[type and signature of a word wrt one of its proper factors and wrt a pattern]}
Given a pattern $\pattern$,
consider a word $\word=\word_1\word_2\dots\word_k$ of $\Language{\pattern}$,
and one of its \emph{proper factors} $v=\word_i\word_{i+1}\dots\word_j$.
The \emph{type} of $\word$ wrt $v$ and $\pattern$
is defined by $\langle t_1,t_2,t_3\rangle$ where
$t_1$, $t_2$ and $t_3$ are respectively the type of words
$\word_1\word_2\dots\word_j$,
$\word_i\word_{i+1}\dots\word_j$ and
$\word_i\word_{i+1}\dots\word_k$
wrt pattern $\pattern$ as defined by Definition~\ref{def:word_status}.
The \emph{signature} of $\word$ wrt $v$ and $\pattern$
is defined by $\langle \mathit{sig}_1,\mathit{sig}_2,\mathit{sig}_3\rangle$, where
$\mathit{sig}_c=\textnormal{if }t_c=\cout\textnormal{ then }1\textnormal{ else }0$ (with $c\in[1,3]$).
\end{definition}
\end{sloppypar}

\begin{theorem}
\label{theorem:map}\emph{[map of feasible types wrt any pattern]}
Of the $125$ possible types of Definition~\ref{def:word_triple},
only $61$ triples shown in Figure~\ref{fig:triples_map} are feasible.
\end{theorem}
\begin{proof}
For each triple of Figure~\ref{fig:triples_map}, Appendix~\ref{sec:witness_feasible_types}
provides a witness pattern which generates such triple.
We now prove that the missing triples cannot be obtained from any pattern.
\begin{sloppypar}
\begin{itemize}
\item[\ding{172}]
$\langle\cout,\neq$\hspace*{1pt}$\cout,\shorthyphen\rangle$
(resp. $\langle\shorthyphen,\neq$\hspace*{1pt}$\cout,\cout\rangle$)
is not feasible since an ``$\cout$'' in
$v=\word_1\word_2\dots\word_j$ (resp.~$\word_i\word_{i+1}\dots\word_k$)
would imply an ``$\cout$'' in any subsequence of $v$, namely in $\word_i\word_{i+1}\dots\word_j$,
a contradiction.
\item[\ding{173}]
$\langle\cfac,\csuf,\shorthyphen\rangle$
(resp. $\langle\shorthyphen,\cpre,\cfac\rangle$)
is not feasible since a ``$\csuf$'' (resp.~``$\cpre$'')
in $\word_i\word_{i+1}\dots\word_j$ would imply a ``$\csuf$'' (resp.~``$\cpre$'') or an ``$\cin$''
in $\word_1\word_2\dots\word_j$ (resp.~$\word_i\word_{i+1}\dots\word_k$),
a contradiction.
\item[\ding{174}]
$\langle\cfac,\cin,\shorthyphen\rangle$
(resp.~$\langle\shorthyphen,\cin,\cfac\rangle$)
is not feasible since a ``$\cin$'' in $\word_i\word_{i+1}\dots\word_j$ would imply
a ``$\csuf$'' (resp.~``$\cpre$'') or an ``$\cin$''
in $\word_1\word_2\dots\word_j$ (resp.~$\word_i\word_{i+1}\dots\word_k$), a contradiction.
\item[\ding{175}]
$\langle\cpre,\csuf,\shorthyphen\rangle$
(resp.~$\langle\shorthyphen,\cpre,\csuf\rangle$)
is not feasible since a ``$\csuf$'' (resp.~``$\cpre$'')
in $\word_i\word_{i+1}\dots\word_j$ and
a ``$\cpre$'' (resp.~``$\csuf$'')
in $v=\word_1\word_2\dots\word_j$ (resp.~$v=\word_i\word_{i+1}\dots\word_k$)
would imply an ``$\cin$'' in $v$, a contradiction.
\item[\ding{176}]
$\langle\cpre,\cin,\shorthyphen\rangle$
(resp.~$\langle\shorthyphen,\cin,\csuf\rangle$)
is not feasible since an ``$\cin$'' in $\word_i\word_{i+1}\dots\word_j$ and
a ``$\cpre$'' (resp.~``$\csuf$'')
in $v=\word_1\word_2\dots\word_j$ (resp.~$v=\word_i\word_{i+1}\dots\word_k$)
would imply an ``$\cin$'' in $v$, a contradiction.\qed
\end{itemize}
\end{sloppypar}
\end{proof}

\begin{figure}[!h]
\centering
\begin{tikzpicture}

\fill[gray!7,rounded corners=2pt] (0.6,2.3) rectangle (1.82,2.9);
\fill[gray!7,rounded corners=2pt] (0.6,-1.15) rectangle (2.6,2.03);
\fill[gray!7,rounded corners=2pt] (4.2,-1.15) rectangle (11.8,2.9);
\fill[gray!7,rounded corners=2pt] (-0.2,-5.7) rectangle (11.8,-1.4);
\node (a) at ($(1.82,2.9)-(0.25,0.2)$) {\scriptsize(A)};
\node (b) at ($(2.6,2.03)-(0.25,0.2)$) {\scriptsize(B)};
\node (c) at ($(11.8,2.9)-(0.25,0.2)$) {\scriptsize(C)};
\node (d) at ($(11.8,-1.4)-(0.25,0.2)$) {\scriptsize(D)};

\begin{scope}[every node/.style={minimum width=.65cm,minimum height=0.65cm,font=\scriptsize\strut},xshift=1cm,yshift=0cm]
\fill[white!15,rounded corners=2pt] (-0.25,-1) rectangle (1.1,0.15);
\node (ooi) at (0,0)             {$\fout\fout\fin$};
\node (oos) [right=0.6em of ooi] {$\fout\fout\fsuf$};
\node (oop) [below=0.6em of ooi] {$\fout\fout\fpre$};
\node (oof) [right=0.6em of oop] {$\fout\fout\ffac$};
\draw[->] (ooi) -- (oos);
\draw[->] (oop) -- (oof);
\draw[->] (ooi) -- (oop);
\draw[->] (oos) -- (oof);
\end{scope}

\begin{scope}[every node/.style={minimum width=.65cm,minimum height=0.65cm,font=\scriptsize\strut},xshift=1cm,yshift=1.73cm]
\fill[white!15,rounded corners=2pt] (-0.25,-1) rectangle (1.1,0.15);
\node (ioo) at (0,0)             {$\fin\fout\fout$};
\node (soo) [right=0.6em of ioo] {$\fsuf\fout\fout$};
\node (poo) [below=0.6em of ioo] {$\fpre\fout\fout$};
\node (foo) [right=0.6em of poo] {$\ffac\fout\fout$};
\draw[->] (ioo) -- (soo);
\draw[->] (poo) -- (foo);
\draw[->] (ioo) -- (poo);
\draw[->] (soo) -- (foo);
\end{scope}

\begin{scope}[every node/.style={minimum width=.65cm,minimum height=0.65cm,font=\scriptsize\strut},xshift=1cm,yshift=2.6cm]
\fill[white!15,rounded corners=2pt] (-0.25,-0.15) rectangle (0.27,0.15);
\node (ooo) at (0,0)             {$\fout\fout\fout$};
\end{scope}

\begin{scope}[every node/.style={minimum width=.65cm,minimum height=0.65cm,font=\scriptsize\strut},xshift=6.3cm,yshift=0.87cm]
\fill[white!15,rounded corners=2pt] (-1.12,-1) rectangle (1.1,1);
\node (ioi) at (0,0)             {$\fin\fout\fin$};
\node (ios) [right=0.6em of ioi] {$\fin\fout\fsuf$};
\node (iop) [left =0.6em of ioi] {$\fin\fout\fpre$};
\node (soi) [above=0.6em of ioi] {$\fsuf\fout\fin$};
\node (poi) [below=0.6em of ioi] {$\fpre\fout\fin$};
\node (sop) [above=0.6em of iop] {$\fsuf\fout\fpre$};
\node (pop) [below=0.6em of iop] {$\fpre\fout\fpre$};
\node (sos) [above=0.6em of ios] {$\fsuf\fout\fsuf$};
\node (pos0)[below=0.6em of ios] {$\fpre\fout\fsuf$};
\draw[fill=gray!30] (pos0) ellipse (1em and 0.5em);
\node (pos) [below=0.6em of ios] {$\fpre\fout\fsuf$};
\node (iof1)[left =0.6em of iop] {$\fin\fout\ffac$};
\node (iof2)[right=0.6em of ios] {$\fin\fout\ffac$};
\node (foi1)[above=0.6em of soi] {$\ffac\fout\fin$};
\node (foi2)[below=0.6em of poi] {$\ffac\fout\fin$};
\node (sof1)[above=0.6em of sop] {$\fsuf\fout\ffac$};
\node (sof2)[above=0.6em of sos] {$\fsuf\fout\ffac$};
\node (pof1)[below=0.6em of pop] {$\fpre\fout\ffac$};
\node (pof2)[below=0.6em of pos] {$\fpre\fout\ffac$};
\node (fop1)[left =0.6em of sop] {$\ffac\fout\fpre$};
\node (fop2)[left =0.6em of pop] {$\ffac\fout\fpre$};
\node (fos1)[right=0.6em of sos] {$\ffac\fout\fsuf$};
\node (fos2)[right=0.6em of pos] {$\ffac\fout\fsuf$};
\draw[->] (ioi) -- (ios);
\draw[->] (ioi) -- (iop);
\draw[->] (ioi) -- (soi);
\draw[->] (ioi) -- (poi);
\draw[->] (iop) -- (sop);
\draw[->] (iop) -- (pop);
\draw[->] (ios) -- (sos);
\draw[->] (ios) -- (pos);
\draw[->] (soi) -- (sop);
\draw[->] (soi) -- (sos);
\draw[->] (poi) -- (pop);
\draw[->] (poi) -- (pos);
\draw[->,densely dotted] (iop) -- (iof1);
\draw[->,densely dotted] (ios) -- (iof2);
\draw[->,densely dotted] (soi) -- (foi1);
\draw[->,densely dotted] (poi) -- (foi2);
\draw[->,densely dotted] (sop) -- (sof1);
\draw[->,densely dotted] (sos) -- (sof2);
\draw[->,densely dotted] (pop) -- (pof1);
\draw[->,densely dotted] (pos) -- (pof2);
\draw[->,densely dotted] (sop) -- (fop1);
\draw[->,densely dotted] (pop) -- (fop2);
\draw[->,densely dotted] (sos) -- (fos1);
\draw[->,densely dotted] (pos) -- (fos2);
\end{scope}

\begin{scope}[every node/.style={minimum width=.65cm,minimum height=0.65cm,font=\scriptsize\strut},xshift=10.5cm,yshift=0.87cm]
\fill[white!15,rounded corners=2pt] (-1.1,-1) -- (-1.1,0.18) -- (-0.18,0.18) -- (-0.18,1.03) -- (1.1,1.03) -- (1.1,-0.18) -- (0.18,-0.18) -- (0.18,-1) -- cycle;
\node (fof) at (0,0)             {$\ffac\fout\ffac$};
\node (pof) [right=0.6em of fof] {$\fpre\fout\ffac$};
\node (fop) [left =0.6em of fof] {$\ffac\fout\fpre$};
\node (sof) [above=0.6em of fof] {$\fsuf\fout\ffac$};
\node (fos) [below=0.6em of fof] {$\ffac\fout\fsuf$};
\node (foi) [below=0.6em of fop] {$\ffac\fout\fin$};
\node (iof) [above=0.6em of pof] {$\fin\fout\ffac$};
\node (foi1)[below=0.6em of foi] {};
\node (fos1)[below=0.6em of fos] {};
\node (pof1)[below=0.6em of pof] {};
\node (iof1)[above=0.6em of iof] {};
\node (sof1)[above=0.6em of sof] {};
\node (fop1)[above=0.6em of fop] {};
\draw[->] (pof) -- (fof);
\draw[->] (fop) -- (fof);
\draw[->] (sof) -- (fof);
\draw[->] (fos) -- (fof);
\draw[->] (foi) -- (fop);
\draw[->] (foi) -- (fos);
\draw[->] (iof) -- (sof);
\draw[->] (iof) -- (pof);
\draw[->,densely dotted] (foi1)-- (foi);
\draw[->,densely dotted] (fos1)-- (fos);
\draw[->,densely dotted] (pof1)-- (pof);
\draw[->,densely dotted] (iof1)-- (iof);
\draw[->,densely dotted] (sof1)-- (sof);
\draw[->,densely dotted] (fop1)-- (fop);
\end{scope}

\begin{scope}[every node/.style={minimum width=.65cm,minimum height=0.65cm,font=\scriptsize\strut},xshift=1.8cm,yshift=-3.5cm]
\fill[white!15,rounded corners=2pt] (-1.1,-1) -- (-1.1,1.88) -- (0.18,1.88) -- (0.18,1.0) -- (1.08,1.0) -- (1.08,-1.88) -- (-0.2,-1.88) -- (-0.2,-1) -- cycle;
\node (iii0) at (0,0)            {$\fin\fin\fin$};
\draw[fill=gray!30] (iii0) ellipse (1em and 0.5em);
\node (iii) at (0,0)             {$\fin\fin\fin$};
\node (sii) [right=0.6em of iii] {$\fsuf\fin\fin$};
\node (iip) [left =0.6em of iii] {$\fin\fin\fpre$};
\node (isi) [above=0.6em of iii] {$\fin\fsuf\fin$};
\node (ipi) [below=0.6em of iii] {$\fin\fpre\fin$};
\node (isp) [above=0.6em of iip] {$\fin\fsuf\fpre$};
\node (ipp) [below=0.6em of iip] {$\fin\fpre\fpre$};
\node (ssi) [above=0.6em of sii] {$\fsuf\fsuf\fin$};
\node (spi) [below=0.6em of sii] {$\fsuf\fpre\fin$};
\node (sss) [above=0.6em of ssi] {$\fsuf\fsuf\fsuf$};
\node (iss0)[above=0.6em of isi] {$\fin\fsuf\fsuf$};
\draw[fill=gray!30] (iss0) ellipse (1em and 0.5em);
\node (iss) [above=0.6em of isi] {$\fin\fsuf\fsuf$};
\node (isf) [above=0.6em of isp] {$\fin\fsuf\ffac$};
\node (ppp) [below=0.6em of ipp] {$\fpre\fpre\fpre$};
\node (ppi0)[below=0.6em of ipi] {$\fpre\fpre\fin$};
\draw[fill=gray!30] (ppi0) ellipse (1em and 0.5em);
\node (ppi) [below=0.6em of ipi] {$\fpre\fpre\fin$};
\node (fpi) [below=0.6em of spi] {$\ffac\fpre\fin$};
\node (ssf) [left =0.6em of isf] {$\fsuf\fsuf\ffac$};
\node (spp1)[left =0.6em of isp] {$\fsuf\fpre\fpre$};
\node (sip1)[left =0.6em of iip] {$\fsuf\fin\fpre$};
\node (spp2)[left =0.6em of ipp] {$\fsuf\fpre\fpre$};
\node (ssp1)[right=0.6em of ssi] {$\fsuf\fsuf\fpre$};
\node (sip2)[right=0.6em of sii] {$\fsuf\fin\fpre$};
\node (ssp2)[right=0.6em of spi] {$\fsuf\fsuf\fpre$};
\node (fpp) [right=0.6em of fpi] {$\ffac\fpre\fpre$};
\node (iff) [right=0.6em of sss] {$\fin\ffac\ffac$};
\node (ffi) [left =0.6em of ppp] {$\ffac\ffac\fin$};
\node (ifs0)at (ssf.south) {$\fin\ffac\fsuf$};
\draw[draw=gray!20,fill=gray!20] (ifs0) ellipse (0.8em and 0.4em);
\node (ifs) at (ssf.south) {$\fin\ffac\fsuf$};
\node (pfi0)at (ffi.north) {$\fpre\ffac\fin$};
\draw[draw=gray!20,fill=gray!20] (pfi0) ellipse (0.8em and 0.4em);
\node (pfi) at (ffi.north) {$\fpre\ffac\fin$};
\node (sfi1) at (ssp1.south) {$\fsuf\ffac\fin$};
\node (sfi2) at (ssp2.north) {$\fsuf\ffac\fin$};
\node (ifi1) at (ssp1.north) {$\fin\ffac\fin$};
\node (ifi2) at (ssp2.south) {$\fin\ffac\fin$};
\node (ifp1) at (spp1.south) {$\fin\ffac\fpre$};
\node (ifp2) at (spp2.north) {$\fin\ffac\fpre$};
\draw[->,densely dotted] (ssi) -- (sfi1);
\draw[->,densely dotted] (spi) -- (sfi2);
\path (isi) edge [->,densely dotted,bend angle=5,bend left]  node {} (ifi1);
\path (ipi) edge [->,densely dotted,bend angle=5,bend right] node {} (ifi2);
\draw[->,densely dotted] (isp) -- (ifp1);
\draw[->,densely dotted] (ipp) -- (ifp2);
\draw[line width=4pt,gray!30] (iii) -- (ipi);
\draw[line width=4pt,gray!30] (ipi) -- (ppi);
\draw[line width=4pt,gray!30] (iii) -- (isi);
\draw[line width=4pt,gray!30] (isi) -- (iss);
\draw[->] (iii) -- (isi);
\draw[->] (iii) -- (sii);
\draw[->] (iii) -- (ipi);
\draw[->] (iii) -- (iip);
\draw[->] (iip) -- (isp);
\draw[->] (iip) -- (ipp);
\draw[->] (sii) -- (ssi);
\draw[->] (sii) -- (spi);
\draw[->] (isi) -- (isp);
\draw[->] (isi) -- (ssi);
\draw[->] (ipi) -- (ipp);
\draw[->] (ipi) -- (spi);
\draw[->] (isp) -- (isf);
\draw[->] (iss) -- (isf);
\draw[->] (isi) -- (iss);
\draw[->] (spi) -- (fpi);
\draw[->] (ppi) -- (fpi);
\draw[->] (ipi) -- (ppi);
\draw[->,densely dotted] (iss) -- (sss);
\draw[->,densely dotted] (ssi) -- (sss);
\draw[->,densely dotted] (ipp) -- (ppp);
\draw[->,densely dotted] (ppi) -- (ppp);
\draw[->,densely dotted] (isf) -- (ssf);
\draw[->,densely dotted] (isp) -- (spp1);
\draw[->,densely dotted] (iip) -- (sip1);
\draw[->,densely dotted] (ipp) -- (spp2);
\draw[->,densely dotted] (ssi) -- (ssp1);
\draw[->,densely dotted] (sii) -- (sip2);
\draw[->,densely dotted] (spi) -- (ssp2);
\draw[->,densely dotted] (fpi) -- (fpp);
\path (isf) edge [->,densely dotted,bend angle=15,bend left]  node {} (iff);
\path (fpi) edge [->,densely dotted,bend angle=15,bend left]  node {} (ffi);
\path (iss) edge [line width=4pt,gray!30,bend angle=5,bend left] node {} (ifs);
\path (iss) edge [->,densely dotted,bend angle=5,bend left] node {} (ifs);
\path (ppi) edge [line width=4pt,gray!30,bend angle=5,bend right] node {} (pfi);
\path (ppi) edge [->,densely dotted,bend angle=5,bend right] node {} (pfi);
\end{scope}

\begin{scope}[every node/.style={minimum width=.65cm,minimum height=0.65cm,font=\scriptsize\strut},xshift=6.3cm,yshift=-3.5cm]
\fill[white!15,rounded corners=2pt] (-1.95,-2.02) rectangle (1.1,1.1);
\node (ifi) at (0,0)              {$\fin\ffac\fin$};
\node (sfi)  [right=0.6em of ifi] {$\fsuf\ffac\fin$};
\node (pfi)  [left =0.6em of ifi] {$\fpre\ffac\fin$};
\node (ifs)  [above=0.6em of ifi] {$\fin\ffac\fsuf$};
\node (ifp)  [below=0.6em of ifi] {$\fin\ffac\fpre$};
\node (sfs)  [right=0.6em of ifs] {$\fsuf\ffac\fsuf$};
\node (pfs0)[left =0.6em of ifs] {$\fpre\ffac\fsuf$};
\draw[fill=gray!30] (pfs0) ellipse (1em and 0.5em);
\node (pfs) [left =0.6em of ifs] {$\fpre\ffac\fsuf$};
\node (sfp)  [right=0.6em of ifp] {$\fsuf\ffac\fpre$};
\node (pfp)  [left =0.6em of ifp] {$\fpre\ffac\fpre$};
\node (ffs)  [left =0.6em of pfs] {$\ffac\ffac\fsuf$};
\node (ffi)  [left =0.6em of pfi] {$\ffac\ffac\fin$};
\node (ffp)  [left =0.6em of pfp] {$\ffac\ffac\fpre$};
\node (sff)  [below=0.6em of sfp] {$\fsuf\ffac\ffac$};
\node (iff)  [below=0.6em of ifp] {$\fin\ffac\ffac$};
\node (pff)  [below=0.6em of pfp] {$\fpre\ffac\ffac$};
\node (fff)  [left =0.6em of pff] {$\ffac\ffac\ffac$};
\node (pff1) [above=0.6em of pfs] {$\fpre\ffac\ffac$};
\node (iff1) [above=0.6em of ifs] {$\fin\ffac\ffac$};
\node (sff1) [above=0.6em of sfs] {$\fsuf\ffac\ffac$};
\node (ffp1) [right=0.6em of sfp] {$\ffac\ffac\fpre$};
\node (ffi1) [right=0.6em of sfi] {$\ffac\ffac\fin$};
\node (ffs1) [right=0.6em of sfs] {$\ffac\ffac\fsuf$};
\draw[line width=4pt,gray!30] (ifs) -- (pfs);
\draw[line width=4pt,gray!30] (pfi) -- (pfs);
\draw[->] (ifi) -- (sfi);
\draw[->] (ifi) -- (pfi);
\draw[->] (ifi) -- (ifs);
\draw[->] (ifi) -- (ifp);
\draw[->] (ifs) -- (sfs);
\draw[->] (ifs) -- (pfs);
\draw[->] (ifp) -- (sfp);
\draw[->] (ifp) -- (pfp);
\draw[->] (pfi) -- (pfs);
\draw[->] (pfi) -- (pfp);
\draw[->] (sfi) -- (sfs);
\draw[->] (sfi) -- (sfp);
\draw[->] (pfs) -- (ffs);
\draw[->] (pfi) -- (ffi);
\draw[->] (pfp) -- (ffp);
\draw[->] (sfp) -- (sff);
\draw[->] (ifp) -- (iff);
\draw[->] (pfp) -- (pff);
\draw[->] (ffi) -- (ffs);
\draw[->] (ffi) -- (ffp);
\draw[->] (iff) -- (sff);
\draw[->] (iff) -- (pff);
\draw[->] (ffp) -- (fff);
\draw[->] (pff) -- (fff);
\draw[->,densely dotted] (pfs) -- (pff1);
\draw[->,densely dotted] (ifs) -- (iff1);
\draw[->,densely dotted] (sfs) -- (sff1);
\draw[->,densely dotted] (sfs) -- (ffs1);
\draw[->,densely dotted] (sfi) -- (ffi1);
\draw[->,densely dotted] (sfp) -- (ffp1);
\path (iff) edge [->,bend angle=15,bend right] node {} (fff);
\path (sff) edge [->,bend angle=15,bend left ] node {} (fff);
\path (ffi) edge [->,bend angle=15,bend right] node {} (fff);
\path (ffs) edge [->,bend angle=15,bend left ] node {} (fff);
\draw[->,densely dotted] ($(ifi.center)-(0.15,0.3)$) -- ($(ifi.center)-(0.1,0.1)$);
\draw[->,densely dotted] ($(sfs.center)-(0.15,0.3)$) -- ($(sfs.center)-(0.1,0.1)$);
\draw[->,densely dotted] ($(pfp.center)-(0.15,0.3)$) -- ($(pfp.center)-(0.1,0.1)$);
\draw[->,densely dotted] ($(iff.center)-(0.3,0.15)$) -- ($(iff.center)-(0.1,0.1)$);
\draw[->,densely dotted] ($(ffi.center)-(0.15,0.3)$) -- ($(ffi.center)-(0.1,0.1)$);
\draw[->,densely dotted] ($(ffp.center)-(0.1,0.3)$) -- ($(ffp.center)-(0.1,0.1)$);
\draw[line width=4pt,gray!30] ($(ifs.center)-(0.15,0.3)$) -- ($(ifs.center)-(0.1,0.1)$);
\draw[->,densely dotted] ($(ifs.center)-(0.15,0.3)$) -- ($(ifs.center)-(0.1,0.1)$);
\draw[line width=4pt,gray!30] ($(pfi.center)-(0.15,0.3)$) -- ($(pfi.center)-(0.1,0.1)$);
\draw[->,densely dotted] ($(pfi.center)-(0.15,0.3)$) -- ($(pfi.center)-(0.1,0.1)$);
\draw[->,densely dotted] ($(sff.center)-(0.3,0.15)$) -- ($(sff.center)-(0.1,0.1)$);
\draw[->,densely dotted] ($(ifp.center)-(0.15,0.3)$) -- ($(ifp.center)-(0.1,0.1)$);
\draw[->,densely dotted] ($(sfi.center)-(0.15,0.3)$) -- ($(sfi.center)-(0.1,0.1)$);
\end{scope}

\begin{scope}[every node/.style={minimum width=.65cm,minimum height=0.65cm,font=\scriptsize\strut},xshift=10.5cm,yshift=-3.5cm]
\fill[white!15,rounded corners=2pt] (-0.22,-1.02) -- (-0.22,-0.18) -- (-1.1,-0.18) -- (-1.1,1.03) -- (0.2,1.03) -- (0.2,0.13) -- (1.08,0.13) -- (1.08,-1.02) -- cycle;
\node (sip) at (0,0)              {$\fsuf\fin\fpre$};
\node (ssp)  [above=0.6em of sip] {$\fsuf\fsuf\fpre$};
\node (spp)  [right=0.6em of sip] {$\fsuf\fpre\fpre$};
\node (sss)  [left =0.6em of sip] {$\fsuf\fsuf\fsuf$};
\node (ppp)  [below=0.6em of sip] {$\fpre\fpre\fpre$};
\node (ssf)  [left =0.6em of ssp] {$\fsuf\fsuf\ffac$};
\node (fpp)  [below=0.6em of spp] {$\ffac\fpre\fpre$};
\node (sfp)  [above=0.6em of spp] {$\fsuf\ffac\fpre$};
\node (sff)  [left =0.6em of ssf] {$\fsuf\ffac\ffac$};
\node (sfs)  [left =0.6em of sss] {$\fsuf\ffac\fsuf$};
\node (ffp)  [below=0.6em of fpp] {$\ffac\ffac\fpre$};
\node (pfp)  [below=0.6em of ppp] {$\fpre\ffac\fpre$};
\draw[->] (sip) -- (ssp);
\draw[->] (sip) -- (spp);
\draw[->] (ssp) -- (ssf);
\draw[->] (spp) -- (fpp);
\draw[->] (sss) -- (ssf);
\draw[->] (ppp) -- (fpp);
\draw[->,densely dotted] (ssp) -- (sfp);
\draw[->,densely dotted] (spp) -- (sfp);
\draw[->,densely dotted] (ssf) -- (sff);
\draw[->,densely dotted] (sss) -- (sfs);
\draw[->,densely dotted] (fpp) -- (ffp);
\draw[->,densely dotted] (ppp) -- (pfp);
\draw[->,densely dotted] ($(sip.center)-(0.25,0.25)$) -- ($(sip.center)-(0.1,0.1)$);
\draw[->,densely dotted] ($(sss.center)-(0.25,0.25)$) -- ($(sss.center)-(0.1,0.1)$);
\draw[->,densely dotted] ($(ppp.center)-(0.25,0.25)$) -- ($(ppp.center)-(0.1,0.1)$);
\end{scope}

\end{tikzpicture}
\caption{\label{fig:triples_map}
Map of the $61$ feasible types where Parts~(A), (B), (C) and~(D)
resp. correspond to triples
with three ``$\cout$'', two ``$\cout$'', one ``$\cout$'' and no ``$\cout$'',
where $\fout$, $\fin$, $\fpre$, $\fsuf$, $\ffac$ resp. are abbreviations for
``$\cout$'', ``$\cin$'', ``$\cpre$'', ``$\csuf$'', ``$\cfac$'';
there is an arc from a triple $t_1$ to a triple $t_2$ iff
(1)~$t_1$ and $t_2$ have the same signature,
(2)~$t_1$ and $t_2$ differ exactly from one position,
(3)~$t_1$ is lexicographically less than $t_2$ assuming
$\fin\prec\fpre$, $\fin\prec\fsuf$, $\fpre\prec\ffac$ and $\fsuf\prec\ffac$;
ellipses denote the types of the $\DecreasingSequencePatternName$ pattern
as described in Example~\ref{ex:word_types}.
}
\end{figure}

\noindent\textbf{Notation}~\emph{In the context of Definition~\ref{def:word_triple},
when $w_1 w_2\dots w_j$ is of type $\cpre$ or $\cin$,
i.e.~it contains a maximal occurrence of a word $x$ in $\Language{\pattern}$ starting at index $1$,
$\finf$ denotes the index of the last letter of $x$.
Similarly, when $w_i w_{i+1}\dots w_k$ is of type $\csuf$ or $\cin$,
i.e.~it contains a maximal occurrence of a word $y$ in $\Language{\pattern}$ ending at index $k$,
$\startf$ denotes the index of the first letter of $y$.}\\

The number of triples being important, $61$ in our case,
we reduce the number of cases to be considered in our proofs,
by introducing Definition~\ref{def:gen_triple} which groups a certain number of triples
in the same class representing the weakest hypothesis associated with the different triples
in this class.
Consider the (finite) set $\mathcal{S}$ of triples
associated with the words of $\Language{\pattern}$ wrt their proper factors.
We partition the set $\mathcal{S}$ into subsets where
all triples of the same subset have the same signature. 
Then, we generalise all the triples that belong to the same subset to a unique representative
using the following definition.

\begin{definition}
\label{def:gen_triple}\emph{[generalising a set of triples]}
Given a pattern $\pattern$, consider the set of triples $\mathcal{S}$
consisting of all types of the words of
$\Language{\pattern}$ wrt their proper factors and wrt $\pattern$
that have the same signature. 
Let $\mathcal{S}_c$ (with $c\in[1,3]$) denote
the set of all the $c^{\rm th}$ components of the triples of $\mathcal{S}$.
The set $\mathcal{S}$ is represented
by a single \emph{representative triple} $\mathcal{R}_\mathcal{S}=\langle r_1,r_2,r_3\rangle$
where $r_c$, with $c\in[1,3]$, is defined by
\begin{itemize}[leftmargin=1.5cm]
\item[$\bullet$]
{\makebox[6cm][l]{$\mathcal{S}_c=\{\cout\}$} $~\Rightarrow r_c=\gout$}
\item[$\bullet$]
{\makebox[6cm][l]{$\cfac\hspace*{1pt}\in\mathcal{S}_c$} $~\Rightarrow r_c=\gfac$}
\item[$\bullet$]
{\makebox[6cm][l]{$\cpre\in\mathcal{S}_c\hspace*{1pt}\land\csuf\notin\mathcal{S}_c\land\cfac\notin\mathcal{S}_c$} $~\Rightarrow r_c=\gpre$}
\item[$\bullet$]
{\makebox[6cm][l]{$\csuf\hspace*{1.2pt}\in\mathcal{S}_c\land\cpre\notin\mathcal{S}_c\land\cfac\notin\mathcal{S}_c$} $~\Rightarrow r_c=\gsuf$}
\item[$\bullet$]
{\makebox[6cm][l]{$\cpre\in\mathcal{S}_c\hspace*{1pt}\land\csuf\in\mathcal{S}_c\land\cfac\notin\mathcal{S}_c$} $~\Rightarrow r_c=\gps$}
\item[$\bullet$]
{\makebox[6cm][l]{$\cin\hspace*{6pt}\in\mathcal{S}_c\land\cpre\notin\mathcal{S}_c\land\csuf\notin\mathcal{S}_c\land\cfac\notin\mathcal{S}_c$} $~\Rightarrow r_c=\gin$}
\end{itemize}
\end{definition}

\begin{definition}
\label{def:class}\emph{[pattern class]}
Given a pattern $\pattern$, the set of representative triples of $\pattern$ is called
the \emph{class} of $\pattern$.
\end{definition}

\begin{sloppypar}
\begin{example}\label{ex:word_types}
The set $\mathcal{S}$ of possible types associated with the $\DecreasingSequencePatternName$ pattern
$\langle\DecreasingSequencePattern,0,0\rangle$ is equal to the union of two subsets
$\mathcal{S}_1=\{
   \langle \cpre,\cfac,\csuf\rangle,$
  $\langle \cpre,\cpre,\cin\rangle,$
  $\langle \cin,\csuf,\csuf\rangle,$
  $\langle \cin,\cin,\cin\rangle\}$
and
$\mathcal{S}_2=\{
   \langle \cpre,\cout,\csuf\rangle\}$,
where each subset corresponds to triples for which
all ``$\cout$'' are located in the same positions, see the five ellipses in Parts~(D) and~(C) of Figure~\ref{fig:triples_map}.
Part~(A) of Figure~\ref{fig:example_triples} gives for each
element of $\mathcal{S}_1$ and $\mathcal{S}_2$
a corresponding example of a word and a proper factor.
The sets $\mathcal{S}_1$ and $\mathcal{S}_2$
are respectively represented by the triples
$\langle\gpre,\gfac,\gsuf\rangle$ and
$\langle\gpre,\gout,\gsuf\rangle$ as
shown in Part~(B) of Figure~\ref{fig:example_triples}.
Finally, Figure~\ref{fig:factor_cartography} provides the representative triples for all
reversible and convex patterns of Table~\ref{tab:fgp},
which do not have the single letter property.
\end{example}
\end{sloppypar}
\begin{figure}[!h]
\centering
\begin{tikzpicture}
\begin{scope}
\node (a) at (0,0.4) {(A)};
\node (b) at (0,1.7) {(B)};
\node (c) at (0,-1.82) {(C)};
\end{scope}
\begin{scope}[xshift=0.6cm]
\fill[gray!20] (1.08, -0.25) rectangle (2.12,1);
\draw[line width=0.7pt,densely dotted] (1.08,-0.3) -- (1.08,1.05);
\draw[line width=0.7pt,densely dotted] (2.12,-0.3) -- (2.12,1.05);
\node (1a) at (0.25,0.4) {$\left\langle\begin{array}{c}\cpre,\\ \cfac,\\ \csuf\end{array}\right\rangle$};
\node[anchor=west] (1aa) at (0,-0.3) {\footnotesize\ding{172}};
\node (start) at (1.45,-0.6) {$\startf$};
\node (end)   at (1.72,1.35) {$\finf$};
\draw[line width=0.4pt,->] (1.45,-0.48) -- (1.45,-0.28);
\draw[line width=0.4pt,->] (1.72,1.2) -- (1.72,1);
\node (1b) at (1.6,0.8) {$>=>>=>$};
\node (1c) at (1.6,0.4) {$>=>>=>$};
\node (1d) at (1.6,0.0) {$>=>>=>$};
\draw[line width=0.7pt,->] (0.85,0.65) -- (2.12,0.65);
\draw[line width=0.7pt,->] (1.08,0.2) -- (2.12,0.2);
\draw[line width=0.7pt,->] (1.08,-0.25) -- (2.35,-0.25);
\end{scope}
\begin{scope}[xshift=4.4cm,yshift=0.2cm]
\node (0) at (0.4,1.5) {$\overbrace{\text{\hspace*{248pt}}}^{\langle\gpre,\gfac,\gsuf\rangle}$};
\end{scope}
\begin{scope}[xshift=3.1cm]
\fill[gray!20] (1.2,-0.25) rectangle (1.73,1);
\draw[line width=0.7pt,densely dotted] (1.20,-0.3) -- (1.20,1.05);
\draw[line width=0.7pt,densely dotted] (1.75,-0.3) -- (1.75,1.05);
\node (2a) at (0.55,0.4) {$\left\langle\begin{array}{c}\cpre,\\ \cpre,\\ \cin\end{array}\right\rangle$};
\node (end1) at (1.35,1.35) {$\finf$};
\node (start1) at (1.35,-0.6) {$\startf$};
\node[anchor=west] (2aa) at (0.3,-0.3) {\footnotesize\ding{173}};
\draw[line width=0.4pt,->] (1.35,1.2) -- (1.35,1);
\draw[line width=0.4pt,->] (1.35,-0.48) -- (1.35,-0.28);
\node (2b) at (1.6,0.8) {$>=>$};
\node (2c) at (1.6,0.4) {$>=>$};
\node (2d) at (1.6,0.0) {$>=>$};
\draw[line width=0.7pt,->] (1.2,0.65) -- (1.75,0.65);
\draw[line width=0.7pt,->] (1.2,0.2) -- (1.75,0.2);
\draw[line width=0.7pt,->] (1.2,-0.25) -- (2,-0.25);
\end{scope}
\begin{scope}[xshift=5.3cm]
\fill[gray!20] (1.47,-0.25) rectangle (1.99,1);
\draw[line width=0.7pt,densely dotted] (1.47,-0.3) -- (1.47,1.05);
\draw[line width=0.7pt,densely dotted] (1.99,-0.3) -- (1.99,1.05);
\node (3a) at (0.55,0.4) {$\left\langle\begin{array}{c}\cin,\\ \csuf,\\ \csuf\end{array}\right\rangle$};
\node[anchor=west] (3aa) at (0.3,-0.3) {\footnotesize\ding{174}};
\node (end2) at (1.83,1.35) {$\finf$};
\node (start2) at (1.83,-0.6) {$\startf$};
\draw[line width=0.4pt,->] (1.83,1.2) -- (1.83,1);
\draw[line width=0.4pt,->] (1.83,-0.48) -- (1.83,-0.28);
\node (3b) at (1.6,0.8) {$>=>$};
\node (3c) at (1.6,0.4) {$>=>$};
\node (3d) at (1.6,0.0) {$>=>$};
\draw[line width=0.7pt,->] (1.2,0.65) -- (1.99,0.65);
\draw[line width=0.7pt,->] (1.47,0.2) -- (1.99,0.2);
\draw[line width=0.7pt,->] (1.47,-0.25) -- (2,-0.25);
\end{scope}
\begin{scope}[xshift=7.3cm]
\fill[gray!20] (1.3,-0.25) rectangle (1.6,1);
\draw[line width=0.7pt,densely dotted] (1.3,-0.3) -- (1.3,1.05);
\draw[line width=0.7pt,densely dotted] (1.6,-0.3) -- (1.6,1.05);
\node (4a) at (0.75,0.4) {$\left\langle\begin{array}{c}\cin,\\ \cin,\\ \cin\end{array}\right\rangle$};
\node[anchor=west] (4aa) at (0.5,-0.3) {\footnotesize\ding{175}};
\node (end4) at (1.45,1.35) {$\finf$};
\node (start4) at (1.45,-0.6) {$\startf$};
\draw[line width=0.4pt,->] (1.45,1.2) -- (1.45,1);
\draw[line width=0.4pt,->] (1.45,-0.48) -- (1.45,-0.28);
\node (4b) at (1.6,0.8) {$>>$};
\node (4c) at (1.6,0.4) {$>>$};
\node (4d) at (1.6,0.0) {$>>$};
\draw[line width=0.7pt,->] (1.3,0.65) -- (1.6,0.65);
\draw[line width=0.7pt,->] (1.3,0.2) -- (1.6,0.2);
\draw[line width=0.7pt,->] (1.3,-0.25) -- (1.85,-0.25);
\end{scope}
\begin{scope}[xshift=10.2cm,yshift=0.2cm]
\node (0) at (0.5,1.5) {$\overbrace{\text{\hspace*{53pt}}}^{\langle\gpre,\gout,\gsuf\rangle}$};
\end{scope}
\begin{scope}[xshift=9.7cm]
\fill[gray!20] (1.46,-0.25) rectangle (1.72,1);
\draw[line width=0.7pt,densely dotted] (1.46,-0.3) -- (1.46,1.05);
\draw[line width=0.7pt,densely dotted] (1.72,-0.3) -- (1.72,1.05);
\node (5a) at (0.55,0.4) {$\left\langle\begin{array}{c}\cpre,\\ \cout,\\ \csuf\end{array}\right\rangle$};
\node[anchor=west] (5aa) at (0.3,-0.3) {\footnotesize\ding{176}};
\node (5b) at (1.6,0.8) {$>=>$};
\node (5c) at (1.6,0.4) {$>=>$};
\node (5d) at (1.6,0.0) {$>=>$};
\draw[line width=0.7pt,->] (1.2,0.65) -- (1.72,0.65);
\draw[line width=0.7pt,->] (1.46,0.2) -- (1.72,0.2);
\draw[line width=0.7pt,->] (1.46,-0.25) -- (1.97,-0.25);
\end{scope}
\begin{scope}[xshift=0.26cm,yshift=-2cm]
\fill[gray!20] (2.6, 0.7) rectangle (4.65,1.08);
\fill[gray!20] (2.97, 0.25) rectangle (4.36,0.62);
\fill[gray!20] (2.57, -0.2) rectangle (4.0,0.17);
\fill[gray!20] (0.96, -0.65) rectangle (2.95,-0.29);
\fill[gray!20] (7.5, -0.65) rectangle (9.5,-0.29);
\fill[gray!20] (2.57, -1.1) rectangle (3.25,-0.75);
\setlength{\extrarowheight}{2pt}
\node[anchor=west] (c1) at (0,0) {
$\left\{\begin{array}{l}
 \text{\ding{172}}~>(=|>)^*s(=^+>^+)^+=^+s=^*>^+(=^+>^+)^*\\
 \text{\ding{173}}~(>(=|>)^*|\epsilon)s(>^+=^+)^+s=^*>^+(=^+>^+)^+\\
 \text{\ding{174}}~>(=|>)^*s(=^+>^+)^+s>^*(=^+>^+)^*\\
 \text{\ding{175}}~s>^+(=^+>^+)^*s=^*>^+(=^+>^+)^*\hspace*{2pt}|\hspace*{2pt}>(=|>)^*s>^+(=^+>^+)^*s>^*(=^+>^+)^*\\
 \text{\ding{176}}~>(=|>)^*s=^+s=^*>^+(=^+>^+)^*
 \end{array}\right.$
};
\end{scope}
\end{tikzpicture}
\caption{\label{fig:example_triples}
(A)~Set of possible types of words wrt their proper factors
of the $\DecreasingSequencePatternName$ pattern
with the corresponding examples of word $\word$ and proper factor (in grey)
where $\finf$ (resp.~$\startf$) denotes the end (resp.~start) of a maximal word
in $\Language{\DecreasingSequencePatternName}$
starting at the first position (resp. ending at the last position) of $\word$,
(B)~corresponding set of representative triples, and
(C)~languages of the types of 
words \ding{172}, \ding{173}, \ding{174}, \ding{175}, \ding{176} as computed from Equation~(\ref{eq:representative_language}) of Theorem~\ref{theorem:representative_language}.}
\end{figure}
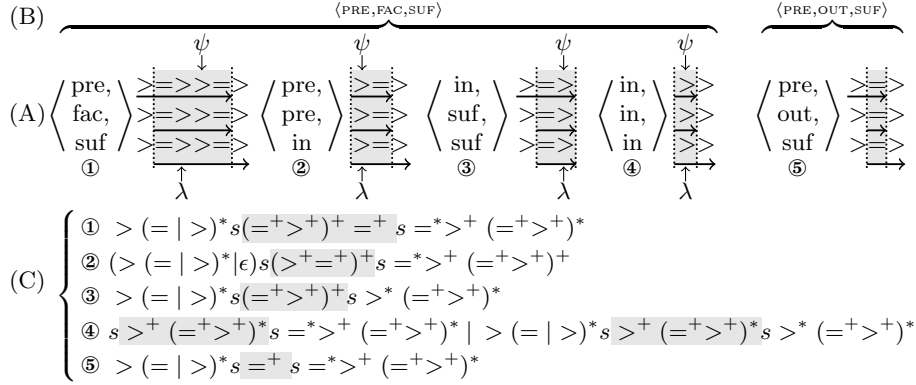
\vspace{-0.6cm}
\subsubsection{Finding all the representatives of a pattern}\label{sec:finding_representatives}
To generate all the representatives of a pattern $\pattern$, 
\emph{(i)}~we first generate all potential word types wrt $\pattern$, and
\emph{(ii)}~we then use Definition~\ref{def:gen_triple}.
For each of the $61$ potential word type $\Tuple{t_1,t_2,t_3}$ with $t_i\in\{\cout,\cfac,\cpre,\csuf,\cin\}$
depicted in Figure~\ref{fig:triples_map},
we describe a systematic method to check whether there exists or not a word $w=w_1 w_2\dots w_k$ of
$\Language{\pattern}$ whose type is $\Tuple{t_1,t_2,t_3}$.
For this purpose we define the language of $\Tuple{t_1,t_2,t_3}$ wrt to $\pattern$ and check
whether it is empty or not.
Since we need the prefix of $w$ associated with $t_1$ to overlap the suffix of $w$ associated with $t_3$,
we first introduce the notion of \emph{shuffle language}.

\begin{definition}
\label{def:insertion_language}\emph{[shuffle language]}
Given a regular language $\Language{}$ over an input alphabet $\Sigma$,
and a possibly new input letter $s$, i.e.~a letter that does not necessarily belong to $\Sigma$,
the \emph{shuffle language of $\Language{}$ wrt $s$},
denoted $\shuffle(\Language{},s)$,
is defined by
all words $w$ over the alphabet $\Sigma\cup\{s\}$ such that
\begin{enumerate}[label=\roman*)]
\item $w$ contains at least one occurrence of the letter $s$,
\item if we remove one single occurrence of the letter $s$ from $w$
      then the resulting word belongs to $\Language{}$.
\end{enumerate}
\end{definition}

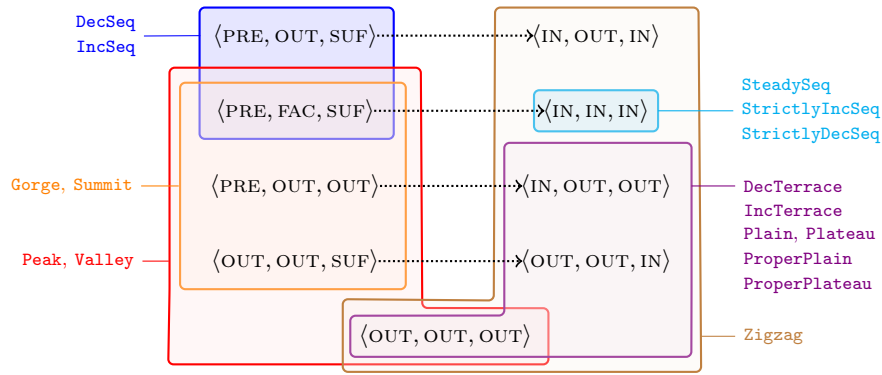
\begin{figure}[!b]
\centering
\begin{tikzpicture}
\begin{scope}
\node (pop) at (0,0)  {$\langle \gpre,\gout,\gsuf\rangle$};
\node (pfp) at (0,-1) {$\langle \gpre,\gfac,\gsuf\rangle$};
\node (poo) at (0,-2) {$\langle \gpre,\gout,\gout\rangle$};
\node (oop) at (0,-3) {$\langle \gout,\gout,\gsuf\rangle$};
\node (ooo) at (2,-4) {$\langle \gout,\gout,\gout\rangle$};
\node (ioi) at (4,0)  {$\langle \gin,\gout,\gin\rangle$};
\node (iii) at (4,-1) {$\langle \gin,\gin,\gin\rangle$};
\node (ioo) at (4,-2) {$\langle \gin,\gout,\gout\rangle$};
\node (ooi) at (4,-3) {$\langle \gout,\gout,\gin\rangle$};
\scoped[on background layer]
  \draw[rounded corners=2pt,thick,blue,fill=blue!10]
	($(pfp.south west)-(0.1,0.1)$) rectangle ($(pop.north east)+(0.1,0.1)$);
\scoped[on background layer]
  \draw[rounded corners=2pt,thick,orange,fill opacity=0.3,fill=orange!10]
    ($(oop.south west)-(0.3,0.1)$) rectangle ($(pfp.north east)+(0.3,0.1)$);
\scoped[on background layer]
  \draw[rounded corners=2pt,thick,red,fill opacity=0.3,fill=red!10]
	($(pfp.north west)+(-0.5,0.3)$) --
	($(oop.south west)+(-0.44,-0.1)$) --
	($(ooo.south west)+(-2.42,-0.1)$) --
	($(ooo.south east)+(0.1,-0.1)$) --
	($(ooo.north east)+(0.1,0.1)$) --
	($(oop.south east)+(0.46,-0.35)$) --
	($(pfp.north east)+(0.5,0.3)$) -- cycle;
\scoped[on background layer]
  \draw[rounded corners=2pt,thick,cyan,fill=cyan!10]
    (iii.south west) rectangle (iii.north east);
\scoped[on background layer]
  \draw[rounded corners=2pt,thick,violet,fill opacity=0.3,fill=violet!10]
	($(ioo.north east)+(0.15,0.3)$) --
	($(ooi.south east)+(0.15,-0.1)$) --
	($(ooo.south east)+(2,0)$) --
	($(ooo.south west)+(0,0)$) --
	($(ooo.north west)+(0,0)$) --
	($(ooi.south west)+(-0.1,-0.45)$) --
	($(ioo.north west)+(-0.1,0.3)$) -- cycle;
\scoped[on background layer]
  \draw[rounded corners=2pt,thick,brown,fill opacity=0.3,fill=brown!10]
	($(ioi.north east)+(0.39,0.1)$) --
	($(ooi.south east)+(0.27,-0.1)$) --
	($(ooo.south east)+(2.13,-0.2)$) --
	($(ooo.south west)+(-0.1,-0.2)$) --
	($(ooo.north west)+(-0.1,0.22)$) --
	($(ooi.south west)+(-0.24,-0.25)$) --
	($(ioi.north west)+(-0.37,0.1)$) -- cycle;
\draw[thick,densely dotted,->] ($(pop.east)-(0.12,0)$) -- ($(ioi.west)+(0.14,0)$);
\draw[thick,densely dotted,->] ($(pfp.east)-(0.12,0)$) -- ($(iii.west)+(0.14,0)$);
\draw[thick,densely dotted,->] ($(poo.east)-(0.12,0)$) -- ($(ioo.west)+(0.14,0)$);
\draw[thick,densely dotted,->] ($(oop.east)-(0.12,0)$) -- ($(ooi.west)+(0.14,0)$);
\node[anchor=east] (ctr1a) at ($(pop.west)-(0.8,0.165)$) {\scriptsize\color{blue}
                                                     $\IncreasingSequencePatternName$};
\draw[blue] ($(ctr1a.east)+(0,0.16)$) -- ($(pop.west)-(0.05,0)$);
\node[anchor=east] (ctr1b) at ($(pop.west)-(0.8,-0.165)$) {\scriptsize\color{blue}
                                                     $\DecreasingSequencePatternName$};
\node[anchor=east] (ctr2) at ($(poo.west)-(0.8,0)$) {\scriptsize\color{orange}
                                                     $\GorgePatternName$,
                                                     $\SummitPatternName$};
\draw[orange] ($(ctr2.east)$) -- ($(poo.west)-(0.3,0)$);
\node[anchor=east] (ctr3) at ($(oop.west)-(0.8,0)$) {\scriptsize\color{red}
                                                     $\PeakPatternName$,
                                                     $\ValleyPatternName$};
\draw[red] ($(ctr3.east)$) -- ($(oop.west)-(0.45,0)$);
                                                 
\node[anchor=west] (ctr4a) at ($(iii.east)+(1,0.33)$) {\scriptsize\color{cyan}
                                                     $\SteadySequencePatternName$};
\node[anchor=west] (ctr4b) at ($(iii.east)+(1,0)$) {\scriptsize\color{cyan}
                                                     $\StrictlyIncreasingSequencePatternName$};
\node[anchor=west] (ctr4c) at ($(iii.east)+(1,-0.33)$) {\scriptsize\color{cyan}
                                                     $\StrictlyDecreasingSequencePatternName$};
\draw[cyan] ($(ctr4b.west)+(0,0)$) -- ($(iii.east)+(0,0)$);

\node[anchor=west] (ctr5a) at ($(ioo.east)+(0.73,0)$) {\scriptsize\color{violet}
                                                     $\DecreasingTerracePatternName$};
\draw[violet] ($(ctr5a.west)+(0,0)$) -- ($(ioo.east)+(0.15,0)$);

\node[anchor=west] (ctr5a) at ($(ioo.east)+(0.73,-0.33)$) {\scriptsize\color{violet}
                                                     $\IncreasingTerracePatternName$};
\node[anchor=west] (ctr5b) at ($(ioo.east)+(0.73,-0.66)$) {\scriptsize\color{violet}
                                                     $\PlainPatternName$,
                                                     $\PlateauPatternName$};
\node[anchor=west] (ctr5c) at ($(ioo.east)+(0.73,-0.99)$) {\scriptsize\color{violet}
                                                     $\ProperPlainPatternName$};
\node[anchor=west] (ctr5d) at ($(ioo.east)+(0.73,-1.32)$) {\scriptsize\color{violet}
                                                     $\ProperPlateauPatternName$};

\node[anchor=west] (ctr6) at ($(ooi.east)+(0.73,-1)$) {\scriptsize\color{brown}
                                                     $\ZigzagPatternName$};
\draw[brown] ($(ctr6.west)+(0,0)$) -- ($(ooi.east)+(0.28,-1)$);
\end{scope}
\end{tikzpicture}
\caption{\label{fig:factor_cartography} Pattern classes, where each class corresponds to a set of representative triples (an arrow from a triple \ding{172} to a triple \ding{173} means that \ding{172} generalises \ding{173})}
\end{figure}

\begin{theorem}
\label{theorem:representative_language}\emph{[language of a word type]}
Given a pattern $\pattern$ and one of its potential word types $\Tuple{t_1,t_2,t_3}$,
the language associated with $\Tuple{t_1,t_2,t_3}$ is defined by
\begin{equation}
\bigcap\left(\begin{array}{c}
\shuffle(\shuffle(\Language{\pattern},s),s)\\
\shuffle(\Language{t_1},s)s\Sigma^*\\
\Sigma^*s\Language{t_2}s\Sigma^+|\Sigma^+s\Language{t_2}s\Sigma^*\\
\Sigma^*s~\shuffle(\Language{t_3},s)
\end{array}\right)\label{eq:representative_language}
\end{equation}
\end{theorem}
\begin{proof}
The four sub-expressions on the right-hand side of~(\ref{eq:representative_language})
respectively correspond to
a word of $\Language{\pattern}$ to which two occurrences of $s$ are inserted, and in three ways of decomposing
it wrt its prefix, to its window, and to its suffix.
The letter $s$ is used to ``synchronise''
these decompositions, i.e.~to enforce a non\nobreakdash-empty intersection between the prefix and the suffix.
Since $\Language{t_2}$ does not contain the empty word,
the two occurrences of $s$ delimit a non\nobreakdash-empty window.\qed
\end{proof}

\begin{example}
[Continuation of Example (\ref{ex:word_types})]
Part~(C) of Figure~\ref{fig:example_triples} gives the languages of the types of words
$\langle \cpre,\cfac,\csuf\rangle$, $\langle \cpre,\cpre,\cin\rangle$,
$\langle \cin,\csuf,\csuf\rangle$, $\langle \cin,\cin,\cin\rangle$, and
$\langle \cpre,\cout,\csuf\rangle$ for the $\DecreasingSequencePatternName$ pattern,
as defined by Theorem~\ref{theorem:representative_language}.
Note that all other triples lead to the empty language.
\end{example}

Evaluating whether the language associated with a regular expression is empty or not
(e.g.~Expression~(\ref{eq:representative_language})) is done by
\emph{(i)}~converting all its operator instances (e.g.~union, intersection, concatenation, Kleene star, shuffle, \dots) to deterministic finite automata, by
\emph{(ii)}~evaluating the corresponding sequence of operations
on finite automata, and by
\emph{(iii)}~checking whether the resulting minimised automaton has at least one accepting state or not.
Following this methodology, Appendix~\ref{sec:evaluating_pattern_properties} gives the corresponding programs which
compute the representatives and the properties of a pattern.
We now show how to generate a finite automaton for the $\shuffle$ operator
that we previously introduce.
\begin{sloppypar}
\begin{itemize}
\item\hspace*{0.0001pt}[$\shuffle$]
From the deterministic and minimised automaton $A_\mathcal{L}$
associated with $\Language{}$, one can build the automaton $A_{\shuffle(\mathcal{L},s)}$
associated with the language $\shuffle(\Language{},s)$ by
(\emph{i})~duplicating all states of $A_\mathcal{L}$ and make them non\nobreakdash-initial,
(\emph{ii})~make all states of $A_\mathcal{L}$ non\nobreakdash-accepting,
(\emph{iii})~add a transition labelled by $s$ from each state to its duplicated state.
\end{itemize}
\end{sloppypar}

\subsubsection{Establishing the properties of a pattern}\label{sec:finding_pattern_properties}
We now describe how to systematically find the properties of a pattern $\pattern$.
We use $\Language{\pattern}^\mathit{ssss}$ (resp.~$\Language{\pattern}^\mathit{ss}$)
as a shortcut for
$\shuffle(\shuffle(\shuffle(\shuffle(\Language{\pattern},s),s),s),s)$
(resp.~$\shuffle(\shuffle(\Language{\pattern},s),s)$).
\begin{sloppypar}
\begin{itemize}[label=\textbullet]
\item
A pattern $\pattern$ has the \ConvexityProp{} iff
\begin{equation}
\scriptstyle
       \bigcap\left(
       \begin{array}{c}
	   \Language{\pattern}^\mathit{ss}\\
	   \Sigma^*s\,\Language{\pattern}\Sigma^*\Language{\pattern}\,s\,\Sigma^*\\
	   \Sigma^*s\,(\Sigma^+\setminus\Language{\pattern})\,s\,\Sigma^*
       \end{array}\right)
       \hspace*{2pt}\bigcup\hspace*{4pt}
       \bigcap\left(
       \begin{array}{c}
       \Language{\pattern}^\mathit{ssss}\\
       \Sigma^*s\,\shuffle(\Language{\pattern},s)\,s\,\Sigma^+ s\,\Sigma^* \\
       \Sigma^*s\,\Sigma^+ s\,\shuffle(\Language{\pattern},s)\,s\,\Sigma^* \\
       \Sigma^*s\,\shuffle(\shuffle(\Sigma^+\setminus\Language{\pattern},s),s)\,s\,\Sigma^*
       \end{array}\right)
       =\emptyset
\label{eq_convexity_prop}
\end{equation}
\item A pattern $\pattern$ has the \InflexionFreeProp{} iff
\begin{equation}(\Language{\pattern}\cap\Sigma^*<\Sigma^*>\Sigma^*)\cup(\Language{\pattern}\cap\Sigma^*>\Sigma^*<\Sigma^*)=\emptyset\label{eq_inflexion_free_prop}
\end{equation}
\item A pattern $\pattern$ has the \OneInflexionProp{} iff
\begin{equation}\Language{\pattern}\setminus\Language{(<|=)^*<=^*>(>|=)^*|(>|=)^*>=^*<(<|=)^*}=\emptyset\label{eq_one_inflexion_prop}
\end{equation}
\item A pattern $\pattern$ has the \ExcludeOutInProp{} iff
\begin{equation}
 \scriptstyle
 \bigcap\left(
 \begin{array}{c}
 \Language{\pattern}^\mathit{ssss} \\
 \Sigma^* s \Sigma^+ s\,\shuffle(\Sigma^+ \setminus (\Sigma^* \Language{\pattern} \Sigma^*),s)\,s\,\Sigma^* \\
 \Sigma^* s\,\shuffle(\Language{\pattern},s)\,s\,\Sigma^*\,s\,\Sigma^* \\
 \Sigma^* s\,\Sigma^+ s\,\Sigma^+ s\,\Sigma^* s\,\Sigma^*
 \end{array}\right)
 \hspace*{2pt}\bigcup\hspace*{4pt}
 \bigcap\left(
 \begin{array}{c}
 \Language{\pattern}^\mathit{ssss} \\
 \Sigma^* s\,\shuffle(\Sigma^+ \setminus (\Sigma^* \Language{\pattern} \Sigma^*),s)\,s\,\Sigma^+\,s\,\Sigma^*\\
 \Sigma^* s\,\Sigma^* s\,\shuffle(\Language{\pattern},s)\,s\,\Sigma^* \\
 \Sigma^* s\,\Sigma^* s\,\Sigma^+ s\,\Sigma^+ s\,\Sigma^*
 \end{array}\right) =\emptyset
\label{eq_exclude_out_in_prop}
\end{equation}
\item A pattern $\pattern$ has the \SingleLetterProp{} iff
\begin{equation}\Language{\pattern}\setminus\Language{<|=|>}=\emptyset\label{eq_single_letter_prop}
\end{equation}
\end{itemize} 
\end{sloppypar}
As the constructions used
in~(\ref{eq_convexity_prop}),
(\ref{eq_inflexion_free_prop}),
(\ref{eq_one_inflexion_prop}),
(\ref{eq_exclude_out_in_prop}) and
(\ref{eq_single_letter_prop})
are similar to the one used in Theorem~\ref{theorem:representative_language},
they are not detailed.

\subsubsection{Proof of Equations Based on Pattern and Feature Properties.}\label{subsec:proofs}

In this section we study the properties of patterns and features
that ensure the validity of Equations~\eqref{formula1}, \eqref{formula2}
and \eqref{formula3}.
While pattern properties were already introduced in Sections~\ref{sec:pattern_prop} and~\ref{sec:classify_pattern},
we first present some feature properties.
Second, we focus on the validity domain of Equation~\eqref{formula1},
and finally, based on these results, we derive the properties of the patterns
and features for Equations~\eqref{formula2} and~\eqref{formula3}.
From now on we focus on commutative features, as well as reversible and
convex patterns, which do not have the \SingleLetterProp{}.

\subsubsection{Feature Properties}
All definitions of this section,
i.e.~Definitions~\ref{pf-sum_decomp} to~\ref{pf-single_pos_inflexion}, as well as all theorems of this section,
i.e.~Theorems~\ref{theorem:E1_P1} to~\ref{theorem:E3}, consider
\emph{(i)}~a reversible and convex pattern
$\pattern=\langle\Language{\pattern},b_\pattern,a_\pattern\rangle$,
\emph{(ii)}~a sequence of variables $\mathcal{X}=x_1 x_2\dots x_\seqlength$,
\emph{(iii)}~an extended $\pattern$-pattern occurrence in $[1,\seqlength]$ given by 
$o=\Tuple{x_\ell x_{\ell+1}\dots x_u}$
with $1\leq \ell \leq u\leq \seqlength$, and
\emph{(iv)}~a commutative feature $f$ applied to $o$.

We first present four feature properties that only depend on the feature $\feature$.
We then introduce two additional feature properties that depend on both
the feature $\feature$ and the pattern $\pattern$.
Finally, Part~(A) of Figure~\ref{fig:feature_properties_theorems_cartography}
summarises the feature properties of each of the features defined
in the time-series catalogue~\cite{arafailova2016global}.

\begin{definition}
\label{pf-sum_decomp}
A feature $\feature$ has the \SumDecompositionProp{} if
$\feature_\pattern(x_{\ell,u})$ can be expressed as $\sum\limits_{\lTotal}^{\uTotal}{h(x_t)}$, 
where $h(x_t)$ is a function. 
\end{definition}
E.g., when $\feature=\Width$, $h(x_t)=1$ 
and the value returned by the application of $\feature$ to the 
extended $\pattern$-pattern occurrence $o$ is $u-\ell-b_\pattern-a_\pattern+1$.

\begin{definition}
\label{pf-same_value} 
A feature $\feature$ has the \SameValueProp{} if 
$\feature_\pattern(x_{\ell,u})=\feature_\pattern(x_{i,j})$ for all $i,j$ 
($\ell\leq i \leq j\leq u$) such that the sequence $x_{i,j}$ alone is an extended $\pattern$\nobreakdash-pattern occurrence.
\end{definition}
E.g., when $\feature=\One$ and $x_{i,j}$ is an extended $\pattern$\nobreakdash-pattern,  
$\feature_\pattern(x_{\ell,u})=\feature_\pattern(x_{i,j})=1$.

\begin{definition}
\label{pf-one_value} 
A feature $\feature$ has the \SinglePositionProp{} if
$\feature_\pattern(x_{\ell,u})$ can be expressed as $h(x_t)$ 
with $x_t\in\{x_{\ell+b_\pattern}, x_{\ell+b_\pattern+1}, ..., x_{u-a_\pattern}\}$. 
\end{definition}
E.g., when $\feature=\MaxFeature$, $h(x_t)=x_t$ 
and $\feature_\pattern(x_{\ell,u})$ is the maximum of the variables 
in~$x_{\ell+b_\pattern,u-a_\pattern}$.

\begin{definition}
\label{pf-positive} 
A feature $\feature$ has the \PositiveProp{} if 
$\feature_\pattern(x_{i,j})\geq 0~\forall~i,j$, 
such that $1\leq i\leq j\leq \seqlength$.
\end{definition}

\begin{sloppypar}
\begin{definition}
\label{pf-single_pos_inflexion_free} 
A feature $\feature$ and a pattern $\pattern$
have the \SinglePositionInflexionFreeProp{} if
(i)~$\feature$ has the \SinglePositionProp{}, 
(ii)~$\pattern$ has the \InflexionFreeProp{} and either
(iii.a)~for all extended $\pattern$\nobreakdash-pattern occurrences
$x_{p,q}$ wrt $x_{p,q}$ (with $\ell\leq p\leq q\leq u$)
$\feature_\pattern(x_{p,q})=h(x_{p+b_\pattern})$, or
(iii.b)~for all extended $\pattern$\nobreakdash-pattern occurrences
$x_{p,q}$ wrt $x_{p,q}$ (with $\ell\leq p\leq q\leq u$)
$\feature_\pattern(x_{p,q})=h(x_{q-a_\pattern})$.
\end{definition}
E.g., the pair $\pattern= \DecreasingSequencePatternName$, $\feature=\MinFeature$, 
has the \SinglePositionInflexionFreeProp{} 
since $\feature_\pattern(x_{\ell,u})=x_{q-a_\pattern}$,
where $q$ is the end of the extended $\pattern$-pattern occurrence in $x_{\ell,u}$.
\end{sloppypar}

\begin{definition}
\label{pf-single_pos_inflexion}
A feature $\feature$ and a pattern $\pattern$ have the 
\SinglePositionInflexionProp{} if
(i)~$\feature$ has the \SinglePositionProp{}, 
(ii)~$\pattern$ has the \OneInflexionProp{} and 
(iii)~$\feature_\pattern(x_{\ell,u})$ is computed from the position of the 
only inflexion of $\pattern$. 
\end{definition}
E.g., the pair $\pattern= \GorgePatternName$, $\feature=\MinFeature$ has the \SinglePositionInflexionProp{} 
since the value of $\feature_\pattern(x_{\ell,u})$ corresponds to the only inflexion of the extended $\pattern$\nobreakdash-pattern occurrence. 
But the pair $\pattern= \GorgePatternName$, $\feature=\MaxFeature$, 
does not have the \SinglePositionInflexionProp{} since $\feature_\pattern(x_{\ell,u})$
corresponds to one of the two extremities of the gorge. 

The next two sections define sufficient conditions where
\emph{(i)}~Equation \eqref{formula1} and
\emph{(ii)}~Equations \eqref{formula2} and \eqref{formula3}
can be used to compute the value of $\featureTw$ wrt a pattern $\pattern$,
depending on the representatives of a pattern.
Part~(B) of Figure~\ref{fig:feature_properties_theorems_cartography} summarises all the theorems introduced in these two sections
wrt the representatives of Figure~\ref{fig:factor_cartography}.

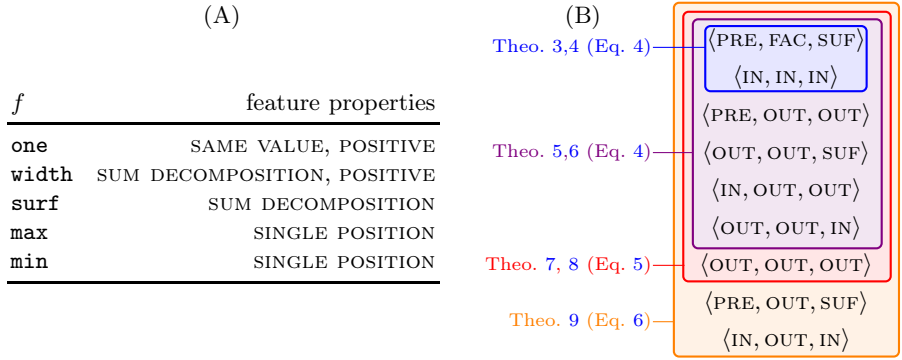
\begin{figure}[!t]
\centering
\begin{tikzpicture}[information text/.style={rounded corners,inner sep=1ex}]
\begin{scope}
\draw node[right,text width=9.6cm,information text]{
   	\begin{tabular}{lr}
     	 $\feature$~~      & $\textnormal{feature properties}$ \\ \midrule
      	 $\One$~~          & \SameValue{},~\Positive{}         \\
      	 $\Width$~~        & \SumDecomposition{},~\Positive{}  \\
	     $\Surf$~~         & \SumDecomposition{}               \\
     	 $\MaxFeature$~~   & \SinglePosition{}                 \\
     	 $\MinFeature$~~   & \SinglePosition{}                 \\ \midrule                        
    \end{tabular}
};
\end{scope}
\begin{scope}[xshift=10.5cm,yshift=2cm]
\node (pfp) at (0,0)    {$\langle \gpre,\gfac,\gsuf\rangle$};
\node (iii) at (0,-0.5) {$\langle \gin,\gin,\gin\rangle$};
\node (poo) at (0,-1)   {$\langle \gpre,\gout,\gout\rangle$};
\node (oop) at (0,-1.5) {$\langle \gout,\gout,\gsuf\rangle$};
\node (ioo) at (0,-2)   {$\langle \gin,\gout,\gout\rangle$};
\node (ooi) at (0,-2.5) {$\langle \gout,\gout,\gin\rangle$};
\node (ooo) at (0,-3)   {$\langle \gout,\gout,\gout\rangle$};
\node (pop) at (0,-3.5) {$\langle \gpre,\gout,\gsuf\rangle$};
\node (ioi) at (0,-4)   {$\langle \gin,\gout,\gin\rangle$};
\scoped[on background layer]
  \draw[rounded corners=2pt,thick,orange,fill=orange!10]
	($(ioi.south west)-(0.52,-0.06)$) rectangle ($(pfp.north east)+(0.33,0.22)$);
\scoped[on background layer]
  \draw[rounded corners=2pt,thick,red,fill=red!10]
	($(ooo.south west)-(0.1,-0.06)$) rectangle ($(pfp.north east)+(0.22,0.1)$);
\scoped[on background layer]
  \draw[rounded corners=2pt,thick,violet,fill=violet!10]
	($(ooi.south west)-(0.12,0)$) rectangle ($(pfp.north east)+(0.1,0)$);
\scoped[on background layer]
  \draw[rounded corners=2pt,thick,blue,fill=blue!10]
	($(iii.south west)-(0.25,-0.09)$) rectangle ($(pfp.north east)+(-0.1,-0.09)$);
\node[anchor=east] (theo_E1_P1) at ($(pfp.west)-(0.5,0.1)$) {\scriptsize\color{blue}
           Theo.~\ref{theorem:E1_P1},\ref{theorem:E1_P3_IF} (Eq.~\ref{formula1})};
\node[anchor=east] (theo_E1_P2) at ($(pfp.west)-(0.5,1.5)$) {\scriptsize\color{violet}
                 Theo.~\ref{theorem:E1_P2},\ref{theorem:E1_P3} (Eq.~\ref{formula1})};
\node[anchor=east] (theo_E2_P1_P) at ($(pfp.west)-(0.5,3)$) {\scriptsize\color{red}
                 Theo.~\ref{theorem:E2_P1_P}, \ref{theorem:E2_P2_P} (Eq.~\ref{formula2})};
\node[anchor=east] (theo_E3) at ($(pfp.west)-(0.5,3.75)$) {\scriptsize\color{orange}
                 Theo.~\ref{theorem:E3} (Eq.~\ref{formula3})};
\draw[blue] ($(pfp.west)-(0.6,0.1)$) -- ($(pfp.west)-(-0.1,0.1)$);
\draw[violet] ($(pfp.west)-(0.6,1.5)$) -- ($(pfp.west)-(0.08,1.5)$);
\draw[red] ($(pfp.west)-(0.6,3)$) -- ($(pfp.west)-(0.18,3)$);
\draw[orange] ($(pfp.west)-(0.6,3.75)$) -- ($(pfp.west)-(0.32,3.75)$);
\end{scope}
\begin{scope}[xshift=3cm,yshift=-2cm]
\node (A) at (0,4.3) {(A)};
\node (B) at (4.8,4.3) {(B)};
\end{scope}
\end{tikzpicture}
\caption{\label{fig:feature_properties_theorems_cartography}
(A)~Properties of the features defined in Table~\ref{tab:fgp} and used in~\cite{arafailova2016global},
(B)~theorems coverage for the different representatives of Figure~\ref{fig:factor_cartography}.}
\end{figure}

\begin{sloppypar}
\begin{remark}\label{rem:symmetries}
Wlog, while doing the proof of such conditions we proceed as follows:
\begin{itemize}
\item
When the representatives $\Tuple{\gpre,\gfac,\gsuf}$ and $\Tuple{\gin,\gin,\gin}$
are both present, only $\Tuple{\gpre,\gfac,\gsuf}$ is considered,
since for $\Tuple{\gin,\gin,\gin}$ $\startf=i$ and $\finf=j$
is a special case of $\Tuple{\gpre,\gfac,\gsuf}$.
\item
Similarly, when the representatives $\Tuple{\gpre,\gout,\gout}$ and $\Tuple{\gin,\gout,\gout}$ 
(resp.~$\Tuple{\gout,\gout,\gsuf}$ and $\Tuple{\gout,\gout,\gin}$) both intervene in a proof, 
only $\Tuple{\gpre,\gout,\gout}$ (resp.~$\Tuple{\gout,\gout,\gsuf}$) is considered, 
as  $\finf=j$ (resp.~$\startf=i$).
\item
When $\Tuple{\gin,\gout,\gout}$ and $\Tuple{\gout,\gout,\gin}$ 
(resp.~$\Tuple{\gpre,\gout,\gout}$ and $\Tuple{\gout,\gout,\gsuf}$) both intervene in a proof,
only $\Tuple{\gin,\gout,\gout}$ (resp.~$\Tuple{\gpre,\gout,\gout}$) is considered, 
as the representative $\Tuple{\gout,\gout,\gin}$ (resp.~$\Tuple{\gout,\gout,\gsuf}$) is symmetric.
\end{itemize}
\end{remark}
\end{sloppypar}

\subsubsection{Sufficient Conditions for the Validity of Equation \eqref{formula1}}

\begin{theorem}
\label{theorem:E1_P1}
Consider a pattern $\pattern$ whose class has a non-empty intersection with the set of representatives  
$\mathcal{S}=\{\Tuple{\gpre,\gfac,\gsuf}$, $\Tuple{\gin,\gin,\gin}\}$.
Equation~\eqref{formula1} can be used to obtain $\featureTw$  
for a sequence $x_{1,n}$ wrt a window $\tw$ whose type is in $\mathcal{S}$, 
assuming that feature $\feature$ has the \SumDecompositionProp{}.
\end{theorem}

\begin{proof}
From Remark~\ref{rem:symmetries}, we just consider $\Tuple{\gpre,\gfac,\gsuf}$.
We first establish properties between the maximum words associated with $\gpre$, $\gfac$ and $\gsuf$.
\begin{itemize}
\item
From the \ConvexityProp{}, the signature of the words $x_{\ell,j}$, $x_{i,j}$ and $x_{i,u}$
respectively contain at most one maximum word in $\Language{\pattern}$.
Because of $\gpre$, $\gfac$ and $\gsuf$, the signature of the words $x_{\ell,j}$, $x_{i,j}$ and $x_{i,u}$
respectively contain at least one word in $\Language{\pattern}$.
Consequently, the signatures of the words $x_{\ell,j}$, $x_{i,j}$ and $x_{i,u}$ contain
one single maximum word in $\Language{\pattern}$, respectively
denoted by $w_\gpre$, $w_\gfac$ and $w_\gsuf$.
\item
Because the words $w_\gpre$ and $w_\gfac$ must not end after position $j$,
and from the \ConvexityProp{}, $w_\gpre$ and $w_\gfac$ end in the same position $\finf$.
\item
Because the words $w_\gsuf$ and $w_\gfac$ must not start before position $i$,
and from the \ConvexityProp{}, $w_\gsuf$ and $w_\gfac$ start at the same position $\startf$.
\item
Because the word $w_\gfac$ starts at position $\startf$ and ends at position $\finf$ we have that $\startf\leq\finf$.
\end{itemize}

Since $\feature$ has the \SumDecompositionProp{}, 
by using the function $h$ of Definition~\ref{pf-sum_decomp},
Equation~\eqref{formula1} can be rewritten as:

\noindent\begin{tikzpicture}[information text/.style={rounded corners,inner sep=1ex}]
\begin{scope}[font=\scriptsize]
\filldraw[fill=gray!50,draw=black] (0.36,0) rectangle (2.16,0.2);
\fill[gray!20] (1.08,0) rectangle (1.44,0.2);
\draw[draw=black] (0.36,0) rectangle (2.16,0.2);
\draw[-] (0.00,0) -- (2.52,0.0);
\draw    (0.00,0) -- (0.00,0.1);
\draw    (0.36,0) -- (0.36,0.1);
\draw    (0.72,0) -- (0.72,0.1);
\draw    (1.08,0) -- (1.08,0.1);
\draw    (1.44,0) -- (1.44,0.1);
\draw    (1.80,0) -- (1.80,0.1);
\draw    (2.16,0) -- (2.16,0.1);
\draw    (2.52,0) -- (2.52,0.1);
\node[anchor=south] (1) at (0.00,-0.4) {$1$};
\node[anchor=south] (2) at (0.36,-0.4) {$\ell$};
\node[anchor=south] (3) at (0.72,-0.4) {$i$};
\node[anchor=south] (4) at (1.08,-0.4) {$\startf$};
\node[anchor=south] (5) at (1.44,-0.45) {$\finf$};
\node[anchor=south] (6) at (1.80,-0.45) {$j$};
\node[anchor=south] (7) at (2.16,-0.4) {$u$};
\node[anchor=south] (8) at (2.52,-0.4) {$n$};
\node               (9) at (0.88,-0.6) {$\underbrace{\hspace*{52pt}}_{\featurePre=f_\sigma(x_{\ell,j})}$};
\node              (10) at (1.60, 0.6) {$\overbrace{\hspace*{52pt}}^{\featureSuf=f_{\sigma^r}(x_{u,i})}$};

\end{scope}
\begin{scope}[xshift=2.2cm,yshift=0.145cm]
\draw node[right,text width=9.6cm,information text]{
\begin{align}
\featureTw=\underbrace{\sum\limits_{\lPre}^{\uPre}{h(x_t)}}_{\featurePre}+
\underbrace{\sum\limits_{\lSuf}^{\uSuf}{h(x_t)}}_{\featureSuf}-\underbrace{\sum\limits_{\lTotal}^{\uTotal}{h(x_t)}}_{\featureTotal}\label{T1_C1}
\end{align}
};
\end{scope}
\end{tikzpicture}

By using the fact that the pattern $\pattern$ is reversible
(i.e.~$a_{\pattern^r}=b_\pattern$ and $b_{\pattern^r}=a_\pattern$)
in the second term of Equation~\eqref{T1_C1},
by expanding the terms $\featureSuf$ and $\featureTotal$
we obtain: 
\begin{align*}
\underbrace{\sum\limits_{\lPre}^{\uPre}{h(x_t)}}_{\featurePre}+
\underbrace{ \sum\limits_{\lTw}^{\uTw}{h(x_t)}+ \sum\limits_{t=\finf-a_\pattern+1}^{u-a_\pattern}{h(x_t)}}_{\featureSuf}-
\underbrace{\sum\limits_{\lPre}^{\uPre}{h(x_t)} - \sum\limits_{t=\finf-a_\pattern+1}^{u-a_\pattern}{h(x_t)}}_{\featureTotal}=\\
\sum\limits_{\lTw}^{\uTw}{h(x_t)}=\featureTw.
\end{align*}
\vspace{-0.5cm}

\noindent Hence, Equation \eqref{formula1} holds.
\qed
\end{proof}

\begin{theorem}
\label{theorem:E1_P3_IF}
Consider a pattern $\pattern$ whose class has a non-empty intersection with the set of representatives 
$\mathcal{S}=\{\Tuple{\gpre,\gfac,\gsuf}$, $\Tuple{\gin,\gin,\gin}\}$. 
Equation~\eqref{formula1} can be used to obtain $\featureTw$ 
for a sequence $x_{1,n}$ wrt window $\tw$ whose type is in $\mathcal{S}$, 
assuming that $\feature,\pattern$ has the \SinglePositionInflexionFreeProp{}.
\end{theorem}

\begin{proof}
Because of Remark~\ref{rem:symmetries} we only consider the representative $\Tuple{\gpre,\gfac,\gsuf}$.
In this context, for the reason quoted in the first part of the proof of Theorem~\ref{theorem:E1_P1}, the signature of $x_{\ell,j}$ contains a maximum word in $\Language{\pattern}$ ending at position $\finf$, the signature of $x_{i,j}$ contains a maximum word in $\Language{\pattern}$ starting at $\startf$ and ending at $\finf$, the signature of $x_{i,u}$ contains a maximum word in $\Language{\pattern}$ starting at $\startf$.
\begin{itemize}
\item[\ding{172}]
When Condition~\emph{(iii.a)} of Definition~\ref{pf-single_pos_inflexion_free} holds,
by using the function $h$ of Definition~\ref{pf-single_pos_inflexion_free},
Equation~\eqref{formula1} can be rewritten as:

\noindent\begin{tikzpicture}[information text/.style={rounded corners,inner sep=1ex}]
\begin{scope}[font=\scriptsize]
\filldraw[fill=gray!50,draw=black] (0.36,0) rectangle (2.16,0.2);
\fill[gray!20] (1.08,0) rectangle (1.44,0.2);
\draw[draw=black] (0.36,0) rectangle (2.16,0.2);
\draw[-] (0.00,0) -- (2.52,0.0);
\draw    (0.00,0) -- (0.00,0.1);
\draw    (0.36,0) -- (0.36,0.1);
\draw    (0.72,0) -- (0.72,0.1);
\draw    (1.08,0) -- (1.08,0.1);
\draw    (1.44,0) -- (1.44,0.1);
\draw    (1.80,0) -- (1.80,0.1);
\draw    (2.16,0) -- (2.16,0.1);
\draw    (2.52,0) -- (2.52,0.1);
\node[anchor=south] (1) at (0.00,-0.4) {$1$};
\node[anchor=south] (2) at (0.36,-0.4) {$\ell$};
\node[anchor=south] (3) at (0.72,-0.4) {$i$};
\node[anchor=south] (4) at (1.08,-0.4) {$\lambda$};
\node[anchor=south] (5) at (1.44,-0.45) {$\psi$};
\node[anchor=south] (6) at (1.80,-0.45) {$j$};
\node[anchor=south] (7) at (2.16,-0.4) {$u$};
\node[anchor=south] (8) at (2.52,-0.4) {$n$};
\node               (9) at (0.88,-0.6) {$\underbrace{\hspace*{52pt}}_{\featurePre=f_\sigma(x_{\ell,j})}$};
\node              (10) at (1.60, 0.6) {$\overbrace{\hspace*{52pt}}^{\featureSuf=f_{\sigma^r}(x_{u,i})}$};

\end{scope}
\begin{scope}[xshift=1.6cm,yshift=0.145cm]
\draw node[right,text width=9.6cm,information text]{
\begin{align}
\featureTw=\underbrace{h(x_{\ell+b_\pattern})}_{\featurePre}+
\underbrace{h(x_{\startf+a_{\pattern^r}})}_{\featureSuf}-\underbrace{h(x_{\ell+b_\pattern})}_{\featureTotal}\label{T3_C1}
\end{align}
};
\end{scope}
\end{tikzpicture}
From \eqref{T3_C1}, and by using the fact that $a_{\pattern^r}=b_{\pattern}$, 
we obtain $\featureTw = h(x_{\startf+b_{\pattern}})$, 
which is true by Condition~\emph{(iii.a)} of Definition~\ref{pf-single_pos_inflexion_free}.
\item[\ding{173}]
When Condition~\emph{(iii.b)} of Definition~\ref{pf-single_pos_inflexion_free} holds,
Equation~\eqref{formula1} can be proven in a similar way as in case~\ding{172}.
\end{itemize}
Hence, Equation \eqref{formula1} holds.
\qed
\end{proof}

\begin{theorem}
\label{theorem:E1_P2}
Consider a pattern $\pattern$ whose class has a non-empty intersection with the set of representatives
$\mathcal{S}=\{\Tuple{\gpre,\gfac,\gsuf}$, $\Tuple{\gin,\gin,\gin}$,   $\Tuple{\gpre,\gout,\gout}$,
              $\Tuple{\gout,\gout,\gsuf}$, $\Tuple{\gin,\gout,\gout}$, $\Tuple{\gout,\gout,\gin}\}$.  
Equation~\eqref{formula1} can be used to obtain $\featureTw$ 
for a sequence $x_{1,n}$ wrt a window $\tw$ whose type is in $\mathcal{S}$, 
assuming that feature $\feature$ has the \SameValueProp{}.
\end{theorem}

\begin{proof} 
\begin{sloppypar}
Because of Remark~\ref{rem:symmetries}, we only consider the representatives   
$\Tuple{\gpre,\gfac,\gsuf}$ and $\Tuple{\gpre,\gout,\gout}$.
\end{sloppypar}

\begin{itemize}
\item[$\bullet$]\hspace{0,000001pt}[$\Tuple{\gpre,\gfac,\gsuf}$]~Since $\feature$ has the 
\SameValueProp{}, Equation \eqref{formula1} can be rewritten as: 
\begin{align}
\featureTw=\underbrace{\featureTotal}_{\featurePre}+
\underbrace{\featureTotal}_{\featureSuf}-\underbrace{\featureTotal}_{\featureTotal}\label{T2_C1}
\end{align}
From \eqref{T2_C1} we obtain $\featureTw=\featureTotal$,
which is true by definition of the \SameValueProp{}. 
Hence, Equation \eqref{formula1} holds.\\

\item[$\bullet$]\hspace{0,000001pt}[$\Tuple{\gpre,\gout,\gout}$]~Since $\feature$ has the 
\SameValueProp{}, Equation \eqref{formula1} can be rewritten as: 
\begin{align}
\featureTw=\underbrace{\featureTotal}_{\featurePre}+
  {\underbrace{\vphantom{{\featureTotal}}0}_{\featureSuf}}
  -\underbrace{\featureTotal}_{\featureTotal}\label{T2_C2}  
\end{align}
From \eqref{T2_C2} we obtain $\featureTw=0$,
which is true by definition of the representative $\Tuple{\gpre,\gout,\gout}$. 
Hence, Equation \eqref{formula1} holds.
\qed
\end{itemize}
\end{proof}

\begin{theorem}
\label{theorem:E1_P3}
Consider a pattern $\pattern$ whose class has a non-empty intersection with the set of representatives 
$\mathcal{S}=\{\Tuple{\gpre,\gfac,\gsuf}$, $\Tuple{\gin,\gin,\gin}$, $\Tuple{\gpre,\gout,\gout}$, $\Tuple{\gout,\gout,\gsuf}$, $\Tuple{\gin,\gout,\gout}$, $\Tuple{\gout,\gout,\gin}\}$.
Equation~\eqref{formula1} can be used to obtain $\featureTw$ 
for a sequence $x_{1,n}$ wrt window $\tw$ whose type is in $\mathcal{S}$, 
assuming that the pair $\feature,\pattern$ has the \SinglePositionInflexionProp{}.
\end{theorem}

\begin{proof}
When the pair $\feature, \pattern$ has the \SinglePositionInflexionProp{} 
the following identity holds:
\begin{align*}
\feature_\pattern(x_{k,k'})=\feature_\pattern({x}_{1,n}), \forall~k,k'\mid 1\leq k\leq k'\leq n,~and
\end{align*}
\vspace{-0.9cm}
\begin{align}
\exists\text{ an extended }\pattern\text{-pattern occurrence in }[k,k'] \text{ wrt }x_{k,k'}\label{e_pf3}
\end{align}
Because of Remark~\ref{rem:symmetries} we only consider the representatives   
$\Tuple{\gpre,\gfac,\gsuf}$ and $\Tuple{\gpre,\gout,\gout}$.

\begin{itemize}
\item[$\bullet$]\hspace{0,000001pt}[$\Tuple{\gpre,\gfac,\gsuf}$]~By Identity~\eqref{e_pf3} 
and because all three intervals $[1,j]$, $[i,j]$ and $[i,n]$ contain an extended $\pattern$-pattern
occurrence, $\featurePre =\featureSuf=\featureTw=\featureTotal$;
thus  Equation~\eqref{formula1} can be rewritten as: 
\begin{align}
\featureTw=\underbrace{\featureTotal}_{\featurePre}+
\underbrace{\featureTotal}_{\featureSuf}-\underbrace{\featureTotal}_{\featureTotal}\label{T4_C1}
\end{align}
From \eqref{T4_C1} we obtain 
$\featureTw = \featureTotal$, 
which is true by definition of the \SinglePositionProp{}. 
Hence, Equation \eqref{formula1} holds.
\vspace{5pt}
\item[$\bullet$]\hspace{0,000001pt}[$\Tuple{\gpre,\gout,\gout}$]~By
Identity~\eqref{e_pf3} $\featurePre =\featureTotal$;
since there is no extended $\pattern$-pattern occurrence neither in $[i,j]$ nor in $[i,n]$,
$\featureSuf=\featureTw=0$, and Equation~\eqref{formula1} can be rewritten as:
\begin{align}
\featureTw=\underbrace{\featureTotal}_{\featurePre}+
{\underbrace{\vphantom{{\featureTotal}}0}_{\featureSuf}}-\underbrace{\featureTotal}_{\featureTotal}\label{T4_C2}
\end{align}
From \eqref{T4_C2} we obtain 
$\featureTw =0$, 
which is true by definition of the second component of the $\Tuple{\gpre,\gout,\gout}$ representative. 
Hence, Equation \eqref{formula1} holds. 
\qed
\end{itemize}
\end{proof}

\subsubsection{Sufficient Conditions for the Validity of Equations \eqref{formula2} and \eqref{formula3}}

\begin{theorem}
\label{theorem:E2_P1_P}
\begin{sloppypar}
Consider a pattern $\pattern$ whose class has a non-empty intersection with the set of representatives 
$\mathcal{S}=\{\Tuple{\gpre,\gfac,\gsuf}$, 
$\Tuple{\gin,\gin,\gin}$, 
$\Tuple{\gpre,\gout,\gout}$, 
$\Tuple{\gout,\gout,\gsuf}$, 
$\Tuple{\gin,\gout,\gout}$, 
$\Tuple{\gout,\gout,\gin}$, 
$\Tuple{\gout,\gout,\gout}\}$.  
Equation~\eqref{formula2} can be used to obtain $\featureTw$ 
for a sequence $x_{1,n}$ wrt window $\tw$ whose type is in $\mathcal{S}$, 
assuming that
feature $\feature$ has the \SumDecomposition{} and the \Positive{} properties.
If, in addition to the set $\mathcal{S}$,
we also have the representative $\Tuple{\gpre,\gout,\gsuf}$
then Equation~\eqref{formula2} can still be used,
provided that pattern $\pattern$ has the \ExcludeOutInProp.
\end{sloppypar}
\end{theorem}

\begin{proof}
\begin{sloppypar}
Because of Remark~\ref{rem:symmetries} we only consider the representatives
$\Tuple{\gpre,\gfac,\gsuf}$, $\Tuple{\gpre,\gout,\gout}$ and $\Tuple{\gout,\gout,\gout}$.
\end{sloppypar}

\begin{itemize}
\item[$\bullet$]\hspace{0,000001pt}[$\Tuple{\gpre,\gfac,\gsuf}$]~Since,
from Theorem~\ref{theorem:E1_P1}, 
Equation~\eqref{formula1} is valid for this representative
when $\feature$ has the \SumDecompositionProp{},
and since $\feature$ has the \PositiveProp{},
the right-hand side of Equation~\eqref{formula2}
is the maximum between zero and a positive value;
therefore Equation~\eqref{formula2} is also valid
for $\Tuple{\gpre,\gfac,\gsuf}$.
\item[$\bullet$]\hspace{0,000001pt}[$\Tuple{\gpre,\gout,\gout}$]~Since
$\feature$ has the \Positive{} and \SumDecomposition{} properties,  
$\featureTotal\geq 0$ and $\featurePre\leq\featureTotal$.
Due to the third component ``$\gout$''
of the representative $\featureSuf=0$.
When Equation~\eqref{formula1} is used, we obtain $\featureTw\leq 0$; 
but with Equation~\eqref{formula2} we get $\featureTw=0$, 
which is true due to the second component ``$\gout$'' of the representative.
Hence~\eqref{formula2} is valid.
\item[$\bullet$]\hspace{0,000001pt}[$\Tuple{\gout, \gout, \gout}$]~Since
$\feature$ has the \PositiveProp{}, $\featureTotal\geq 0$.
Due to the first and third ``$\gout$'' components of the representative,
$\featurePre =\featureSuf=0$. 
When Equation~\eqref{formula1} is used, we obtain $\featureTw\leq 0$;
but with Equation~\eqref{formula2} we get $\featureTw=0$, 
which is true due to the second component ``$\gout$'' of the representative. 
Hence, Equation~\eqref{formula2} is valid.
\item[$\bullet$]\hspace{0,000001pt}[$\Tuple{\gpre,\gout,\gsuf}$]~
\begin{itemize}
\item[--]
From the \ConvexityProp{}, the signature of the words $x_{\ell,j}$ and $x_{i,u}$
respectively contain at most one maximum word in $\Language{\pattern}$.
Because of $\gpre$ and $\gsuf$, the signature of the words $x_{\ell,j}$ and $x_{i,u}$
respectively contain at least one word in $\Language{\pattern}$.
Consequently, the signature of the words $x_{\ell,j}$ and $x_{i,u}$ contain
one single maximum word in $\Language{\pattern}$, respectively
denoted by $w_\gpre$ and $w_\gsuf$.
\item[--]
Because of the $\gout$ of $\Tuple{\gpre,\gout,\gsuf}$,
the signature of the word $x_{i,j}$ does not contain any subword that belongs to $\Language{\pattern}$.
In addition, since the pattern $\pattern$ has the \ExcludeOutInProp{} we have that
$w_\gpre$ ends before position $i$, and $w_\gsuf$ starts after position $j$,
i.e.~$w_\gpre$ and $w_\gsuf$ do not overlap.
Consequently, since in addition $\feature$ has the \Positive{} and the \SumDecomposition{} properties,
$\featureTotal\geq 0$ and $\featurePre+\featureSuf\leq\featureTotal$.
When Equation~\eqref{formula1} is used, we obtain $\featureTw\leq 0$; 
but with Equation~\eqref{formula2} we get $\featureTw=0$, 
which is true due to the second component ``$\gout$'' of the representative. 
Hence, Equation~\eqref{formula2} is valid.\qed
\end{itemize}
\end{itemize}
\end{proof}

\begin{theorem}
\label{theorem:E2_P2_P}
Consider a pattern $\pattern$ whose class has a non-empty intersection with the set of representatives 
$\mathcal{S}=\{\Tuple{\gpre,\gfac,\gsuf}$, 
$\Tuple{\gin,\gin,\gin}$, 
$\Tuple{\gpre,\gout,\gout}$, 
$\Tuple{\gout,\gout,\gsuf}$, 
$\Tuple{\gin,\gout,\gout}$, 
$\Tuple{\gout,\gout,\gin}$, 
$\Tuple{\gout,\gout,\gout}\}$.  
Equation~\eqref{formula2} can be used to obtain $\featureTw$ 
for a sequence $x_{1,n}$ wrt window $\tw$ whose type is in $\mathcal{S}$, 
assuming that feature $\feature$ has the \SameValue{} and the \Positive{} properties.
\end{theorem}
\begin{proof}
Because of Remark~\ref{rem:symmetries} we only consider the representatives   
$\Tuple{\gpre,\gfac,\gsuf}$, $\Tuple{\gpre,\gout,\gout}$ and $\Tuple{\gout,\gout,\gout}$.

\begin{sloppypar}
\begin{itemize}
\item[$\bullet$]\hspace{0,000001pt}[$\Tuple{\gpre,\gfac,\gsuf}$]~Since,
from Theorem~\ref{theorem:E1_P2}, 
Equation~\eqref{formula1} is valid for this representative
when $\feature$ has the \SameValueProp{},
and since $\feature$ has the \PositiveProp{},
Equation~\eqref{formula2} is the maximum between 
zero and a positive value; therefore Equation~\eqref{formula2} is also 
valid for $\Tuple{\gpre,\gfac,\gsuf}$.
\item[$\bullet$]\hspace{0,000001pt}[$\Tuple{\gpre,\gout,\gout}$]~Since
$\feature$ has the \SameValueProp{}, $\featurePre =\featureTotal$.
Due to the third component ``$\gout$'' of the
$\Tuple{\gpre,\gout,\gout}$ representative $\featureSuf=0$. 
Consequently, by Equation~\eqref{formula2}, we have 
$\featureTw=\max(0,\feature_\pattern(x_{1,j})+\feature_{\pattern^r}(x_{\seqlength,i})-$ $\feature_\pattern(x_{1,\seqlength}))=0$, 
which is true due to the second component ``$\gout$''
of the $\Tuple{\gpre,\gout,\gout}$ representative. 
Hence, Equation~\eqref{formula2} is valid.
\item[$\bullet$]\hspace{0,000001pt}[$\Tuple{\gout,\gout,\gout}$]~Since $\feature$ 
has the \PositiveProp{}, $\featureTotal\geq 0$.
Due to the first and third ``$\gout$'' components of
the $\Tuple{\gout,\gout,\gout}$ representative,
$\featurePre =\featureSuf=0$.
When Equation~\eqref{formula1} is used, we obtain $\featureTw\leq 0$;
but with Equation~\eqref{formula2} we get $\featureTw=0$,
which is true due to the second component ``$\gout$'' of
the $\Tuple{\gout,\gout,\gout}$ representative. 
Hence, Equation~\eqref{formula2} is valid.
\qed
\end{itemize}
\end{sloppypar}
\end{proof}

\begin{theorem}
\label{theorem:E3}
\begin{sloppypar}
Consider a pattern $\pattern$ whose class has a non-empty 
intersection with the set of representatives 
$\mathcal{S}=\{
\Tuple{\gpre,\gfac,\gsuf}$,
$\Tuple{\gin,\gin,\gin}$,
$\Tuple{\gpre,\gout,\gsuf}$,
$\Tuple{\gin,\gout,\gin}$,
$\Tuple{\gpre,\gout,\gout}$,
$\Tuple{\gin,\gout,\gout}$,
$\Tuple{\gout,\gout,\gsuf}$,
$\Tuple{\gout,\gout,\gin}$,
$\Tuple{\gout,\gout,\gout}\}$.
Equation~\eqref{formula3} can be used to obtain $\featureTw$ 
for a sequence $x_{1,n}$ wrt a window $\tw$ 
whose type is in $\mathcal{S}$ if
(a)~either both 
$\Tuple{\gpre,\gfac,\gsuf}$ and $\Tuple{\gin,\gin,\gin}$ are not representatives of the pattern $\pattern$,
(b)~or if one of the following conditions holds:
\begin{enumerate}[label=\roman*)]
\item $\feature$ has the \SumDecompositionProp{}.
\item $\feature$ has the \SameValueProp{}. 
\item the pair $\feature,\pattern$ has
      the \SinglePositionInflexionFree{} or the \SinglePositionInflexion{} properties.
\end{enumerate}
\end{sloppypar}
\end{theorem}

\begin{proof}
\noindent[CASE 1]
Consider the representatives that have an extended $\pattern$-pattern occurrence in $\tw$.
In this case the only two representatives are
$\Tuple{\gpre,\gfac,\gsuf}$ and $\Tuple{\gin,\gin,\gin}$. 
Since Equation~\eqref{formula1} is valid for these representatives when:
\begin{enumerate}[label=\emph{\roman*)}]
\item $\feature$ has the \SumDecompositionProp{} (see Theorem~\ref{theorem:E1_P1}),
\item $\feature$ has the \SameValueProp{} (see Theorem~\ref{theorem:E1_P2}),
\item the pair $\feature,\pattern$ has the \SinglePositionInflexionFreeProp{} (see Theorem~\ref{theorem:E1_P3_IF})
or the \SinglePositionInflexionProp{} (see Theorem~\ref{theorem:E1_P3}),
\end{enumerate}
\vspace{-7pt}
Equation~\eqref{formula3} is also valid.
\vspace{5pt}

\noindent[CASE 2]
Consider the representatives that do not have an extended $\pattern$-pattern occurrence in $\tw$.
When using Equation~\eqref{formula3},
because of the check
``if no $\sigma$\nobreakdash-pattern in $x_{i,j}$''
in Equation~\eqref{formula3},
the value of $\featureTw$ is zero.
Hence, Equation~\eqref{formula3} is valid.
\qed
\end{proof}

\vspace{-0.5cm}
\subsubsection{Synthesis}

The classification induced by theorems~\ref{theorem:E1_P1} to~\ref{theorem:E3}
is presented in Table~\ref{table:summary}: 
for each pattern class corresponding to the same set of representative triples
(see Figure~\ref{fig:factor_cartography}) we select one
pattern (see the columns of Table~\ref{table:summary}, e.g.~$\PlainPatternName$) and provide
for each feature property (see the rows of Table~\ref{table:summary}, e.g.~{\footnotesize{SV}}) and
for each feature/pattern property (see the cells of Table~\ref{table:summary}, e.g.~{\footnotesize{SPN}})
the theorem proving that an Equation is valid under such properties.
Note that any missing Equation is due to a counterexample given
in Appendix~\ref{sec:counterexamples},
and not to the fact that we are missing a theorem.
Coloured grey cells indicate a non-existing time-series constraint in the time-series catalogue~\cite{arafailova2016global}.

Equation~\eqref{formula3} can be used to compute the value of $\featureTw$, 
for all reversible and convex patterns without the \SingleLetterProp{} from~\cite{arafailova2016global}, except for $\ZigzagPatternName$ with the $\MaxFeature$
and $\MinFeature$ features (see the cells marked with ``none'' in Table \ref{table:summary}),
as $\ZigzagPatternName$ uses the representative triple $\Tuple{\gin,\gin,\gin}$
without having the $\SinglePositionInflexion$ or the $\SinglePositionInflexionFree$ properties.

\vspace{-1cm}
\noindent
\begin{table}[!h]
\resizebox{12.2cm}{!}
{
\centering
\begin{tikzpicture}[information text/.style={rounded corners,inner sep=1ex}]
\begin{scope}[xshift=0.9cm,yshift=4.5cm]
\draw[rounded corners=2pt,thick,blue,fill=blue!10] (0,0) rectangle (2.2,0.4) node[pos=.5] {\footnotesize\color{black}$\Tuple{\apre,\aout,\asuf}\hspace*{1pt}\Tuple{\apre,\afac,\asuf}$};
\draw[rounded corners=2pt,thick,brown,fill=brown!10] (4.75,0) rectangle (10.15,0.4) node[pos=.5] {\footnotesize\color{black}$\Tuple{\aout,\aout,\aout}\hspace*{1pt}\Tuple{\ain,\aout,\aout}\hspace*{1pt}\Tuple{\aout,\aout,\ain}\hspace*{1pt}\Tuple{\ain,\aout,\ain}\hspace*{1pt}\Tuple{\ain,\ain,\ain}$};
\draw[rounded corners=2pt,thick,red,fill=red!10] (1.2,-0.5) rectangle (5.9,-0.1) node[pos=.5] {\footnotesize\color{black}$\Tuple{\apre,\afac,\asuf}\hspace*{1pt}\Tuple{\apre,\aout,\aout}\hspace*{1pt}\Tuple{\aout,\aout,\asuf}\hspace*{1pt}\Tuple{\aout,\aout,\aout}$};
\draw[rounded corners=2pt,thick,cyan,fill=cyan!10] (9.3,-0.5) rectangle (10.15,-0.1) node[pos=.5] {\footnotesize\color{black}$\Tuple{\ain,\ain,\ain}$};
\draw[rounded corners=2pt,thick,orange,fill=orange!10] (1.2,-1) rectangle (4.7,-0.6) node[pos=.5] {\footnotesize\color{black}$\Tuple{\apre,\afac,\asuf}\hspace*{1pt}\Tuple{\apre,\aout,\aout}\hspace*{1pt}\Tuple{\aout,\aout,\asuf}$};
\draw[rounded corners=2pt,thick,violet,fill=violet!10] (4.8,-1) rectangle (8.3,-0.6) node[pos=.5] {\footnotesize\color{black}$\Tuple{\aout,\aout,\aout}\hspace*{1pt}\Tuple{\ain,\aout,\aout}\hspace*{1pt}\Tuple{\aout,\aout,\ain}$};
\draw[thin,blue]   (0.90,0)    -- (0.90,-2.85);
\draw[thin,orange] (3.10,-1)   -- (3.10,-2.85);
\draw[thin,red]    (4.75,-0.5) -- (4.75,-2.85);
\draw[thin,violet] (6.50,-1)   -- (6.50,-2.85);
\draw[thin,brown]  (8.50,0)    -- (8.50,-2.85);
\draw[thin,cyan]   (10.1,-0.5) -- (10.1,-2.85);
\fill[fill=gray!20] (1.95,-5.58) rectangle (7.3,-5);
\fill[fill=gray!20] (9.2,-5.58) rectangle (11.5,-5);
\draw[rounded corners=2pt,densely dashed,black!66] (-0.1,-1.1) rectangle (11.3,0.5);
\node[anchor=west] at (-0.1,0.67) {\footnotesize \color{black!66}\textsc{representative triples}};
\draw[rounded corners=2pt,densely dashed,black!66] (-1.98,-6.25) rectangle (-0.92,-3.2);
\node[anchor=west,rotate=90] at (-2.15,-6.4) {\footnotesize \color{black!66}\textsc{feature properties}};
\draw[rounded corners=2pt,thick,gray!30,fill=gray!10] (-1.4,-3.8) rectangle (-1.0,-3.25) node[pos=.5,rotate=90] {\scriptsize\color{black}SV};
\draw[rounded corners=2pt,thick,gray!30,fill=gray!10] (-1.4,-5) rectangle (-1.0,-3.9) node[pos=.5,rotate=90] {\scriptsize\color{black}SD};
\draw[rounded corners=2pt,thick,gray!30,fill=gray!10] (-1.4,-6.18) rectangle (-1.0,-5.08) node[pos=.5,rotate=90] {\scriptsize\color{black}SP};
\draw[rounded corners=2pt,thick,gray!30,fill=gray!10] (-1.9,-4.46) rectangle (-1.5,-3.25) node[pos=.5,rotate=90] {\scriptsize\color{black}P};
\draw[rounded corners=2pt,thick,gray!30,fill=gray!10] (1,-2.3) rectangle (1.5,-1.8) node[pos=.5] {\scriptsize\color{black}N,E};
\draw[rounded corners=2pt,thick,gray!30,fill=gray!10] (2.77,-2.3) rectangle (3.02,-1.8) node[pos=.5] {\scriptsize\color{black}O};
\draw[rounded corners=2pt,thick,gray!30,fill=gray!10] (4.42,-2.3) rectangle (4.67,-1.8) node[pos=.5] {\scriptsize\color{black}O};
\draw[rounded corners=2pt,thick,gray!30,fill=gray!10] (6.18,-2.3) rectangle (6.43,-1.8) node[pos=.5] {\scriptsize\color{black}O};
\draw[rounded corners=2pt,thick,gray!30,fill=gray!10] (10.18,-2.3) rectangle (10.78,-1.8) node[pos=.5] {\scriptsize\color{black}N,E};
\draw[rounded corners=2pt,densely dashed,black!66] (-0.1,-1.7) rectangle (11.3,-2.35);
\node[anchor=west] at (-0.1,-1.55) {\footnotesize \color{black!66}\textsc{pattern properties}};
\node[anchor=west] at (-0.8,-3.08) {\footnotesize $\feature\hspace*{3pt}\backslash\hspace*{3pt}\pattern$};
\node[anchor=west] at (0.6,-5.12) {\tiny \bf SPN};
\node[anchor=west] at (0.6,-5.68) {\tiny \bf SPN};
\node[anchor=west] at (2.15,-5.68) {\tiny \bf SPO};
\node[anchor=west] at (4.2,-5.68) {\tiny \bf SPO};
\node[anchor=west] at (5.93,-5.68) {\tiny \bf SPO};
\node[anchor=west] at (9.4,-5.68) {\tiny \bf SPN};
\end{scope}
\begin{scope}
\draw node[right,text width=12cm,information text]{
\begin{tabular}{lccccccccc}
& \hspace*{9pt}\colorbox{blue!10}{$\DecreasingSequencePatternName$}\hspace*{9pt}
& \hspace*{7pt}\colorbox{orange!10}{$\GorgePatternName$}\hspace*{7pt}
& \hspace*{7pt}\colorbox{red!10}{$\ValleyPatternName$}\hspace*{7pt}
& \hspace*{7pt}\colorbox{violet!10}{$\PlainPatternName$}\hspace*{7pt}
& \hspace*{7pt}\colorbox{brown!10}{$\ZigzagPatternName$}\hspace*{7pt}
& \hspace*{7pt}\colorbox{cyan!10}{$\SteadySequencePatternName$}\hspace*{7pt} \\ \toprule
$\One$ &
{\scriptsize\ref{theorem:E3}(\ref{formula3})} &
{\scriptsize\ref{theorem:E1_P2}(\ref{formula1}),\ref{theorem:E2_P2_P}(\ref{formula2}),\ref{theorem:E3}(\ref{formula3})} &
{\scriptsize\ref{theorem:E2_P2_P}(\ref{formula2}),\ref{theorem:E3}(\ref{formula3})} &
{\scriptsize\ref{theorem:E2_P2_P}(\ref{formula2}),\ref{theorem:E3}(\ref{formula3})} &
{\scriptsize\ref{theorem:E3}(\ref{formula3})} &
{\scriptsize\ref{theorem:E1_P2}(\ref{formula1}),\ref{theorem:E2_P2_P}(\ref{formula2}),\ref{theorem:E3}(\ref{formula3})}
\\ \midrule
$\Width$ &
{\scriptsize\ref{theorem:E2_P1_P}(\ref{formula2}),\ref{theorem:E3}(\ref{formula3})} &
{\scriptsize\ref{theorem:E2_P1_P}(\ref{formula2}),\ref{theorem:E3}(\ref{formula3})} &
{\scriptsize\ref{theorem:E2_P1_P}(\ref{formula2}),\ref{theorem:E3}(\ref{formula3})} &
{\scriptsize\ref{theorem:E2_P1_P}(\ref{formula2}),\ref{theorem:E3}(\ref{formula3})} &
{\scriptsize\ref{theorem:E3}(\ref{formula3})} &
{\scriptsize\ref{theorem:E1_P1}(\ref{formula1}),\ref{theorem:E2_P1_P}(\ref{formula2}),\ref{theorem:E3}(\ref{formula3})} &
\\ \midrule
$\Surf$ &
{\scriptsize\ref{theorem:E3}(\ref{formula3})} &
{\scriptsize\ref{theorem:E3}(\ref{formula3})} &
{\scriptsize\ref{theorem:E3}(\ref{formula3})} &
{\scriptsize\ref{theorem:E3}(\ref{formula3})} &
{\scriptsize\ref{theorem:E3}(\ref{formula3})} &
{\scriptsize\ref{theorem:E1_P1}(\ref{formula1}),\ref{theorem:E3}(\ref{formula3})} &\\ \midrule
$\MaxFeature$ &
{\scriptsize\ref{theorem:E3}(\ref{formula3})} &
&
&
&
{\scriptsize none}\\ \midrule
$\MinFeature$ &
{\scriptsize\ref{theorem:E3}(\ref{formula3})} &
{\scriptsize\ref{theorem:E1_P3}(\ref{formula1}),\ref{theorem:E3}(\ref{formula3})} &
{\scriptsize\ref{theorem:E3}(\ref{formula3})} &
{\scriptsize\ref{theorem:E3}(\ref{formula3})} &
{\scriptsize none}&
{\scriptsize\ref{theorem:E1_P3_IF}(\ref{formula1}),\ref{theorem:E3}(\ref{formula3})}
\\ \bottomrule
\end{tabular}
};
\end{scope}
\end{tikzpicture}
}
\caption{Indicates, for existing combinations of feature $\feature$ and pattern $\pattern$ from the time-series catalogue, which of the Equations~\eqref{formula1}, \eqref{formula2} and~\eqref{formula3} are valid, as well as the corresponding justifying theorem, where:
(\emph{i})~within a representative triple we use as a shortcut the first letter of each component;
(\emph{ii})~{\sc p}, {\sc sp}, {\sc sd} and {\sc sv} resp. indicate whether the feature $\feature$ has
the \Positive{},
the \SinglePosition{},
the \SumDecomposition{} or
the \SameValue{} property;
(\emph{iii})~{\sc n}, {\sc e}, and {\sc o} resp. indicate whether the pattern $\pattern$ has
the \InflexionFree{},
the \ExcludeOutIn{} or
the \OneInflexion{} property;
(\emph{iv})~{\sc spn}, {\sc spo} resp. indicate whether the pair 
$\feature, \pattern$ has the \SinglePositionInflexionFreeProp{} 
or the \SinglePositionInflexionProp{}.
 \label{table:summary}}
\end{table}

\subsection{Optimal Time Complexity Checkers}

Since checking each window of $\windowsize$ consecutive positions of a sequence
of size $\seqlength$ independently gives a time complexity of $O(\windowsize\cdot\seqlength)$,
we now introduce a theorem leading to an optimal time complexity.

\begin{theorem}\label{theo:checker}
The time complexity of evaluating Equations~(\ref{formula1}) and (\ref{formula2})
on a sequence $\sequence=\seq$ for all sliding windows of size $\windowsize$ is $\Theta(\seqlength)$.
Moreover, assuming one can check in constant time whether a sliding window of the sequence $\sequence$
contains or not a $\pattern$-pattern, the time complexity
of evaluating Equation~(\ref{formula3}) for all sliding windows of size $\windowsize$ of sequence $\sequence$
is also $\Theta(\seqlength)$.
\end{theorem}
\begin{proof}
Evaluating~(\ref{formula1}), (\ref{formula2}) and~(\ref{formula3})
for all sliding windows $[i,j]$ (with $i\in[1,\seqlength-\windowsize+1]$ and $j=i+\windowsize-1$)
requires evaluating
$\feature_\pattern(x_{1,i+\windowsize-1})$,
$\feature_{\pattern}(x_{i,\seqlength})$ and
$\feature_\pattern(x_{1,\seqlength})$.
\begin{itemize}
\item
First, note that within Equation~(\ref{formula3})
all the tests ``$\text{if no }\sigma\text{-pattern in}~x_{i,j}$'' on the different sliding windows
(with $i\in[1,\seqlength-\windowsize+1]$ and $j=i+\windowsize-1$)
can be done in $O(n)$ because of our assumption.
\item
Second, evaluating $\feature_\pattern(x_{1,i+\windowsize-1})$
for all $i\in[1,\seqlength-\windowsize+1]$
as well as $\feature_\pattern(x_{1,\seqlength})$
can be done in $O(\seqlength)$ by using a register automaton~\cite{Beldiceanu:synthesis} for
$\constraint{sum\_}f\constraint{\_}\pattern(\result,$ $x_1 x_2\dots x_\seqlength)$,
which exposes all its intermediate register values~\cite{Sicstus19}.
\item
Third, since the pattern $\pattern$ is reversible and since the feature $\feature$ is commutative,
$\feature_{\pattern}(x_{i,\seqlength})=\feature_{\pattern^r}(x_{\seqlength,i})$.
Evaluating $\feature_{\pattern^r}(x_{\seqlength,i})$
for all $i\in[1,\seqlength-\windowsize+1]$ can also be done in $O(\seqlength)$ by using
a register automaton for
$\constraint{sum\_}f\constraint{\_}\pattern^r(\result,$ $x_\seqlength x_{\seqlength-1}\dots x_1)$, which exposes all its intermediate register values.
\end{itemize}
Therefore, the time complexity of evaluating Equations~(\ref{formula1}), (\ref{formula2})
and~(\ref{formula3}) is~$O(\seqlength)$.
Since each variable of $x_{1,n}$ needs to be scanned at least once to identify pattern occurrences,
this time complexity is optimum.\qed
\end{proof}

\subsubsection{Pattern Properties for Checking in Linear Time the Occurrence of Pattern in Sliding Windows}

We now introduce some additional pattern properties to check in time $O(\seqlength)$
whether or not the different sliding windows of size $\windowsize$
of a sequence $\sequence=\seq$ contain a pattern occurrence.
As these properties cover all reversible patterns of the time-series catalogue,
one can also use Equation~(\ref{formula3}) for such patterns for the entries
of Table~\ref{table:summary} mentioning~(\ref{formula3}).

\begin{definition}
\label{def:letter_regexp}
A pattern $\pattern$ has the \LetterProp{} wrt a letter $e$ if
$e$ is a word in $\Language{\pattern}$, and if
any word of $\Language{\pattern}$ contains at least one occurrence of $e$,
i.e.~if $\Language{\pattern}\cap\{e\}\neq\emptyset$ and
if $\Language{\pattern}\cap(\Sigma\setminus e)^*=\emptyset$.
\end{definition}

\begin{definition}
\label{def:suffix_unavoidable_regexp}
A pattern $\pattern$ has the \SuffixUnavoidableProp{} wrt a letter $e\in\{\reg{<},\reg{=},\reg{>}\}$ if
all words in $\Language{\pattern}$ contain at least one occurrence of $e$, and if
each suffix starting with the letter $e$ of any word of $\Language{\pattern}$
belongs also to $\Language{\pattern}$,
i.e.~if $\Language{\pattern}\cap(\Sigma\setminus e)^*=\emptyset$ and
if $\shuffle(\Language{\pattern},s)\land\Sigma^*s\,e\,\Sigma^*\land\Sigma^*s\,(\Sigma^*\setminus\Language{\pattern})=\emptyset$.
\end{definition}

\begin{definition}
\label{def:incompressible_regexp}
A pattern $\pattern$ has the \IncompressibleProp{} if 
all \emph{proper factors} of any word in $\Language{\pattern}$ do not belong to $\Language{\pattern}$,
i.e.~if $\Sigma^+\Language{\pattern}\Sigma^*\cap\Language{\pattern}=\emptyset$ and
if $\Sigma^*\Language{\pattern}\Sigma^+\cap\Language{\pattern}=\emptyset$.
\end{definition}

\begin{definition}
\label{def:factor_regexp}
A pattern $\pattern$ has the \FactorProp{} if
for any word $\word$ in $\Language{\pattern}$
all factors of $\word$, whose length is greater than or equal to
the smallest length $\omega_\pattern$ of a word in $\Language{\pattern}$, belong also
to $\Language{\pattern}$,
i.e.~if $\shuffle(\shuffle(\Language{\pattern},s),s)\land\Sigma^*s\Sigma^*\Sigma^{\omega_\pattern}\Sigma^*s\Sigma^*\land\Sigma^*s(\Sigma^*\setminus\Language{\pattern})s\Sigma^*=\emptyset$.
\end{definition}

\begin{sloppypar}
\begin{example}[pattern properties, continuation of Example~\ref{ex:pattern_properties}]
\begin{itemize}
\item
Eight out of the $19$ reversible patterns of~\cite{arafailova2016global} have the \LetterProp{}.
For instance, the patterns $\DecreasingPatternName$, $\DecreasingSequencePatternName$ and
$\StrictlyDecreasingSequencePatternName$ all have the \LetterProp{} wrt $\{\reg{>}\}$
since (\emph{i})~the word $\reg{>}$ is in $\Language{\DecreasingPatternName}$,
in $\Language{\DecreasingSequencePatternName}$ and in $\Language{\StrictlyDecreasingSequencePatternName}$,
and (\emph{ii})~any word in $\Language{\DecreasingPatternName}$,
in $\Language{\DecreasingSequencePatternName}$ or
in $\Language{\StrictlyDecreasingSequencePatternName}$ contains at least one occurrence of $\reg{>}$.
\item
$16$ out of the $19$ reversible patterns of~\cite{arafailova2016global} have
the $\SuffixUnavoidable{}$ property.
For instance, the pattern $\PeakPatternName$ has the $\SuffixUnavoidable{}$
property wrt the letter $\reg{<}$,
since (\emph{i})~any occurrence of peak contains at least one occurrence of $\reg{<}$, and
since (\emph{ii})~any suffix, starting with a $\reg{<}$, of a word of $\Language{\PeakPatternName}$
is also a peak.
\item
Six out of the $19$ reversible patterns of~\cite{arafailova2016global} have the \IncompressibleProp{}. 
The pattern $\DecreasingTerracePatternName$ has the \IncompressibleProp{} because, 
if an occurrence of the letter \reg{>} is removed from any word in $\Language{\DecreasingTerracePatternName}$,
the corresponding proper factor is not in $\Language{\DecreasingTerracePatternName}$.
\item
Seven out of the $19$ reversible patterns of~\cite{arafailova2016global} have the \FactorProp{}.
For instance, the pattern $\ZigzagPatternName$ has the $\Factor{}$ property because
any factor of length greater than or equal to $\omega_\ZigzagPatternName=3$ of a zigzag is also a zigzag.
\end{itemize}
\end{example}
\end{sloppypar}

For each pattern property described in Definitions~\ref{def:letter_regexp} to \ref{def:factor_regexp}
we now show how to check in $O(\seqlength)$ which sliding windows are empty or not.
\begin{itemize}
\item
Consider a pattern $\pattern$ that has the \LetterProp{} wrt a letter $e$.
First compute in one scan the number of occurrences $\mathit{nocc}[k]$ of $e$
in $x_{1,k}$ for all $k\in[1,\seqlength]$;
second, for each sliding window $[i,j]$,
check in constant time that $\mathit{nocc}[i]=\mathit{nocc}[j]$.
\item
Consider a pattern $\pattern$ that has the \SuffixUnavoidableProp{} wrt a letter $e$.
First compute in one scan the number of occurrences $\mathit{nocc1}[k]$ of $e$
in $x_{1,k}$ for all $k\in[1,\seqlength]$;
second compute in one scan the number of maximal occurrences $\mathit{nocc2}[k]$ of pattern $\pattern$
in $x_{1,k}$ for all $k\in[1,\seqlength]$;
third, for each sliding window $[i,j]$,
check in constant time that $\mathit{nocc1}[i]=\mathit{nocc1}[j]\lor\mathit{nocc2}[i]=\mathit{nocc2}[j]$.
\item
Consider a pattern $\pattern$ that has the \Incompressible{} or the \Factor{} property.
First compute for each $k=1,2,\dots,\seqlength$ the end $\mathit{end}[k]$
of the next pattern occurrence
(which will be set to $\seqlength+1$ if no pattern occurrence ends after $k$,
e.g.~$\mathit{end}[\seqlength]=\seqlength+1$).
Second compute for each $k=\seqlength,\seqlength-1,\dots,1$ the start $\mathit{start}[k]$
of the previous pattern occurrence (which will be set to $0$ if no pattern occurrence starts before $k$,
e.g.~$\mathit{start}[1]=0$).
Third, depending on whether the pattern has the \Incompressible{} or the \Factor{} property,
do the following check in constant time for each sliding window $[i,j]$:
\begin{itemize}
\item\hspace*{0.01pt}[\Incompressible{}]
{\footnotesize$\texttt{return}~\mathit{end}[i]>j\lor\mathit{start}[j]<i$}
\item\hspace*{0.01pt}[\Factor{}]
{\footnotesize
\hspace*{39pt}$\mathit{endi} = \mathit{end}[i]$, $\mathit{startj} = \mathit{start}[j]$\\
\hspace*{83pt}$\texttt{if}~\mathit{endi}>\seqlength \lor \mathit{startj}<1~\texttt{then return true}$\\
\hspace*{83pt}$\texttt{if}~\mathit{endi}-i\hspace*{6pt}\geq\omega_\pattern~\texttt{then}~i'\hspace*{1.5pt}=\hspace*{1pt}i~\texttt{else}~i'\hspace*{1pt}=\mathit{endi}$\\
\hspace*{83pt}$\texttt{if}~j-\mathit{startj}\geq\omega_\pattern~\texttt{then}~j'=j~\texttt{else}~j'=\mathit{startj}$\\
\hspace*{83pt}$\mathit{endi}'= \mathit{end}[i']$, $\mathit{startj}' = \mathit{start}[j']$\\
\hspace*{83pt}$\texttt{if}~\mathit{endi}'>\seqlength \lor \mathit{startj}'<1~\texttt{then return true}$\\
\hspace*{83pt}$\texttt{return}~\min(j',\mathit{endi}')-\max(i',\mathit{start}[\min(j',\mathit{endi}')])<\omega_\pattern$
}
\end{itemize}

Computing the end (resp.~start) of the next (resp.~previous) pattern occurrence is done
by using a register automaton derived from the transducer~\cite{Beldiceanu:synthesis}
which recognises pattern occurrences.\footnote{How to generate a transducer that
recognises all maximal pattern occurrences was described in~\cite{RodriguezFlenerPearson17}.}
Figure~\ref{fig:automata_next} give the register automaton
associated with the $\PlainPatternName$ and the $\ZigzagPatternName$ patterns.
In (A),~the dotted transition marks the end of a plain.
In (B),~the dashed (resp. dotted) transitions indicate that we are inside a zigzag
(resp.~that a zigzag is ending). Depending whether we were in a zigzag or not we set
$\mathit{end}[\seqlength-1]$ to $\seqlength$ or to $\seqlength+1$.

\begin{example}[Running automata that compute the end of the next pattern occurrence]
Table~\ref{table:example_running_end_plain_automata}
(resp. Table~\ref{table:example_running_end_zigzag_automata}) shows an example of execution
of the register automaton given in Part~(A) (resp.~(B)) of Figure~\ref{fig:automata_next}.
\end{example}

\subsection{Optimal Space Complexity Reformulation}

Rather than stating a time-series constraint on each window of size $\windowsize$, which
would result in an $O(\windowsize\cdot\seqlength)$ space complexity, we now show how to reformulate
the \constraint{slide\_sum\_}$\feature$\constraint{\_}$\pattern(\windowsize,\low,\up,\seq)$
constraint as a conjunction of constraints with a space complexity of $\Theta(\seqlength)$.
This reformulation was extended to the patterns of Table~\ref{tab:fgp} for Equation~(\ref{formula3})
to reformulate condition ``$\text{if no }\sigma\text{-pattern in}~x_{i,j}$'',
but is not described here for space reasons.

\begin{theorem}
For those time-series constraints for which Equations~(\ref{formula1}) or~(\ref{formula2}) holds,
the constraint
\constraint{slide\_sum\_}$\feature$\constraint{\_}$\pattern(\windowsize,\low,\up,\seq)$
can be reformulated with a space complexity of $\Theta(\seqlength)$.
\end{theorem}
\begin{proof}
For Equation~(\ref{formula1}), it can be reformulated as the conjunction
\begin{equation}
\left\{
\begin{array}{ll}
\constraint{sum\_}\feature\constraint{\_}\pattern\left(r,~\seq,\hspace*{14pt}\overrightarrow{r_1} \overrightarrow{r_2}\dots \overrightarrow{r_\seqlength}\right)~\land\\[2pt]
\constraint{sum\_}\feature\constraint{\_}\pattern\left(r,~\seqrev,~\overleftarrow{r_1} \overleftarrow{r_2}\dots \overleftarrow{r_\seqlength}\right)~\land\\[2pt]
\forall i\in[1,\seqlength-\windowsize+1]: r_{i,j}=\overrightarrow{r_j}+\overleftarrow{r_i}-r \text{~(with~}j=i+\windowsize-1\text{)}~\land\\[2pt]
\low=\min(r_{1,\windowsize} r_{2,\windowsize+1}\dots r_{\seqlength-\windowsize+1,\seqlength})~\land\\[2pt]
\up\hspace*{4pt}=\max(r_{1,\windowsize} r_{2,\windowsize+1}\dots r_{\seqlength-\windowsize+1,\seqlength})
\end{array}\label{reformulation}
\right.
\end{equation}
where $\overleftarrow{r_i}$ (resp.~$\overrightarrow{r_j}$) is the exposed register value corresponding
to the first argument of $\constraint{sum\_}f\constraint{\_}\pattern(\overrightarrow{r_j},$ $x_1 x_2\dots x_j)$,
(resp.~$\constraint{sum\_}f\constraint{\_}\pattern^r(\overleftarrow{r_i},$ $x_\seqlength x_{\seqlength-1}\dots x_i)$).
For Equation~(\ref{formula2}), we replace in~(\ref{reformulation}) the term
$r_{i,j}=\overrightarrow{r_j}+\overleftarrow{r_i}-r$ by the term
$r_{i,j}=\max(0,\overrightarrow{r_j}+\overleftarrow{r_i}-r)$.
\qed
\end{proof}

\begin{figure}[!h]
\centering
\begin{tikzpicture}
\begin{scope}[->,>=stealth',shorten >=1pt,auto,node distance=14mm,semithick,yshift=0.8cm]
\node[initial,initial where=left,initial text=,initial distance=2.7mm,accepting,state,minimum size=5pt] (s) {\scriptsize$s$};
\node[accepting,state,minimum size=5pt] (r)[right of=s] {\scriptsize$r$};
\path
 (s) edge [loop above] node{\scriptsize $<$} (s)
 (s) edge [loop below] node{\scriptsize $=$} (s)
 (s) edge [bend left]  node {\scriptsize $>$} (r)
 (r) edge [loop above] node{\scriptsize $>$} (r)
 (r) edge [loop below] node{\scriptsize $=$} (r)
 (r) edge [densely dotted,bend left]  node {\scriptsize $<$} (s);
\end{scope}
\begin{scope}[->,>=stealth',shorten >=1pt,auto,node distance=12mm,semithick,xshift=-0.33cm,yshift=1.3cm]
\node (a) at (-0.2,0.1) {\footnotesize (A)};
\fill[rounded corners=2pt,gray!20] (-0.6,-3.9) rectangle (2.25,-1.4);
\node (tr) at (0.3,-1.6) {\bf \scriptsize transitions:};
\node (o) at (0.3,-3.7) {\scriptsize $\circ\in\{<,=,>\}$};
\draw[->] (-0.5,-2.5) -- (2.2,-2.5) node[pos=.5] {\scriptsize $\begin{array}{c}x_k\circ x_{k+1}\\ \mathit{end}[k]=\mathit{end}[k+1]\end{array}$};
\draw[densely dotted,->] (-0.5,-3.4) -- (2.2,-3.4) node[pos=.5] {\scriptsize $\begin{array}{c}x_k \circ x_{k+1}\\ \mathit{end}[k]=k+1\end{array}$};
\end{scope}
\begin{scope}[->,>=stealth',shorten >=1pt,auto,node distance=12mm,semithick,xshift=2.7cm]
\node[initial,initial where=left,initial text=,initial distance=3mm,accepting,state,minimum size=5pt] (s) {\scriptsize$s$};
\node[accepting,state,minimum size=5pt] (a)[above right of=s] {\scriptsize$a$};
\node[accepting,state,minimum size=5pt] (b)[right of=a]       {\scriptsize$b$};
\node[accepting,state,minimum size=5pt] (c)[right of=b]       {\scriptsize$c$};
\node[accepting,state,minimum size=5pt] (d)[below right of=s] {\scriptsize$d$};
\node[accepting,state,minimum size=5pt] (e)[right of=d]       {\scriptsize$e$};
\node[accepting,state,minimum size=5pt] (f)[right of=e]       {\scriptsize$f$};
\node (cc) at ($(c)+(0.9,0)$) {$s$};
\node (ff) at ($(f)+(0.9,0)$) {$s$};
\path
 (s) edge [in=140,out=110,loop]      node[above]                        {\scriptsize $=$} (s)
 (s) edge [bend right]               node[left]                         {\scriptsize $>$} (d)
 (s) edge [bend left]                node[left]                         {\scriptsize $<$} (a)
 (a) edge                            node[above]                        {\scriptsize $>$} (b)
 (a) edge                            node[above=-2.5pt,sloped,pos=0.45] {\scriptsize $=$} (s)
 (a) edge [in=140,out=110,loop]      node[above]                        {\scriptsize $<$} (a)
 (b) edge                            node[below=-2pt,sloped,pos=0.2]    {\scriptsize $>$} (d)
 (b) edge [bend angle=8,bend left]   node[below,sloped,pos=0.6]         {\scriptsize $=$} (s)
 (b) edge [densely dashed]           node[above]                        {\scriptsize $<$} (c)
 (c) edge [densely dashed,bend angle=20,bend left] node                 {\scriptsize $>$} (f)
 (c) edge [densely dotted]           node[above]                        {\scriptsize $=$} (cc)
 (c) edge [densely dotted,bend angle=46,bend right] node[above]         {\scriptsize $<$} (a)
 (d) edge [in=230,out=200,loop]      node[below]                        {\scriptsize $>$} (d)
 (d) edge                            node[below=-2pt,sloped,pos=0.45]   {\scriptsize $=$} (s)
 (d) edge                            node[below]                        {\scriptsize $<$} (e)
 (e) edge [densely dashed]           node[below]                        {\scriptsize $>$} (f)
 (e) edge [bend angle=8,bend right]  node[above,sloped,pos=0.6]         {\scriptsize $=$} (s)
 (e) edge                            node[above=-2pt,sloped,pos=0.2]    {\scriptsize $<$} (a)
 (f) edge [densely dotted,bend angle=46,bend left]  node[below]         {\scriptsize $>$} (d)
 (f) edge [densely dotted]           node[below]                        {\scriptsize $=$} (ff)
 (f) edge [densely dashed,bend angle=20,bend left] node                 {\scriptsize $<$} (c);
\end{scope}
\begin{scope}[->,>=stealth',shorten >=1pt,auto,node distance=12mm,semithick,xshift=-0.33cm,xshift=8cm,yshift=2.2cm]
\node (b) at (-5.3,-0.8) {\footnotesize (B)};
\fill[rounded corners=2pt,gray!20] (-0.6,-4.8) rectangle (3.3,-1.4);
\node (tr) at (0.3,-1.6) {\bf \scriptsize transitions:};
\node (res) at (1.3,-0.8) {\scriptsize $\mathit{end}[\seqlength-1]=\seqlength+1-\mathit{in}[\seqlength-1]$};
\node (o) at (0.4,-4.6) {\scriptsize $\circ\in\{<,=,>\}$};
\draw[->] (-0.5,-2.5) -- (3.2,-2.5) node[pos=.5] {\scriptsize $\begin{array}{c}x_{k} \circ x_{k+1}\\ \mathit{end}[k-1]=\mathit{end}[k], \mathit{in}[k]=0 \end{array}$};
\draw[densely dashed,->] (-0.5,-3.4) -- (3.2,-3.4) node[pos=.5] {\scriptsize $\begin{array}{c}x_k \circ x_{k+1}\\ \mathit{end}[k-1]=\mathit{end}[k], \mathit{in}[k]=1 \end{array}$};
\draw[densely dotted,->] (-0.5,-4.3) -- (3.2,-4.3) node[pos=.5] {\scriptsize $\begin{array}{c}x_k \circ x_{k+1}\\ \mathit{end}[k-1]=k, \mathit{in}[k]=0 \end{array}$};
\end{scope}
\draw[densely dashed] (2,-2.6) -- (2,1.8);
\end{tikzpicture}
\caption{\label{fig:automata_next}
Register automata computing the end of the next pattern maximal occurrence for
(A)~the $\PlainPatternName$ and (B)~the $\ZigzagPatternName$ patterns}
\end{figure}
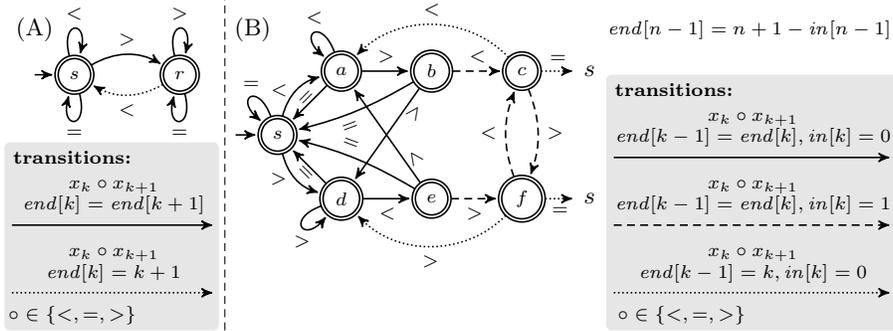
\end{itemize}

{
\begin{table}[!h]
\begin{subtable}{.5\linewidth}
\centering
\begin{tabular}{lrrrrrrrrrrr}
\toprule
$x_k$               & $0$ & $1$ & $0$ & $1$ & $0$ & $0$ & $1$ & $0$ &  $1$ &  $2$ &  $0$ \\
$s_k$               & $<$ & $>$ & $<$ & $>$ & $=$ & $<$ & $>$ & $<$ &  $<$ &  $>$ &  $<$ \\
$x_{k+1}$           & $1$ & $0$ & $1$ & $0$ & $0$ & $1$ & $0$ & $1$ &  $2$ &  $0$ &  $1$ \\
\midrule
$k+1$               & $2$ & $3$ & $4$ & $5$ & $6$ & $7$ & $8$ &  $9$ & $10$ & $11$ & $12$ \\
$\mathit{end}[k]$   & $4$ & $4$ & $4$ & $7$ & $7$ & $7$ & $9$ &  $9$ & $12$ & $12$ & $12$ \\
$\mathit{end}[k+1]$ & $4$ & $4$ & $7$ & $7$ & $7$ & $9$ & $9$ & $12$ & $12$ & $12$ & $13$ \\
\bottomrule
\end{tabular}
\caption*{(A1)}
\end{subtable}
\begin{subtable}{.5\linewidth}
\centering
\begin{tabular}{lrrrrrrrrrrr}
\toprule
$x_k$               &  $1$ &  $0$ & $2$ & $1$ & $0$ & $1$ & $0$ & $0$ & $1$ & $0$ & $1$ \\
$s_k$               &  $>$ &  $<$ & $>$ & $>$ & $<$ & $>$ & $=$ & $<$ & $>$ & $<$ & $>$ \\
$x_{k+1}$           &  $0$ &  $2$ & $1$ & $0$ & $1$ & $0$ & $0$ & $1$ & $0$ & $1$ & $0$ \\
\midrule
$k+1$               & $11$ & $10$ & $9$ & $8$ & $7$ & $6$ & $5$ & $4$ & $3$ & $2$ & $1$ \\
$\mathit{end}[k]$   & $10$ & $10$ & $7$ & $7$ & $7$ & $4$ & $4$ & $4$ & $2$ & $2$ & $0$ \\
$\mathit{end}[k+1]$ & $10$ &  $7$ & $7$ & $7$ & $4$ & $4$ & $4$ & $2$ & $2$ & $0$ & $0$ \\
\bottomrule
\end{tabular}
\caption*{(A2)}
\end{subtable}
\caption{\label{table:example_running_end_plain_automata}Running the register automaton of Figure~\ref{fig:automata_next} that computes the end of the next plain on (A1)~the sequence $x=0 1 0 1 0 0 1 0 1  2  0  1$ and (A2)~on its reverse}
\end{table}
}

{
\begin{table}[!h]
\begin{subtable}{.5\linewidth}
\centering
\begin{tabular}{lrrrrrrrrrrr}
\toprule
$x_k$               & $0$ & $1$ & $0$ & $1$ & $0$ & $0$ & $1$ & $0$ &  $1$ &  $2$ &  $0$ \\
$s_k$               & $<$ & $>$ & $<$ & $>$ & $=$ & $<$ & $>$ & $<$ &  $<$ &  $>$ &  $<$ \\
$x_{k+1}$           & $1$ & $0$ & $1$ & $0$ & $0$ & $1$ & $0$ & $1$ &  $2$ &  $0$ &  $1$ \\
\midrule
$k$                 & $1$ & $2$ & $3$ & $4$ & $5$ & $6$ & $7$ & $8$ &  $9$ & $10$ & $11$ \\
$\mathit{end}[k-1]$ & $5$ & $5$ & $5$ & $5$ & $5$ & $9$ & $9$ & $9$ &  $9$ & $12$ & $12$ \\
$\mathit{end}[k]$   & $5$ & $5$ & $5$ & $5$ & $9$ & $9$ & $9$ & $9$ & $12$ & $12$ & $12$ \\
$\mathit{in}[k]$    & $0$ & $0$ & $1$ & $1$ & $0$ & $0$ & $0$ & $1$ &  $0$ &  $0$ &  $1$ \\
\bottomrule
\end{tabular}
\caption*{(B1)}
\end{subtable}
\begin{subtable}{.5\linewidth}
\centering
\begin{tabular}{lrrrrrrrrrrr}
\toprule
$x_k$               &  $1$ &  $0$ &  $2$ & $1$ & $0$ & $1$ & $0$ & $0$ & $1$ & $0$ & $1$ \\
$s_k$               &  $>$ &  $<$ &  $>$ & $>$ & $<$ & $>$ & $=$ & $<$ & $>$ & $<$ & $>$ \\
$x_{k+1}$           &  $0$ &  $2$ &  $1$ & $0$ & $1$ & $0$ & $0$ & $1$ & $0$ & $1$ & $0$ \\
\midrule
$k$                 & $12$ & $11$ & $10$ & $9$ & $8$ & $7$ & $6$ & $5$ & $4$ & $3$ & $2$ \\
$\mathit{end}[k-1]$ &  $9$ &  $9$ &  $9$ & $9$ & $6$ & $6$ & $6$ & $1$ & $1$ & $1$ & $1$ \\
$\mathit{end}[k]$   &  $9$ &  $9$ &  $9$ & $6$ & $6$ & $6$ & $1$ & $1$ & $1$ & $1$ & $1$ \\
$\mathit{in}[k]$    &  $0$ &  $0$ &  $1$ & $0$ & $0$ & $1$ & $0$ & $0$ & $0$ & $1$ & $1$ \\
\bottomrule
\end{tabular}
\caption*{(B2)}
\end{subtable}
\caption{\label{table:example_running_end_zigzag_automata}Running the register automaton of Figure~\ref{fig:automata_next} that computes the end of the next zigzag on (B1)~the sequence $x=0 1 0 1 0 0 1 0 1  2  0  1$ and (B2)~on its reverse}
\end{table}
}

\section{Conclusion}

Based on a detailed analysis of feature and pattern properties
of time\nobreakdash-series constraints of the time-series catalogue that use the $\SumAggr$ aggregator,
we came up with a $\Theta(\seqlength)$ time complexity checker,
and a $\Theta(\seqlength)$ space complexity reformulation for such constraints.
It is an open question how to generalise our results to
other aggregators such as $\min$ or $\max$.
Unlike the sum aggregator,
the equality $\agg(a,x)=b$ where $a,b$ are fixed integers and $x$ is a variable does not uniquely determine $x$ when $\agg\in\{\min,\max\}$.\\

\noindent{\footnotesize\textbf{Acknowledgment} We thank Pierre Flener for some feedback on an early version of this paper, and Colin de la Higuera for discussions on regular expressions, on the properties of their languages and on operators such as $\shuffle$.}

\bibliographystyle{splncs03}

\bibliography{paper}

\clearpage
\appendix

\section{List of Feasible Types with Corresponding Witnesses}
\label{sec:witness_feasible_types}

\begin{multicols}{2}
\noindent Table~\ref{table:feasible_types_and_witnesses} provides for each of the $61$ feasible type
$\langle t_1,t_2,t_3\rangle$ that occurs in the map shown by Figure~\ref{fig:triples_map},
a regular expression for which the language defined by Theorem~\ref{theorem:representative_language}
is not empty. For instance, the type $\langle\csuf,\cfac,\cpre\rangle$ can be obtained 
from the following regular expression $\reg{<<<=<<<|<<=<|<=<<|=}$ as illustrated by the figure below.

\begin{tikzpicture}[scale=1]
\node[anchor=west] (dummy)  at (-0.8, 0) {~};
\node[anchor=west] (word)  at (1.6, 0.03) {\scriptsize word   in $1..k$};
\node[anchor=west] (suf)  at (1.6, -0.33) {\scriptsize suffix in $1..j$};
\node[anchor=west] (fac)  at (1.6, -0.66) {\scriptsize factor in $i..j$};
\node[anchor=west] (pre)  at (1.6, -0.99) {\scriptsize prefix in $i..k$};
\node[anchor=west] (w)  at (0, 0.00) {\tiny${\color{black}<<<=<<<}$};
\node[anchor=west] (p2) at (0.2,-0.33) {\tiny{\color{black}$<<=<$}};
\node[anchor=west] (w2) at (0.58,-0.66) {\tiny{\color{black}$=$}};
\node[anchor=west] (s1) at (0.39,-0.99) {\tiny{\color{black}$<=<<$}};
\draw[densely dotted] (0.12,0.1) -- (0.12,-1.1);
\draw[densely dotted] (0.52,0.1) -- (0.52,-1.1);
\draw[densely dotted] (1.07,0.1) -- (1.07,-1.1);
\draw[densely dotted] (1.45,0.1) -- (1.45,-1.1);
\node[anchor=north] (1) at (0.12,-1.1) {\tiny$1$};
\node[anchor=north] (i) at (0.52,-1.1) {\tiny$i$};
\node[anchor=north] (j) at (1.07,-1.1) {\tiny$j$};
\node[anchor=north] (k) at (1.45,-1.1) {\tiny$k$};
\end{tikzpicture}
\end{multicols}

{
\setlength{\tabcolsep}{4pt}
\begin{table}[!h]
\centering
\noindent\makebox[0.98\textwidth]{
\begin{tabular}{ll|ll}\toprule
\footnotesize Triple  & \footnotesize Witness & \footnotesize Triple  & \footnotesize Witness \\\midrule

$\langle\cfac,\cfac,\cfac\rangle$ & \scriptsize$<<=<<|=$ &
$\langle\cout,\cout,\cpre\rangle$ & \scriptsize$<<=<|<=$ \\

$\langle\cfac,\cfac,\cin \rangle$ & \scriptsize$<=<<|=$ &
$\langle\cout,\cout,\csuf\rangle$ & \scriptsize$<<=|=$  \\

$\langle\cfac,\cfac,\cpre\rangle$ & \scriptsize$<<=<<<|<=<<|=<<|=$ &
$\langle\cpre,\cfac,\cfac\rangle$ & \scriptsize$<=<==|<$           \\

$\langle\cfac,\cfac,\csuf\rangle$ & \scriptsize$<<=<=|=$      &
$\langle\cpre,\cfac,\cin \rangle$ & \scriptsize$<=<=<|=<=<|<$ \\

$\langle\cfac,\cout,\cfac\rangle$ & \scriptsize$<=<=<|=$           &
$\langle\cpre,\cfac,\cpre\rangle$ & \scriptsize$<<=<<<|<=<<|<<=|=$ \\

$\langle\cfac,\cout,\cin \rangle$ & \scriptsize$<<=<=|<=$ &
$\langle\cpre,\cfac,\csuf\rangle$ & \scriptsize$<=<=<|<$  \\

$\langle\cfac,\cout,\cout\rangle$ & \scriptsize$<=<<|=$ &
$\langle\cpre,\cout,\cfac\rangle$ & \scriptsize$<=<=|<$ \\

$\langle\cfac,\cout,\cpre\rangle$ & \scriptsize$<=<<>|<<|=$ &
$\langle\cpre,\cout,\cin \rangle$ & \scriptsize$<=<=|<=$    \\

$\langle\cfac,\cout,\csuf\rangle$ & \scriptsize$<=<=|=$ &
$\langle\cpre,\cout,\cout\rangle$ & \scriptsize$<==|<$  \\

$\langle\cfac,\cpre,\cin \rangle$ & \scriptsize$<=<=|=$   &
$\langle\cpre,\cout,\cpre\rangle$ & \scriptsize$<=<=<|<=$ \\

$\langle\cfac,\cpre,\cpre\rangle$ & \scriptsize$<=<<|=$ &
$\langle\cpre,\cout,\csuf\rangle$ & \scriptsize$<=<|<$  \\

$\langle\cin, \cfac,\cfac\rangle$ & \scriptsize$<<=<|=$ &
$\langle\cpre,\cpre,\cin \rangle$ & \scriptsize$<=<|<$  \\

$\langle\cin, \cfac,\cin \rangle$ & \scriptsize$<<=<<|<<=<|<=<<|=$ &
$\langle\cpre,\cpre,\cpre\rangle$ & \scriptsize$<<==|<$ \\

$\langle\cin, \cfac,\cpre\rangle$ & \scriptsize$<<=<<<|<<=<|<=<<|=$ &
$\langle\csuf,\cfac,\cfac\rangle$ & \scriptsize$<<<=<<|<<=<|<<=|=$  \\

$\langle\cin, \cfac,\csuf\rangle$ & \scriptsize$<<=<=|<<=<|=$       &
$\langle\csuf,\cfac,\cin \rangle$ & \scriptsize$<<<=<<|<<=<|<=<<|=$ \\

$\langle\cin, \cin, \cin \rangle$ & \scriptsize$<<|<$                &
$\langle\csuf,\cfac,\cpre\rangle$ & \scriptsize$<<<=<<<|<<=<|<=<<|=$ \\

$\langle\cin, \cin, \cpre\rangle$ & \scriptsize$<<=|<$             &
$\langle\csuf,\cfac,\csuf\rangle$ & \scriptsize$<<<=<<|<<=<|=<<|=$ \\

$\langle\cin, \cout,\cfac\rangle$ & \scriptsize$<=<=<|<=$ &
$\langle\csuf,\cin, \cin \rangle$ & \scriptsize$<==|=$    \\

$\langle\cin, \cout,\cin \rangle$ & \scriptsize$<<<|<<$ &
$\langle\csuf,\cin, \cpre\rangle$ & \scriptsize$<=<|=$  \\

$\langle\cin, \cout,\cout\rangle$ & \scriptsize$<=|<$       &
$\langle\csuf,\cout,\cfac\rangle$ & \scriptsize$<<=><|<=|>$ \\

$\langle\cin, \cout,\cpre\rangle$ & \scriptsize$<<<=|<<$ &
$\langle\csuf,\cout,\cin \rangle$ & \scriptsize$<===|==$ \\

$\langle\cin, \cout,\csuf\rangle$ & \scriptsize$<=<=|<=$ &
$\langle\csuf,\cout,\cout\rangle$ & \scriptsize$<<=<|<=$ \\

$\langle\cin, \cpre,\cin \rangle$ & \scriptsize$<<=<|<<=|<$ &
$\langle\csuf,\cout,\cpre\rangle$ & \scriptsize$<===<|==$   \\

$\langle\cin, \cpre,\cpre\rangle$ & \scriptsize$<<=|<$    &
$\langle\csuf,\cout,\csuf\rangle$ & \scriptsize$<<=<=|<=$ \\

$\langle\cin, \csuf,\cfac\rangle$ & \scriptsize$<=<=|<$      &
$\langle\csuf,\cpre,\cin \rangle$ & \scriptsize$<<=<=|<=<|=$ \\

$\langle\cin, \csuf,\cin \rangle$ & \scriptsize$<=<<|=<<|<$  &
$\langle\csuf,\cpre,\cpre\rangle$ & \scriptsize$<<=<<|<=<|=$ \\

$\langle\cin, \csuf,\cpre\rangle$ & \scriptsize$<=<<=|=<<|<$ &
$\langle\csuf,\csuf,\cfac\rangle$ & \scriptsize$<<=<|=$      \\

$\langle\cin, \csuf,\csuf\rangle$ & \scriptsize$<<=|=$ &
$\langle\csuf,\csuf,\cin \rangle$ & \scriptsize$<=<|=$ \\

$\langle\cout,\cout,\cfac\rangle$ & \scriptsize$<<=<|=$      &
$\langle\csuf,\csuf,\cpre\rangle$ & \scriptsize$<<=<<|<=<|=$ \\

$\langle\cout,\cout,\cin \rangle$ & \scriptsize$<=|=$   &
$\langle\csuf,\csuf,\csuf\rangle$ & \scriptsize$<<==|=$ \\

$\langle\cout,\cout,\cout\rangle$ & \scriptsize$<<<<|<<<$ \\

\bottomrule
\end{tabular}
}
\caption{\label{table:feasible_types_and_witnesses}{\normalsize List of feasible types and associated regular expressions witnesses}
}
\end{table}
}

\clearpage
\section{Counterexamples for Equations \eqref{formula1}, \eqref{formula2} and~\eqref{formula3}}
\label{sec:counterexamples}

For each time-series constraint of the time-series constraint catalogue
this appendix provides small time series corresponding to counterexamples of the validity
of Equations \eqref{formula1}, \eqref{formula2} and~\eqref{formula3}
for all equations missing in Table~\ref{table:summary}.
For instance, for \constraint{nb\_decreasing\_sequence} and
Equation~\eqref{formula1}, we get the following counterexample:
consider the three windows of size $2$ wrt the sequence $\langle1,0,0,-1\rangle$;
using Equation~\eqref{formula1} returns $\langle1,{\color{red}1},1\rangle$
rather than the expected values $\langle1,0,1\rangle$,
i.e.~on the second subsequence ``$0,0$'', \eqref{formula1} returns $1+1-1={\color{red}1}$
rather than the expected value $0$;
Value $0$ reflects the fact that subsequence ``$0,0$'' does not contain
any decreasing sequence.\\ \\ \\

{\tiny
\begin{tabular}{lccc} \toprule constraint & \eqref{formula1} & \eqref{formula2} & \eqref{formula3} \\ \midrule
\constraint{nb\_decreasing\_sequence}& {\tiny$\begin{array}{c} 2,\langle1,0,0,-1\rangle,\\\langle1,0,1\rangle,\langle1,{\color{red}1},1\rangle\end{array}$} 
& {\tiny$\begin{array}{c} 2,\langle1,0,0,-1\rangle,\\\langle1,0,1\rangle,\langle1,{\color{red}1},1\rangle\end{array}$} 
& - \\
\constraint{sum\_width\_decreasing\_sequence}& {\tiny$\begin{array}{c} 2,\langle1,0,0,0,-1\rangle,\\\langle2,0,0,2\rangle,\langle2,{\color{red}-1},{\color{red}-1},2\rangle\end{array}$} 
& -  
& - \\
\constraint{sum\_surf\_decreasing\_sequence}& {\tiny$\begin{array}{c} 2,\langle-1,-2,-2,-2,-3,2\rangle,\\\langle-3,0,0,-5,0\rangle,\langle-3,{\color{red}2},{\color{red}2},-5,0\rangle\end{array}$} 
& {\tiny$\begin{array}{c} 2,\langle-1,0,-1\rangle,\\\langle0,-1\rangle,\langle0,{\color{red}0}\rangle\end{array}$} 
& - \\
\constraint{sum\_max\_decreasing\_sequence}& {\tiny$\begin{array}{c} 2,\langle1,-1,-1,-2\rangle,\\\langle1,0,-1\rangle,\langle1,{\color{red}-1},-1\rangle\end{array}$} 
& {\tiny$\begin{array}{c} 2,\langle-2,-3,2\rangle,\\\langle-2,0\rangle,\langle{\color{red}0},0\rangle\end{array}$} 
& - \\
\constraint{sum\_min\_decreasing\_sequence}& {\tiny$\begin{array}{c} 2,\langle1,-1,-1,-2\rangle,\\\langle-1,0,-2\rangle,\langle-1,{\color{red}-1},-2\rangle\end{array}$} 
& {\tiny$\begin{array}{c} 2,\langle-1,0,-1\rangle,\\\langle0,-1\rangle,\langle0,{\color{red}0}\rangle\end{array}$} 
& - \\
\constraint{nb\_decreasing\_terrace}& {\tiny$\begin{array}{c} 2,\langle1,0,0,-1\rangle,\\\langle0,0,0\rangle,\langle0,{\color{red}-1},0\rangle\end{array}$} 
& -  
& - \\
\constraint{sum\_width\_decreasing\_terrace}& {\tiny$\begin{array}{c} 2,\langle1,0,0,-1\rangle,\\\langle0,0,0\rangle,\langle0,{\color{red}-2},0\rangle\end{array}$} 
& -  
& - \\
\bottomrule\end{tabular}}

{\tiny\begin{tabular}{lccc} \toprule constraint & \eqref{formula1} & \eqref{formula2} & \eqref{formula3} \\ \midrule
\constraint{nb\_plain}& {\tiny$\begin{array}{c} 2,\langle0,-2,-2,1\rangle,\\\langle0,0,0\rangle,\langle0,{\color{red}-1},0\rangle\end{array}$} 
& -  
& - \\
\constraint{sum\_width\_plain}& {\tiny$\begin{array}{c} 2,\langle0,-2,-2,1\rangle,\\\langle0,0,0\rangle,\langle0,{\color{red}-2},0\rangle\end{array}$} 
& -  
& - \\
\constraint{sum\_surf\_plain}& {\tiny$\begin{array}{c} 2,\langle0,-2,-2,1\rangle,\\\langle0,0,0\rangle,\langle0,{\color{red}4},0\rangle\end{array}$} 
& {\tiny$\begin{array}{c} 3,\langle0,-1,0\rangle,\\\langle-1\rangle,\langle{\color{red}0}\rangle\end{array}$} 
& - \\
\constraint{sum\_height\_plain}& {\tiny$\begin{array}{c} 2,\langle0,-2,-2,1\rangle,\\\langle0,0,0\rangle,\langle0,{\color{red}2},0\rangle\end{array}$} 
& {\tiny$\begin{array}{c} 3,\langle0,-1,0\rangle,\\\langle-1\rangle,\langle{\color{red}0}\rangle\end{array}$} 
& - \\
\constraint{nb\_plateau}& {\tiny$\begin{array}{c} 2,\langle-1,1,1,0\rangle,\\\langle0,0,0\rangle,\langle0,{\color{red}-1},0\rangle\end{array}$} 
& -  
& - \\
\constraint{sum\_width\_plateau}& {\tiny$\begin{array}{c} 2,\langle-1,1,1,0\rangle,\\\langle0,0,0\rangle,\langle0,{\color{red}-2},0\rangle\end{array}$} 
& -  
& - \\
\constraint{sum\_surf\_plateau}& {\tiny$\begin{array}{c} 2,\langle-1,1,1,0\rangle,\\\langle0,0,0\rangle,\langle0,{\color{red}-2},0\rangle\end{array}$} 
& {\tiny$\begin{array}{c} 3,\langle-5,-4,-5,4\rangle,\\\langle-4,0\rangle,\langle{\color{red}0},0\rangle\end{array}$} 
& - \\
\bottomrule\end{tabular}}

\clearpage
\section{Evaluating Pattern Properties}
\label{sec:evaluating_pattern_properties}

This appendix provides the program that computes all the representatives of a pattern
and the program that evaluates the properties of a pattern.
Both programs \emph{(i)}~convert regular expression formulas of this paper
in a sequence of operations on finite automata, and
\emph{(ii)}~check that the final automaton contains or not an accepting state.

{\scriptsize
%
}

\end{document}